\documentclass[11pt,a4paper]{article}
\usepackage[T1]{fontenc}
\usepackage[utf8]{inputenc}
\usepackage{amsmath,amssymb,amsfonts,amsthm}
\usepackage{booktabs}
\usepackage{array}
\usepackage{rotating}
\usepackage[a4paper,margin=2.6cm]{geometry}
\usepackage[colorlinks,citecolor=blue,urlcolor=blue,linkcolor=blue,breaklinks,bookmarks]{hyperref}
\hypersetup{pdftitle={Proofs Without Nominals: G\"odel's Ontological Argument, its Shallow Embedding, and the Open Questions of the Monatshefte Notes},pdfauthor={Christoph Benzm\"uller}}

\usepackage[numbers,square,sort,elide]{natbib}

\newtheorem{theorem}{Theorem}
\newtheorem{lemma}[theorem]{Lemma}
\newtheorem{proposition}[theorem]{Proposition}
\newtheorem{corollary}[theorem]{Corollary}
\theoremstyle{definition}
\newtheorem{definition}{Definition}
\newtheorem{example}[theorem]{Example}
\theoremstyle{remark}

\newcommand{\lean}{\texorpdfstring{\textsf{Lean\,4}}{Lean 4}}
\newcommand{\isa}{\texorpdfstring{\textsf{Isabelle/HOL}}{Isabelle/HOL}}
\newcommand{\tc}[1]{\texttt{#1}}
\newcommand{\val}[1]{\lfloor #1\rfloor}
\newcommand{\Pos}{\mathcal{P}}
\newcommand{\Ess}{\mathrel{\mathit{Ess}}}
\newcommand{\NE}{\mathit{NE}}
\newcommand{\actE}{^{E}}

\newcommand{\NFig}[2][]{\cite[Fig.~#2#1]{J75}}

\title{Proofs Without Nominals: G\"odel's Ontological Argument, its Shallow Embedding, and the Open Questions of the \emph{Monatshefte} Notes}
\author{Christoph Benzm\"uller\\[2pt]
  \small AI Systems Engineering, Otto-Friedrich-Universit\"at Bamberg, Germany\\ \small Faculty of Mathematics and Computer Science, Freie Universit\"at Berlin, Germany\\
  \small \texttt{christoph.benzmueller@uni-bamberg.de}}
\date{}

\makeatletter
\newcommand{\appendixheadings}{%
  \let\venue@seccntformat\@seccntformat
  \def\@seccntformat##1{\ifnum\pdfstrcmp{##1}{section}=0 \appendixname~\thesection\quad%
    \else\venue@seccntformat{##1}\fi}}
\makeatother

\begin{document}
\maketitle

\begin{abstract}
The shallow embedding of higher-order modal logic in classical higher-order logic, used in
Benzm\"uller and Scott's \emph{Notes on G\"odel's and Scott's variants of the ontological
argument}, reaches beyond the modal object language: its
quantifiers range over domains containing the nominals and satisfaction operators of
hybrid logic.  A proof using one proves a theorem of the embedding, not necessarily of
the modal logic.  Whether a result is one of the modal logic can be settled by replaying it in
an explicit calculus, or, more cheaply and for all results at once, by analysing what the
embedding's proofs substitute.  Every statement the \emph{Notes} prove has a
proof in the embedding substituting nothing outside the object language, 294 of them
explicit.  The \emph{Notes}' own proofs substitute none either; what it flags are terms a prover introduced.  A soundness theorem over general models closed under the object language,
and a sentence of that language the embedding proves but which fails in such a model, fix what
the certificate says.

\looseness=-1 The \emph{Notes}' three open questions are settled too, without nominals, but by
the empty instance of the conjunction axiom, too cheaply to convince.  Generalised to
G\"odel's ``any number of summands'', the axiom covers the empty conjunction and the singleton.  This article restricts it to at least two different conjuncts, a
non-degenerate reading of G\"odel's footnote, and settles them again without nominals.  In effect only the empty conjunction is excluded; the singleton comes back with the
other axioms.  In one variant the restriction refutes the possible existence of a God-like being; the answers
return, with no side condition, once the \emph{Notes}' non-emptiness clause is confined to
G\"odel's essence.  Every theorem here is verified in \isa{} and independently
in \lean{}, the separating model in its propositional part; the countermodels are Nitpick's,
checked by the build.
\end{abstract}


\allowdisplaybreaks

\section{Introduction}

In~\cite{J75}, Benzm\"uller and Scott revisit Kurt G\"odel's 1970 sketch of a modal
ontological argument~\cite{GoedelNotes}, and Dana Scott's variant of
it~\cite{ScottNotes}, with \isa{}~\cite{IsabelleHOL}, a proof
assistant for classical higher-order logic (HOL).  The higher-order modal logic (HOML) of the arguments is
shallowly embedded in HOL~\cite{J48,ECAI2014}: a modal proposition is a predicate on worlds, a
modal property a function from individuals to such predicates, and the connectives are defined
on these.  The embedding is faithful with respect to Henkin semantics in the sense of~\cite{J7}: it is sound and complete
for the quantified modal logics whose models' propositional domains are those of a Henkin model
of HOL~\cite{R45,J26}.\footnote{That clause is what the construction requires, not a gloss on it,
  and~\cite{J26} does not state it with its theorems; see Section~\ref{sec:discussion}.}  It has
been the workhorse of computational metaphysics for over a
decade~\cite{J41,ECAI2014,IJCAI2016,J50}.  It also
reaches further than the modal object language, deliberately so.  Its quantifiers over
propositions and properties range over the corresponding domains of a model of HOL, and these
contain ``being at world $w$'',
``being $x$ at $w$'' and ``being at a successor of $w$''.  These are terms of HOL that the modal
object language of the arguments cannot express: a nominal of hybrid
logic~\cite{Blackburn2000,ArecesTenCate2007}, a formula containing one, a
backward-looking modality applied to one.  That is no coincidence.  The embedding descends from
Brown's encoding of hybrid logic in HOL~\cite{Brown2005,R45}.  Benzm\"uller and Paulson, who
introduced it for quantified multimodal logics~\cite{R45,J26}, gave nominals and the satisfaction
operator as $\lambda$-terms among the further operators the embedding admits, and dropped an
earlier syntactic restriction that had kept world-dependent subterms out~\cite{B9}, on the
grounds that soundness and completeness do not need it.  Wisniewski and
Steen~\cite{WisniewskiSteen2014}
treated the quantified nominal case explicitly, and a proposed TPTP syntax for non-classical
logics~\cite{W56} includes hybrid syntax.  A proof that instantiates a property variable with a
nominal is therefore correct in the embedding, and what it proves holds in every model of HOL.  But it need not
correspond to any derivation in the modal logic of the argument, and what it establishes need not
hold in the models the object logic's own semantics admits --- those whose property domains are
closed under its operations alone (Section~\ref{sec:discussion}).  What the absence of that
restriction permits is the subject of this article.

Whether a given result is a result of the modal logic can be answered from outside the
embedding, by deriving it in an explicit calculus for that logic; that is one route.  Kanckos and Woltzenlogel
Paleo~\cite{KanckosPaleo2017}, for instance, give four such proofs by hand, in a natural
deduction calculus for a rigid higher-order modal logic: two of Scott's variant of
G\"odel's ontological argument, in \textsf{S5} and in \textsf{KB}, and one each of
Anderson's and Bj\o{}rdal's.
The other is to analyse the proofs the embedding itself produces, the one this
article takes; checking the \emph{Notes} yields those proofs anyway, so one mechanical pass
over them covers every result at once.

What that route certifies should be stated exactly.  \emph{Hybrid-free} in this article means
a syntactic property of a proof (Definition~\ref{def:hybridfree}): every term it substitutes for a variable of
one of the types the axioms quantify over beyond the individuals is a term of the
modal object language, and Proposition~\ref{prop:sound} says
that the detector of Section~\ref{sec:criterion}, which implements the criterion, accepts only
such proofs.  The proposition is about proofs, not models.  The models whose domains are closed
under the operations of the object language alone are defined in Section~\ref{sec:criterion-open},
where a soundness theorem shows that a proof meeting a stronger, whole-term condition establishes
its theorem in every one of them, and an example shows a sentence of the object language that
the embedding proves and no such proof proves; for the theorems of the \emph{Notes} and of this
article the two notions are not known to come apart.  The certificate is syntactic, and
Section~\ref{sec:discussion} says what it does and does not establish.

The problem is not peculiar to this argument, but it is peculiar to higher-order object
logics.  Wherever a higher-order object logic is represented by a shallow embedding in a richer
host, the host may admit witnesses the object language has no term for, and whether a given
formal proof has used one is not visible in its statement.  Where the object quantifiers range
over individuals only, as in propositional and first-order modal logics, the embedding is
complete for the object language and nothing is lost (Section~\ref{sec:discussion}).  Beyond settling the three questions the \emph{Notes} leave open, this article makes
four contributions:
\begin{enumerate}
\item a syntactic criterion on proof terms for detecting witnesses outside the object language,
  with the soundness it has and the completeness it lacks (Section~\ref{sec:criterion},
  Proposition~\ref{prop:sound}), a soundness theorem over general models closed under the
  object language, and a sentence of that language, provable in the
  embedding, that no proof meeting the theorem's hypothesis proves, with a model of the axioms of
  variant~2 in which it fails; whether the two notions come apart for variant~3 stays open
  (Sections~\ref{sec:criterion-open} and~\ref{sec:open});
\item an audit of every fact of two formal developments under that criterion --- the 427 of
  the \emph{Notes} and the 399 of the port~\cite{Port} (Section~\ref{sec:audit});
\item an analysis of the strength of the generalised conjunction axiom \tc{Ax1Gen}, which
  separates two sources --- two degenerate instances and the world-dependence of $\Phi$ --- and
  says what each is answerable for (Sections~\ref{sec:gex} and~\ref{sec:emendation}); and
\item an emendation of \tc{Ax1Gen}, namely \tc{Ax1GenTwo}, which requires a conjunction to have at least
  two different conjuncts and so licenses neither degenerate instance.  One of the two comes back
  all the same: a singleton padded with the universal property has two members and the same
  conjunction, and such a set is the witness the proofs use.  The empty conjunction does not
  come back, and the answers do not need it.  In one of the three variants \tc{Ax1GenTwo} costs
  the possible existence of a God-like being (\tc{Th4}); that returns once the \emph{Notes}'
  non-emptiness clause is confined to G\"odel's essence (Sections~\ref{sec:ax1gentwo}
  and~\ref{sec:plain}).  Four neighbouring readings are weighed beside \tc{Ax1GenTwo} against
  three requirements, and \tc{Ax1GenTwo} is the one the article recommends: no other reading
  keeps the answers to all three questions while excluding the two degenerate instances
  (Section~\ref{sec:assessment}, Table~\ref{tab:verdicts}).
\end{enumerate}
Two facts frame the answers and are worth knowing from the start: under the \emph{Notes}'
conjunction axiom all three answers are derivable already in \textsf{K}, through the empty
instance; and in variant~3 the definitions alone force, with a nominal, a single world and with
it the modal collapse (Section~\ref{sec:v3}).
Every result this article establishes is proved in both \isa{} and \lean{}; the countermodels are
Nitpick's, certified by the build (Section~\ref{sec:verification}).  What exists in one system
only is an auxiliary step of a proof or a model search, never a result
(Section~\ref{sec:contributed}).  The \emph{Notes}' own theorems are not proved again here but
read: their \isa{} sources and their \lean{} port are what the audit of
Section~\ref{sec:audit} runs over.
The article proceeds in two steps: it answers the three open questions, and it asks of every
proof the \emph{Notes} give whether its witnesses stay inside the object language.

First, the three questions, all of which concern modal logic \textsf{S4}, one of them
\textsf{S5} as well:
\begin{enumerate}
\item For G\"odel's argument with the definition of essence repaired
  (variant~2 of the \emph{Notes}, \NFig{7}), theorem \tc{Th3} --- if a God-like being possibly
  exists then one necessarily exists: ``no such proof has yet been reported, nor a counterexample found in
  logic S4 \dots\ This question thus remains open''~\cite[p.~591]{J75}.
\item For G\"odel's argument with the notion of necessary property inclusion repaired
  instead (variant~3, \NFig{8}), the same theorem \tc{Th3} in \textsf{S4}, and the two
  lemmas on the uniqueness of essences, \tc{UniqueEss1} and \tc{UniqueEss2}, open in
  \textsf{S5} already: ``neither
  did the model finder report a counterexample nor did the automated theorem provers come
  up with a proof \dots\ (though we expect them to be provable)''~\cite[Sect.~4.5]{J75}.
\item For Scott's variant with Scott's axiom \tc{A3}, the positivity of being God-like,
  replaced by G\"odel's conjunction axiom in its generalised form \tc{Ax1Gen}: the \textsf{S4}
  countermodels Nitpick finds for Scott's necessary-existence theorem \tc{T3} with \tc{A3}
  ``could not yet be reproduced \dots, but also
  no proof could be found.  This inconclusive situation thus requires further
  study''~\cite[Sect.~5.3]{J75}.
\end{enumerate}
The three questions are settled, affirmatively and hybrid-free, but not on the first
attempt,\footnote{The question of this article arose in the discussion after a talk by the author
  on the \emph{Notes} and the questions they leave open, at the conference of the
  Kurt-G\"odel-Forschungsstelle of the Berlin-Brandenburg Academy of Sciences and Humanities in
  Berlin in July 2026; the talk had suggested that current AI
  systems might help to settle them.  Proofs of the \textsf{S4} questions followed within a day,
  but were not circulated, because nominals appeared to be doing work in them --- and that is the
  question taken up here.  Section~\ref{sec:verification} says what the AI assistant did.}
and the detour is the more interesting part.  Table~\ref{tab:verdicts} at the end of
Section~\ref{sec:emendation} collects what each question receives under each reading of the
conjunction axiom, and on what condition.  The dataset of the \emph{Notes} records
the first two questions in several quantifier settings, as ten open statements: \tc{Th3} of
variants~2 and~3 in \textsf{S4}, actualist and possibilist, and \tc{UniqueEss1} and
\tc{UniqueEss2} of variant~3, actualist, possibilist and mixed.  The third question has no
statement there, and the possibilist and mixed copies are settled as well
(Section~\ref{sec:copies}).

Section~\ref{sec:gex} keeps the generalised conjunction axiom \tc{Ax1Gen} as the \emph{Notes}
state it.\footnote{What is weighed below is therefore a reading the literature shares, not a
  stipulation of the \emph{Notes}; Section~\ref{sec:setting} states the axiom and gives its
  provenance.}  With the exclusivity axiom on positive properties (\tc{Ax2a}) and their closure
under necessary property inclusion (\tc{Ax4}), it follows with no assumption on the accessibility
relation that an actual God-like being exists at every world, and with it \tc{Th3}, \tc{T3} and
modal collapse in \textsf{K}.  That is too quick to be a convincing answer, because the proof
turns on the empty conjunction: where no actual God-like being exists, \tc{Ax1Gen} applied to
the empty set makes that very absence a positive property.  A reader who rejects that instance rejects the
proof, and Section~\ref{sec:emendation} gives the reason: the empty and the
singleton conjunction are what make the axiom too strong.  Reading G\"odel's ``any number of
summands'' non-degenerately, as at least two different conjuncts --- \tc{Ax1GenTwo} ---
excludes exactly them, and under that reading the questions are settled again.  For variant~2 and Scott's variant this is
outright, the axioms excluding a world with nothing actual in it; variant~3 needs, in addition, the
confinement of the non-emptiness clause mentioned above (Section~\ref{sec:plain}).  Nominal instances yield more than the answers: in
variant~2 they make every world its own only successor, and in variant~3 they force the frame
to a single world (Sections~\ref{sec:gex} and~\ref{sec:v3}).

Second, the question this article adds --- whether a result is one of the modal logic ---
receives the answer the criterion can give (Section~\ref{sec:criterion}), for the whole
development the three questions come from and in both systems.  Every statement the
\emph{Notes} prove has a proof that substitutes nothing outside the object language: the port's
in \lean{}, and in \isa{} either the \emph{Notes}' own proof or one written for the purpose
where the detector does not certify theirs (Section~\ref{sec:audit}).  The audit corrects the record in one place: the possible
existence of a God-like being in the mixed-quantifier setting, which the \emph{Notes} report as
needing symmetry, needs no frame condition at all
(Proposition~\ref{prop:th4mixed}).\footnote{Of the \emph{Notes}' 427 facts, read with the proofs they themselves give, 313
are certified and none instantiates a nominal; the 114 that are not carry a term the prover
substituted, not one the argument needs, and each of their statements has a proof written out
for the purpose that is certified.  What the detector can read depends on how a proof was
made, not on the system it was made in; Appendix~\ref{app:runs} gives the 114 by kind and states
what each certificate does not cover.}

\section{The embedding and the three developments}
\label{sec:setting}

The embedding follows~\cite{J75} exactly and is only summarised.  Worlds have type $i$,
individuals type $e$; a modal proposition is a predicate $\sigma = i\to o$ on worlds, $o$
being the type of truth values, a modal property a function $\tau = e\to\sigma$. The connectives are lifted pointwise, among
them $\vee^{e}$ for exclusive disjunction and ${\sim}\varphi$ for the complement of a
property, and $\Box\varphi$ holds at $w$ iff $\varphi$ holds at every $v$ with $w\,r\,v$.
Quantifiers
over individuals come in a possibilist form $\forall x$ and an actualist form
$\forall\actE x$, over the individuals that exist at the world of evaluation.  The
possibilist quantifier is polymorphic: the same $\forall$ binds variables for propositions
and for properties, and the definitions below use it so; only individuals have an actualist
form.  The notation $\val{\varphi}$ says that $\varphi$ holds at every world, and all axioms are
stated under $\val{\cdot}$.  Logic \textsf{S4} postulates that $r$ is reflexive and
transitive.  It is the logic here because the open questions are posed in it: whether
\textsf{S4} suffices where the proofs of~\cite{J75} used symmetry, and whether two of the
lemmas of variant~3 hold even in \textsf{S5}.  Nothing else is assumed about $r$; where a result below is said
to hold in \textsf{K}, its proof uses neither reflexivity nor transitivity.

Three developments are considered, each in \textsf{S4}: variants~2 and~3 of the \emph{Notes},
and Scott's.  The two variants differ in which of G\"odel's 1970 definitions the \emph{Notes}
repair, the essence in variant~2 and necessary property inclusion in variant~3; unrepaired, as
variant~1, the axioms are inconsistent~\cite[Sect.~4.2]{J75}.  A fourth setting, the
mixed-quantifier reading of variant~2 in \textsf{K}, is introduced in
Section~\ref{sec:copies} and returns in Sections~\ref{sec:ax1gentwo}, \ref{sec:audit}
and~\ref{sec:discussion}, where the quantifier settings are compared.

\subsection{Variant 2}

A predicate $\Pos$ on properties
(``positive''), with the axioms of~\NFig{7}: \tc{Ax1},
$\val{\Pos\varphi \wedge \Pos\psi \supset \Pos(\lambda x.\ \varphi x \wedge \psi x)}$;
\tc{Ax2a}, $\val{\Pos\varphi \vee^{e} \Pos{\sim}\varphi}$;
\tc{Ax2b}, $\val{\Pos\varphi\supset\Box\Pos\varphi}$;
\tc{Ax3}, $\val{\Pos\,E}$; \tc{Ax4}, $\val{\Pos\varphi \wedge (\varphi\supset_N\psi)
\supset \Pos\psi}$; and the generalised conjunction axiom
\[
\tc{Ax1Gen}:\qquad \val{(\mathit{PosProps}\,\Phi \wedge
  \mathit{ConjOfPropsFrom}\,\varphi\,\Phi) \supset \Pos\varphi},
\]
where $\Phi$ is a (world-dependent) set of properties, $\mathit{PosProps}\,\Phi$ says
that every member of $\Phi$ is positive, and $\mathit{ConjOfPropsFrom}\,\varphi\,\Phi =
\Box(\forall\actE z.\ \varphi z \leftrightarrow \forall\psi.\ \Phi\psi\supset\psi z)$
says that $\varphi$ is necessarily the conjunction of $\Phi$.  This axiom formalises
G\"odel's footnote on ``any number of summands''~\cite{GoedelNotes}, following Anderson and
Gettings~\cite{AndersonGettings1996} and Fitting~\cite{Fitting2002}.  An axiom of that form
stands in G\"odel's own notes of 1952/53, as Kanckos and Lethen transcribe and reconstruct
them~\cite[p.~1019]{KanckosLethen2021}: in their notation $\Delta A$ is the conjunction of the
properties a second-order $A$ collects, and G\"odel's Axiom~2 says that $\Delta A$ is positive if
every property $A$ collects is positive; the positivity of being God-like is Theorem~1 of those
notes.  Axiom~2 carries no restriction on $A$, as the footnote carries none on the number of
summands.  Definitions: $G\,x =
\forall\varphi.\ \Pos\varphi\supset\varphi x$ (God-like); $\varphi\supset_N\psi =
\Box(\forall\actE y.\ \varphi y\supset\psi y)$ (necessary property inclusion);
$\varphi\Ess x = \varphi x \wedge \forall\psi.\ \psi x \supset (\varphi\supset_N\psi)$
(essence, with the conjunct $\varphi x$ that repairs G\"odel's 1970 definition);
$E\,x = \forall\varphi.\ \varphi\Ess x\supset\Box\exists\actE y.\ \varphi y$ (necessary
existence).  The conjunct $\varphi x$ is Scott's, but it was not new with him: G\"odel had it
already in the essence of one of his own earlier versions, the one Kanckos and Lethen
conjecture to be the latest of them~\cite[p.~1017]{KanckosLethen2021}.

\subsection{Variant 3}

As variant 2, but with the definitions of~\NFig{8}: essence is G\"odel's
original $\varphi\Ess x = \forall\psi.\ \psi x\supset(\varphi\supset_N\psi)$ and the
repair is made in the inclusion instead: $\varphi\supset_N\psi = \Box(\varphi\neq
(\lambda x.\bot) \wedge \forall\actE y.\ \varphi y\supset\psi y)$, with $\varphi$ required
to be non-empty.

The two essences are not interdefinable, and it is worth saying which way the difference runs.

\begin{lemma}
\label{lem:esscompare}
Variant~2's essence implies variant~3's from the definitions alone (\tc{Ess\_of\_Ess2}).  The
converse fails: with no axiom of the argument assumed, Nitpick finds a $\varphi$ that is an
essence of $x$ in the sense of~\NFig{8} and is not had by $x$ (\tc{EssenceComparison}).  In the
theories of the argument the two coincide, $\varphi x$ being a theorem there
(\tc{Ess\_iff\_Ess2}, and \tc{Ess\_iff\_Ess2\_K} with no frame condition and hybrid-free).
\end{lemma}

\begin{proof}
From $\varphi x$ the property $\varphi$ is not $\lambda z.\bot$, and that inequality holds at
every world, so variant~2's essence carries the clause of~\NFig{8} under the $\Box$ and with it
variant~3's.  Conversely variant~3's gives variant~2's inclusion by dropping the clause, so
that only $\varphi x$ is wanting --- the clause says that $\varphi$ is not the empty property,
not that $x$ has it.  In variant~3 the theorem $\varphi x$ is \tc{EssMemI}, which passes through
Theorem~\ref{thm:nonecexist} and its nominal, while \tc{EssMemI\_K} proves it hybrid-free
(Theorem~\ref{thm:uniqueess-k}(a)); the equivalence is proved both ways, in both systems.
\end{proof}

\subsection{Scott's variant with \tc{Ax1Gen}}

Scott's axioms~\cite{ScottNotes}, as~\NFig{12} states them:
\tc{A1}, $\val{\neg\Pos\varphi\leftrightarrow\Pos{\sim}\varphi}$; \tc{A2},
$\val{\Pos\varphi\wedge\Box(\forall\actE y.\ \varphi y\supset\psi y)\supset\Pos\psi}$;
\tc{A4}, $\val{\Pos\varphi\supset\Box\Pos\varphi}$; \tc{A5}, $\val{\Pos\,\NE}$; and
\tc{Ax1Gen} in place of \tc{A3}, $\val{\Pos\,G}$, which then follows (Lemma \tc{L}).
Essence and necessary existence are as in variant~2.

In all three, the \emph{Notes} prove the possible existence of a God-like being
(\tc{Th4}, resp.\ \tc{Coro}), that a God-like being has the essence $G$ (\tc{Th1}/\tc{T2})
and that its existence is necessary (\tc{Th2}); none of these uses a condition on $r$.  From
the last two, with symmetry, they prove \tc{Th3}:
$\val{\Diamond\exists\actE x.\,G x \supset \Box\exists\actE y.\,G y}$, hence
\tc{Th5}/\tc{T3}: $\val{\Box\exists\actE x.\,G x}$, and modal collapse \tc{MC}:
$\val{\varphi\supset\Box\varphi}$.  Their theories are set in \textsf{S5}, but symmetry is
the only frame condition these proofs use.  In \textsf{S4}, \tc{Th3} is the open question.

\section{Hybrid witnesses}
\label{sec:criterion}

Write $\mathcal{L}$ for the modal object language of the \emph{Notes}: terms built from
the lifted connectives, $\Box$, $\Diamond$, the possibilist and actualist quantifiers, the
constants the theory in question declares --- $\Pos$, $G$ and its other defined notions ---
and bound variables of the types $e$, $\sigma$, $\tau$, \dots.  Quantification at $\sigma$ makes
$\mathcal{L}$ a propositionally quantified modal language, and what such a language can say
turns on which sets of worlds its propositional domain
contains~\cite{Fine1970,Holliday2019,Fritz2024}, which is where the question of this article
lies (Section~\ref{sec:discussion}).  The connectives are the
same in every theory; only the signature changes from one variant to the next, so
$\mathcal{L}$-definability is relative to the theory the proof under examination belongs
to.  Among the connectives is the primitive equality $=$ of the \emph{Notes}, declared at every
type: at the type of individuals it is the identity of the modal model's domain, at a lifted
type it is identity across all worlds.  So $\varphi = \top$ is the universal modality
$\mathsf{A}\varphi$, its dual $\varphi \neq \bot$ the existential modality $\mathsf{E}\varphi$,
and $\mathcal{L}$ is the language of \textsf{K}, or \textsf{S4}, with the global modality.  That
is the language the \emph{Notes} write their axioms in, the non-emptiness clause of \NFig{8},
$\varphi \neq \bot$, among them.  What the global modality cannot express is
what this article is about: $\mathcal{L}$ has no term of the world type, so nothing in it
refers to a world, as a nominal and the hybrid devices built on one do.  In the embedding every term of $\mathcal{L}$ is a HOL term,
but not conversely.  The HOL terms $\lambda u{:}i.\ u = w$, $\lambda z\,u.\ z = x \wedge u
= w$ and $\lambda x\,u.\ w\,r\,u$ are a proposition and two properties of the embedding, and
they are not $\mathcal{L}$-definable: they are the nominal $\mathsf{w}$ of hybrid logic, the formula
$\mathsf{w}\wedge(z = x)$, and $\Diamond^{-1}\mathsf{w}$, that nominal under the converse
modality of tense logic.  Meta-level
reasoning about worlds, on the other hand --- ``fix $w$, let $v$ be a successor'', a case
split on whether some world is a dead end --- is nothing but Kripke semantics and
involves no hybrid expressiveness.  The line therefore has to be drawn at
\emph{instantiations}:

\begin{definition}[hybrid-free]
\label{def:hybridfree}
A proof is \emph{hybrid-free} if every term it substitutes for a variable of one of the types
below stays within the object language.  These are the types the axioms quantify over beyond
the individuals, and they are of two kinds.  The \emph{world-lifted} types are those built from
$e$ and $\sigma$ that end in $\sigma$: a proposition $\sigma$, a property $\tau$, a
world-dependent set of properties $(e\to\sigma)\to\sigma$.  The others end in $o$ and take a
world-lifted argument, as a set of properties \emph{simpliciter} $(e\to\sigma)\to o$ does.  At a world-lifted type the substituted term must
be a term of $\mathcal{L}$ of that type.  The other types have no $\mathcal{L}$-terms at all, since
the formulas of $\mathcal{L}$ are world-lifted.  There the substituted term must express a
condition formed in
$\mathcal{L}$ --- the validity of formulas of $\mathcal{L}$, combined by the propositional
connectives and by quantification over anything but worlds --- as
$\{\psi\mid\val{\chi\,\psi}\}$ does, for a formula $\chi$ of $\mathcal{L}$ with a free property
variable.  Such a
substitution is made in an axiom instance, in instantiating the quantifier in the definition of
God-likeness, of essence or of necessary existence, in the set $\Phi$ of \tc{Ax1Gen}, and in a lemma or a hypothesis applied
to such a term: wherever a variable of one of those types is instantiated.
\end{definition}

Variables of other types are not at issue.  An object-language reading of the axioms restricts
the domains just named and nothing else.  Substituting a property or proposition those domains
need not contain changes what the axiom says; substituting an individual changes nothing, since
an individual is an element of the modal model's own domain whatever term a proof names it by.  In the terms of propositionally quantified
modal logic~\cite{Fine1970,Holliday2019,Fritz2024}: an instantiation of a quantifier over
properties or propositions uses an instance of comprehension --- that some property or
proposition satisfies the substituted condition --- and a proof is hybrid-free when the
condition of every comprehension instance it uses at those types is a formula of
$\mathcal{L}$, with parameters of the types $\mathcal{L}$ quantifies over and never a world.
A case distinction on worlds in the metalanguage asserts the existence of no proposition; an
instantiation does.  What such a certificate says has content only relative to general models,
whose propositional domain need not contain every set of worlds, since over standard models,
where it does, $\mathcal{L}$ defines nominals by itself; Section~\ref{sec:criterion-open} gives the
term and states the semantics the certificate answers to.

Four of those types occur in the developments read here, the three world-lifted ones above and
the sets of properties \emph{simpliciter}: a rigid reading of the conjunction axiom gives $\Phi$
a set of the latter kind, whereas the theories audited here give it a world-dependent one.  A
satisfaction operator can enter at that type as readily as a nominal can at a world-lifted one,
as it does in the rigid reading (\tc{L\_R}, Section~\ref{sec:rigid}), so that type is examined
together with the three world-lifted ones.

A proof term records what a proof substituted, and nothing else does: the statement of a
theorem says nothing about the witnesses its proof used, and a proof script names them only
where the author wrote them out.  The detector therefore reads the proof term --- reconstructed
in \isa{}, which records proofs only on request, and as it stands in \lean{} --- and checks
every term substituted at one of Definition~\ref{def:hybridfree}'s types against a grammar of
$\mathcal{L}$, after unfolding every defined notion other than the connectives of the embedding.
The grammar is a whitelist with two rules.  First, every binder, variable and constant of the
witness must have a type of $\mathcal{L}$: a type built from $e$, $o$ and $\sigma$ by $\to$,
with $\sigma$ taken as a unit, so that $i$ itself is not one.  The type of worlds is admitted
in two places only: as the type of a binder of the witness, and as the type argument of a
quantifier, which is then read as a modality, or of an equality, which is then a nominal.
Second, the body of a binder over a world may mention no other world bound inside the witness:
a lifted term speaks of its own world of evaluation.  A quantifier over worlds is therefore
accepted exactly where it is the universal or the existential modality --- $\forall v.\,F\,v$
or $\exists v.\,F\,v$ for a term $F$ of $\mathcal{L}$, which are $\mathsf{A}F$ and
$\mathsf{E}F$, and validity $\val{F}$, which is $\mathsf{A}F$ again.
Whatever falls outside is flagged, and the flag says why: an equation between worlds, the
accessibility relation, a world variable bound outside the term, a second world under a binder
over worlds, or a binder, variable or constant whose type is not a type of $\mathcal{L}$.
The detector detects nothing about the statement proved, and it does not say whether the
theorem has a derivation in an object-level calculus.
Appendix~\ref{app:detector} gives the grammar, the flags and the two implementations.  The check is syntactic, and what it
establishes is this.

\begin{proposition}[soundness of the criterion]
\label{prop:sound}
If the detector accepts every witness of a proof, then that proof is hybrid-free: every term
it substitutes at one of the types of Definition~\ref{def:hybridfree} is a term of
$\mathcal{L}$.
\end{proposition}

\begin{proof}
The clauses of the grammar are those of the definition of $\mathcal{L}$ above, one for one:
(i)~the connectives of the embedding, the polymorphic ones --- the possibilist quantifiers and
the lifted equalities --- at a type of $\mathcal{L}$, or at the type of worlds, where the
quantifier is the universal or the existential modality and the equality is flagged; (ii)~the
raw connectives, equality and quantification they unfold to, likewise restricted; (iii)~the
constants the theory declares, at a type of $\mathcal{L}$; and (iv)~variables of a type of
$\mathcal{L}$, bound in the term or outside it, and worlds bound in the term, the body of each
such binder mentioning no other.  A witness the detector accepts is built from these and
nothing else, and induction on the witness gives the rest: a subterm of type $\sigma$ speaks, under its binder, of
its own world of evaluation alone, so it is a term of $\mathcal{L}$ at that world; a raw
quantifier over worlds whose body does the same is $\mathsf{A}F$ or $\mathsf{E}F$ for such a
term $F$; and at the other types of $\mathcal{L}$ the term is a condition formed in
$\mathcal{L}$, which is what the definition asks there.  Definitions other than the
connectives are unfolded before the check, so no witness reaches outside $\mathcal{L}$ behind
a defined constant.  The grammar thus admits nothing the definition does not admit; it is a
whitelist, not a list of patterns a term must avoid.
\end{proof}

\noindent What the proposition gives is a sufficient condition for the absence of hybrid
witnesses, and no more.  A proof the detector rejects may be hybrid-free all the same, since
the grammar is exact with respect to its own clauses and conservative with respect to what is
expressible: $\lambda u.\ \forall v.\ F\,v \to G\,u$ is $\mathsf{E}F \to G$, and is flagged.  And a statement whose proof it rejects may have
another proof it accepts, as the two cases of Section~\ref{sec:criterion-open} show.  A proof it accepts is
certified in what it substitutes, not in having a derivation in an object-level calculus
(Section~\ref{sec:discussion}).

\section{A God-like being exists at every world, from unrestricted \tc{Ax1Gen} as stated}
\label{sec:gex}

This section and the next settle the three questions in the terms the \emph{Notes} pose them,
with no frame condition and hybrid-free.  Two reservations attach
to the answers, and the article takes them up separately.  The proofs turn on the empty
instance of \tc{Ax1Gen}, which Section~\ref{sec:emendation} shows --- with the singleton --- to
be what makes the axiom too strong, and it gives the answers again under an emendation that
excludes both.  And the same
axioms yield more than the answers --- that each world is its own only successor here, that there
is a single world in variant~3 --- but only through a nominal, so those are statements about the
models and not of the object language.

\tc{Ax1Gen} may be applied to a set $\Phi$ of properties that depends on the world of
evaluation.  A suitable $\Phi$ then makes the property ``no actual God-like being exists''
necessarily the conjunction of $\Phi$, with $\Phi$ empty at exactly the worlds where that
property holds --- so that the axiom applies there with nothing to check.

\begin{lemma}[\tc{L2}]
\label{lem:l2}
Let $\chi = \lambda z.\ \neg\exists\actE y.\,G y$ (a property constant in $z$) and
\[
\Phi = \lambda\psi.\ (\exists\actE y.\,G y) \wedge
   \bigl(\Pos\psi \vee \Box(\forall z.\ \psi z \leftrightarrow \neg G z)\bigr).
\]
Then $\val{\mathit{ConjOfPropsFrom}\,\chi\,\Phi}$.  The proof uses only Lemma
\tc{L}, $\val{\Pos\,G}$, and is hybrid-free.
\end{lemma}

\begin{proof}
Fix a world $v$ and an actual individual $z$ at $v$.  If no actual God-like being exists
at $v$, then $\chi z$ holds at $v$ and $\Phi$ is empty at $v$, so the conjunction of $\Phi$
holds of $z$ vacuously: it asks that $\psi z$ hold for every $\psi$ in $\Phi$, and at $v$
there is none.  If an actual God-like being exists at $v$, then $\chi z$ fails,
and $\Phi$ contains both $G$, positive by \tc{L}, and ${\sim}G$, which satisfies the second
disjunct because ${\sim}G\,z$ is $\neg G z$ by definition; their conjunction fails of every $z$.
\end{proof}

\begin{theorem}[\tc{G\_ex}]
\label{thm:gex}
In variant 2, $\val{\exists\actE x.\,G x}$: at every world an actual God-like being
exists.  The proof uses \tc{Ax1Gen}, \tc{Ax2a} and \tc{Ax4}, with no frame condition, and is
hybrid-free.
\end{theorem}

\begin{proof}
Suppose no actual God-like being exists at $w$.  Then $\Phi$ of Lemma~\ref{lem:l2} is
empty at $w$, so $\mathit{PosProps}\,\Phi$ holds at $w$ vacuously, and by
Lemma~\ref{lem:l2} and \tc{Ax1Gen}, $\chi$ is positive at $w$.  Now $\chi\supset_N{\sim}G$
holds at $w$: at any successor $v$, if no actual God-like being exists at $v$ then no
actual $y$ is God-like at $v$.  By \tc{Ax4}, ${\sim}G$ is positive at $w$.  But $G$ is
positive by \tc{L}, contradicting the exclusivity in \tc{Ax2a}.
\end{proof}

The same proof, with \tc{A2} for \tc{Ax4} and \tc{A1} for \tc{Ax2a}, gives:

\begin{theorem}
\label{thm:gex-scott}
In Scott's variant with \tc{Ax1Gen}, $\val{\exists\actE x.\,G x}$, from \tc{Ax1Gen},
\tc{A1} and \tc{A2}, with no frame condition, and hybrid-free.
\end{theorem}

What the \emph{Notes} ask about variant~2 and Scott's variant follows.

\begin{theorem}
\label{thm:consequences}
In variant 2 and in Scott's variant with \tc{Ax1Gen}:
\begin{enumerate}
\item[(a)] \tc{Th3}: $\val{\Diamond\exists\actE x.\,G x\supset\Box\exists\actE y.\,G y}$,
  and \tc{Th5}/\tc{T3}: $\val{\Box\exists\actE x.\,G x}$ --- with no frame condition.
\item[(b)] \tc{MC}: $\val{\varphi\supset\Box\varphi}$ for every $\varphi$ --- with no frame
  condition (the proof uses, besides Theorem~\ref{thm:gex}, only \tc{Th1}, resp.\ \tc{T2}).
\item[(c)] \tc{OneWorld}: if $w\,r\,v$ then $v = w$.
\end{enumerate}
Parts (a) and (b) are hybrid-free; (c) is not, its proof instantiating the collapse at a
nominal.
\end{theorem}

\begin{proof}
(a) is immediate: the consequent of \tc{Th3} holds at every world by
Theorem~\ref{thm:gex}.  (b) Let $\varphi$ hold at $w$ and $w\,r\,v$.  By
Theorem~\ref{thm:gex} some $g$ is God-like at $w$, and by \tc{Th1} $G$ is an essence of
$g$; since the constant property $\lambda y.\varphi$ holds of $g$ at $w$, $G$ necessarily
implies it, so at $v$ every actual God-like being has it.  By Theorem~\ref{thm:gex} some
actual God-like being exists at $v$, and $\lambda y.\varphi$ holding of it at $v$ is
$\varphi$ holding at $v$.  (c) Apply (b) to the nominal $\lambda u.\ u = w$.
\end{proof}

What these proofs exploit, and all the detector finds in them, is that the set $\Phi$ may
depend on the world.  Theorem~\ref{thm:consequences}(c) is the frame
correspondent of (b), read off inside the embedding with that nominal, and is stated for
what it says about the models.  Questions~1 and~3 are thereby answered as the
\emph{Notes} pose them: \tc{Th3} holds in \textsf{S4}, and indeed in \textsf{K}, and so
does \tc{T3}; since the proof uses no frame condition, no countermodel exists in any normal
modal logic.  This is why the \textsf{S4} countermodels that Nitpick had found for Scott's
variant with \tc{A3} could not be reproduced once \tc{A3} was replaced by \tc{Ax1Gen}.

The proof of Theorem~\ref{thm:gex} turns on the case in which
$\Phi$ is empty, where its conjunction is the universal property; the case in which $\Phi$
has two members does the rest.  Table~\ref{tab:instances} records which of the results below
rest on a degenerate instance, read off the \lean{} proofs, where the proof applies the axiom
to a set written out in the source; where a goal is closed by a tactic instead, in either system,
there is nothing of the kind to read (Section~\ref{sec:audit}).  The empty conjunction is an
instance of \tc{Ax1Gen} that
G\"odel's ``any number of summands'' need not license, and
Section~\ref{sec:emendation} shows it to be the instance that makes the axiom too strong:
together with the singleton it yields the possible existence of a God-like being, which is
what \tc{Ax4} is for.  Read strictly, then, this section establishes something about
\tc{Ax1Gen} rather than about the open questions.  Section~\ref{sec:v3} first follows the same
argument into variant~3, where it reaches further still;
Section~\ref{sec:emendation} then gives the answers again under an emendation that excludes the
degenerate instances.

\section{Variant 3: there is exactly one world, with a nominal and no conjunction axiom}
\label{sec:v3}

In the variant with modified property inclusion the argument of Section~\ref{sec:gex} goes
through as well.  The inclusion $\chi\supset_N{\sim}G$ now also requires $\chi$ to be
non-empty.  It is: the proof applies it under the assumption that no actual God-like being
exists at $w$, and $\chi$ then holds of every individual there.  So \tc{Th3},
\tc{Th5} and \tc{MC} hold hybrid-free and frame-free here too (\tc{Th3\_K},
\tc{Th5\_K}, \tc{MC\_K} in the sources).  This answers the \tc{Th3} part of question~2.
But here a much stronger statement holds, and it comes from the definitions rather than
from the axioms --- at the price of nominals.

The section has two layers.  With a nominal, variant~3's definitions together with \tc{Ax2a},
\tc{Ax3} and \tc{Ax4} --- no conjunction axiom --- leave exactly one world, and in a single
world an essence can only be ``being $x$'', the property $\lambda z\,u.\ z = x$.  Without a
nominal, and with no frame condition, the two uniqueness lemmas the \emph{Notes} leave open
are proved outright.  That is the answer to question~2.  The single world yields the two
lemmas as well, but only through the nominal, so it explains the answer rather than being what
establishes it.  Both layers are proved below, and Section~\ref{sec:emendation} takes their
statements up again under the restricted reading.

\begin{theorem}[\tc{NoNecExist}]
\label{thm:nonecexist}
In variant 3, if $v\neq w$ then $\neg E\,x\,w$ for every $x$: if there are two distinct
worlds, nothing has necessary existence.  The proof uses reflexivity and no axiom, and it is hybrid: its two witnesses are ``being
$x$ at $w$'' and ``being at $v$''.
\end{theorem}

\begin{proof}
Suppose $E\,x\,w$ and $v\neq w$.  If $w\,r\,v$, consider the property $\varphi = \lambda
z\,u.\ z = x\wedge u = w$, ``being $x$ at $w$''.  It is non-empty, and it is an essence of
$x$ at $w$: for any $\psi$ with $\psi x w$ and any successor $u$ of $w$, an actual $y$
with $\varphi y u$ is $x$ at $w$, so $\psi y u$.  By $E\,x\,w$, $\varphi$ is necessarily
exemplified by an actual being; at $v$ this gives $v = w$, against the hypothesis.  If not
$w\,r\,v$, consider
$\varphi = \lambda z\,u.\ u = v$, ``being at $v$''.  It is non-empty, and it is an essence
of $x$ at $w$ vacuously, since no successor of $w$ is $v$.  By $E\,x\,w$ and reflexivity,
$\varphi$ is exemplified at $w$, so $w = v$, against the hypothesis again.
\end{proof}

\begin{theorem}
\label{thm:v3}
In variant 3:
\begin{enumerate}
\item[(a)] \tc{Th3} holds, from \tc{Ax3} and reflexivity, with no conjunction axiom.
\item[(b)] \tc{OneWorld}: any two worlds are equal.
\item[(c)] \tc{G\_ex}, \tc{Th4}, \tc{Th5} and \tc{MC} hold, and so does \tc{AllExist},
  that is $\val{\forall y.\ \mathit{existsAt}\,y}$.
\item[(d)] \tc{EssDetermined}: if $\varphi$ is an essence of $x$ at some world then
  $\varphi = \lambda z\,u.\ z = x$.  Consequently \tc{UniqueEss1},
  $\val{\varphi\Ess x\wedge\psi\Ess x\supset\Box\forall\actE y.\ \varphi y\leftrightarrow
  \psi y}$, and \tc{UniqueEss2}, $\val{\varphi\Ess x\wedge\psi\Ess x\supset
  \Box(\varphi\equiv\psi)}$ with Leibniz equality, both hold.
\end{enumerate}
Of these, \tc{G\_ex}, \tc{Th4} and \tc{AllExist} in (c) are hybrid-free.  The rest --- (a),
(b), \tc{Th5} and \tc{MC} in (c), and (d) --- pass through Theorem~\ref{thm:nonecexist} and its
nominal here.  \tc{Th5}, \tc{MC} and the uniqueness results of (d) have hybrid-free proofs as
well, from the conjunction axiom (\tc{Th5\_K}, \tc{MC\_K} above, Theorem~\ref{thm:uniqueess-k}
below); (b) has none, being a statement about the worlds themselves, hybrid by nature
(Section~\ref{sec:criterion-open}).  For (a) this passage is a second route to a formula the
opening of this section has already proved hybrid-free, and what the nominal buys on that route is the
absence of the conjunction axiom: \tc{Th3\_K} rests on \tc{Ax1Gen}, (a) on \tc{Ax3} and reflexivity
alone.
\end{theorem}

\begin{proof}
(a) Let $\Diamond\exists\actE x.\,G x$ hold at $w$: some successor $v$ of $w$ carries an actual
God-like being $x$.  By \tc{Ax3} it has necessary existence there, so by
Theorem~\ref{thm:nonecexist} no world differs from $v$; every successor of $w$ is $v$, and $x$ is
actual and God-like at it.
(b) Were $u\neq v$, nothing would have necessary existence
(Theorem~\ref{thm:nonecexist}), so $E$ would be the empty property, which is positive by
\tc{Ax3}; by \tc{Ax2a} its complement, the universal property, is then not positive.  But by
\tc{Ax2a} one of ``being at $u$'', $\lambda x\,t.\ t = u$ --- the nominal this proof turns on ---
and its complement is positive; both are non-empty, and each is necessarily included in the
universal property, so by \tc{Ax4} the universal property is positive after all.
(c) \tc{G\_ex} is the argument of Section~\ref{sec:gex} as above, and \tc{Th4} follows from it
with reflexivity; \tc{Th5} follows from (a) and \tc{Th4}, and \tc{MC} from (b), there being one
world only.  For \tc{AllExist}: were $y$ not actual at $w$, ``being $y$ and not actual'' would be
a non-empty property and, vacuously, an essence of the God-like being of \tc{G\_ex} at $w$.  That
being has necessary existence by \tc{Ax3}, so the property is exemplified at a successor of $w$ by
an actual being --- which it makes not actual there.
(d) Let $\varphi$ be an essence of $x$ at $w$.  ``Being $x$'' is a property $x$ has, so $\varphi$
necessarily implies it, and by (b) and \tc{AllExist} this can be read off at the one world: every $z$
with $\varphi z$ is $x$.  Conversely $x$ has $\varphi$, for otherwise the complement of $\varphi$
would be a property $x$ has, $\varphi$ would necessarily imply it, and $\varphi$ would be the
empty property, which variant~3's inclusion requires it not to be.  So $\varphi = \lambda z\,u.\
z = x$, and \tc{UniqueEss1} and \tc{UniqueEss2} follow by rewriting.
\end{proof}

The detector confirms the division.  It flags the nominal of Theorem~\ref{thm:nonecexist} and
that of Theorem~\ref{thm:consequences}(c) --- there $\lambda u.\ u = w$ itself --- and marks as
hybrid every result their proofs are used in; of the proofs it certifies here, only
\tc{AllExist} carries a flag at all, the informational \textsc{actual} for the existence
predicate.

For the two open lemmas, however, the nominal can be dispensed with, and two steps replace it.
Theorem~\ref{thm:nonecexist} reads a $\Box$ off at some world, which in \textsf{S4} is the world
of evaluation itself; that is the only service reflexivity does here, and the next lemma shows
that the axioms supply such a world by themselves.  And where the nominal pins a property to a
world, Lemma~\ref{lem:reach} does the same work with a constant property.

\begin{lemma}[\tc{Serial}]
\label{lem:serial}
$\val{\Diamond\top}$, every world has a successor, from \tc{Ax2a}, \tc{Ax3} and \tc{Ax4}.
\end{lemma}

\begin{proof}
Let $w$ have no successor.  Every $\Box\varphi$ then holds at $w$ vacuously, in particular every
necessary property inclusion, so \tc{Ax4} carries necessary existence, positive by \tc{Ax3}, to
an arbitrary $\psi$: at $w$ every property is positive, $\lambda x.\top$ and its complement among
them, against the exclusivity of \tc{Ax2a}.
\end{proof}

The lemma is what lets the hybrid-free results below, and those of
Sections~\ref{sec:ax1gentwo} and~\ref{sec:plain}, do without a frame condition altogether.
No conjunction axiom enters, and the shape of the inclusion is immaterial: the clause of \NFig{8}
sits inside the $\Box$, so at a dead end the whole inclusion is vacuously true, clause included,
and the lemma holds in variant~2 and in variant~3 alike.  Section~\ref{sec:audit} exploits the
same dead end for \tc{Th4\_K}, there with \tc{Ax1Gen}, whose empty instance makes every property
positive; the restriction of Section~\ref{sec:ax1gentwo} closes that route and not this one,
which is why the lemma is stated separately.  The lemma also
removes reflexivity from Theorem~\ref{thm:nonecexist}, at the price of the lemma's three axioms,
and hence
from Theorem~\ref{thm:v3}(b): variant~3's definitions with \tc{Ax2a}, \tc{Ax3} and \tc{Ax4} force
a single world in \textsf{K}.  The proofs recorded for the \emph{Notes}' own theories state
Theorem~\ref{thm:nonecexist} as it stands, since reflexivity without any axiom is the sharper hypothesis; the reflexivity-free ones are in \tc{Ax1GenTwoVariant3Th4}.

In \textsf{S4} the lemma adds nothing, reflexivity giving seriality at once, and the
\emph{Notes} have little occasion for it.  Of their results, the only frame condition
any one of them uses is symmetry, with one exception: \tc{Th3} of their \textsf{S4} variant-1 theory, which uses
reflexivity and \tc{Ax3} and nothing else --- in the theory the \emph{Notes} show inconsistent.
Reflexivity and transitivity otherwise occur only in their two test theories, where the schemes
\textsf{D}, \textsf{M}, \textsf{4} and \textsf{5} are derived from them on purpose.

\begin{lemma}[\tc{Reach\_K}]
\label{lem:reach}
In variant 3, every non-empty property is exemplified at some successor of every world:
if $\varphi z t$ for some $z$, $t$, then for every $w$ there are $v$ with $w\,r\,v$ and $y$
with $\varphi y v$.  Hybrid-free, and with no frame condition; uses \tc{Ax1Gen}, \tc{Ax2a},
\tc{Ax3} and \tc{Ax4}.
\end{lemma}

\begin{proof}
Let $z$ and $t$ be as in the hypothesis, so that $\varphi z t$, and let $g$ be a God-like being
at $w$ (Theorem~\ref{thm:v3}(c)); by \tc{Ax3},
$E\,g\,w$.  Let $\chi = \lambda x.\ \varphi z$, the property, constant in its individual
argument, that $z$ has $\varphi$.  It is not the empty property, since $\chi z t$.  Suppose
nothing at any successor of $w$ had $\varphi$.  Then nothing there has $\chi$ either, and
$\chi$ is vacuously an essence of $g$ at $w$: the implication inside the inclusion has no
instance there, while the non-emptiness clause, being $\chi\neq\bot$ and so world-independent, is
met by $\chi z t$ at $t$.  So $E\,g\,w$
gives $\Box\exists\actE x.\ \chi x$ at $w$, and at a successor $v$ of $w$
(Lemma~\ref{lem:serial}) some actual being has $\chi$, which is to say $\varphi z v$:
$\varphi$ is had at a successor of $w$ after all.
\end{proof}

\begin{theorem}
\label{thm:uniqueess-k}
In variant 3, hybrid-free: (a) an essence of $x$ at $w$ is had by $x$ at $w$
(\tc{EssMemI\_K}); (b) \tc{UniqueEss1}; (c) \tc{UniqueEss2}.
\end{theorem}

\begin{proof}
(a) If $\varphi\Ess x$ at $w$ but not $\varphi x w$, the essence clause for ${\sim}\varphi$
(which $x$ has) gives that $\varphi$ is non-empty and, at every successor of $w$, no actual
thing has $\varphi$; by \tc{AllExist} everything is actual, so $\varphi$ has no instance at any
successor, against Lemma~\ref{lem:reach}.  (b) Let $\varphi,\psi$ be essences of $x$ at
$w$, $w\,r\,v$, and $\varphi y v$ for an actual $y$.  The essence clause for
$\lambda z.\ z\equiv x$ gives $y\equiv x$ at $v$.  By (a) $\psi x w$, by \tc{MC\_K}
$\psi x v$; and $y$ has the object-language property $\lambda z.\ \psi z\supset\psi y$
trivially, so by Leibniz identity $x$ has it too, which with $\psi x v$ gives $\psi y v$;
symmetrically.  (c) Were
$\varphi z t\wedge\neg\psi z t$ for some $z$, $t$, the constant property ``$z$ has $\varphi$
and not $\psi$'' would be non-empty, so by Lemma~\ref{lem:reach} exemplified at some
successor $u$ of $w$ --- against (b) at $u$, where \tc{AllExist} makes $z$ actual.  Hence
$\varphi = \psi$, and $\Box(\varphi\equiv\psi)$ follows.
\end{proof}

Question~2 is thereby answered in the object language: \tc{Th3}, \tc{UniqueEss1} and
\tc{UniqueEss2} hold in \textsf{S4}, and indeed in \textsf{K} --- the last two as the
\emph{Notes} expected them to be, though not from the essence clause alone: what proves them is
\tc{AllExist} and \tc{MC\_K}, the collapse variant~3's definitions force.
In \textsf{S5} the \emph{Notes} report
\tc{UniqueEss1} proved and \tc{UniqueEss2} refuted by Nitpick for variant~2
(\NFig[, line~40]{7}).  In variant~3 both hold, and the single world explains why: an essence of $x$ can only
be ``being $x$''.  Theorem~\ref{thm:v3}(d) shows this with the one world and its nominal;
Theorem~\ref{thm:uniqueess-k} shows it without either, in the form (a)--(c).

\subsection{The possibilist and mixed-quantifier copies}
\label{sec:copies}
In the possibilist theories Lemma~\ref{lem:l2} and Theorem~\ref{thm:gex} hold as stated, with
$\exists$ for $\exists\actE$; in variant~3 the non-emptiness clause of $\chi\supset_N{\sim}G$
holds because $\chi$ is true at the world where no God-like being is assumed to exist.  So \tc{Th3} holds in \textsf{K}
in both variants.

In the mixed setting \cite[Footnote~20]{J75}, $\mathit{ConjOfPropsFrom}$
and necessary existence quantify actualistically and the inclusion possibilistically.  There the
same $\chi$ and $\Phi$, with possibilist $\exists$, still form a conjunction in the sense of
$\mathit{ConjOfPropsFrom}$, and the inclusion admits the step to ${\sim}G$: a God-like being,
possibly not actual, exists at every world.  That
being is all Lemma~\ref{lem:reach} and the modal collapse need: its necessary existence,
actualist here, yields an actual
witness, and Lemma~\ref{lem:serial} the successor.  The argument with $\exists\actE$ does not go
through there, since the
possibilist inclusion would require $\chi$ to exclude God-like beings that are not actual.

Theorem~\ref{thm:uniqueess-k} then holds in
both settings with no frame condition, and part~(a) becomes simpler, since with a possibilist
inclusion \tc{AllExist} is not needed.  Neither copy's proof uses \tc{Th4}, which the \emph{Notes}
postulate in both theories, a prover having reported a proof that no trusted tactic replayed;
in the mixed setting it is a theorem of \tc{Ax1Gen} and \tc{Ax2a}
alone, by the empty and the singleton instance that Section~\ref{sec:emendation} examines.  All these proofs are hybrid-free.  Resting on \tc{Ax1Gen} as stated, they carry the
reservation recorded at the end of Section~\ref{sec:gex}; under \tc{Ax1GenTwo} the copies are
re-established in sessions \tc{M}--\tc{P}, with no side condition and for
variant~3 in its emended form (Section~\ref{sec:corpus}).

Both sections have taken \tc{Ax1Gen} as the \emph{Notes} state it.  On it rest \tc{G\_ex} and
everything drawn from it, the hybrid-free answers to question~2 included; the statements about
the frame, \tc{NoNecExist} and \tc{OneWorld} of this section, need no conjunction axiom at all.  Whether
the axiom may be taken as stated is the question of the next section.

\section{Emendations of the conjunction axiom \tc{Ax1Gen}}
\label{sec:emendation}

The hybrid-free answers of Sections~\ref{sec:gex} and~\ref{sec:v3} pass through the empty
instance of \tc{Ax1Gen}, so what a conjunction axiom should say is the next question; note that
emending G\"odel's axioms has a literature of its
own~\cite{Anderson1990,AndersonGettings1996,Hajek2002,Fitting2002}.  Three requirements
on such an axiom:
\begin{enumerate}
\item[(i)] It yields \tc{L}: being God-like is positive.
\item[(ii)] With \tc{Ax2a} alone it does not yield \tc{Th4}, the possible existence of a
  God-like being.  In G\"odel's argument that step belongs to \tc{Ax4} (\tc{Th4} from
  \tc{Ax2a}, \tc{Ax4}, \tc{L} in the \emph{Notes}).
\item[(iii)] It need not block modal collapse, which comes with the maximality of the
  positive properties~\cite{MuehlenbeckBenzmueller2026,Benzmueller2026comment} and which
  Kova\v{c} argues G\"odel intended~\cite{Kovac2012}.
\end{enumerate}

G\"odel's footnote asks for conjunctions of ``any number of summands''.  Read so as to
include the empty set and the singletons, it gives \tc{Ax1Gen}, and that reading fails (ii).
Suppose no God-like being were possible.  The conjunction of the singleton $\{G\}$ is $G$
itself, and \tc{Ax1Gen} makes positive whatever is necessarily coextensive with it; under the
supposition nothing actual is God-like at any successor, so that is the empty property, and it
is positive.  The conjunction of $\emptyset$ is the universal property, which \tc{Ax1Gen} makes
positive with nothing to check.  But the universal property is the complement of the empty one,
and \tc{Ax2a} allows at most one of a property and its complement to be positive.  So a God-like
being is possible after all: \tc{Th4} follows from \tc{Ax1Gen} and \tc{Ax2a} alone, and the two
offending instances are exactly the degenerate ones.  That derivation is how \tc{Th4} is proved in five modules of
the \lean{} port, where the set handed to the axiom is written out in each proof term:
variant~2 in the actualist and possibilist settings, and variant~3 in all
three.  The mixed-quantifier setting of variant~2 reaches \tc{Th4} through symmetry instead,
and Scott's variant from \tc{A1}, \tc{A2} and his \tc{A3}.  Where the \emph{Notes}' variant-2
theories prove \tc{Th4} themselves they use \tc{Ax4} besides, and in variant~3 they postulate
it.  Requirement~(ii) is tested by putting a reading of the conjunction axiom as a
hypothesis of \tc{Th4}, in a theory whose only axiom is \tc{Ax2a}, so that the readings are
compared with nothing else to draw on.  Of the readings stated below, \tc{Ax1Gen} and
\tc{Ax1GenBox} yield \tc{Th4} there, and
\tc{Ax1GenOne} and \tc{Ax1GenTwo} do not (session \tc{I}, in both systems).  The combination
of the last two restrictions inherits the answer from \tc{Ax1GenTwo}
(Section~\ref{sec:boxtwo}).

The natural restriction is to require
a conjunction to have at least two different conjuncts:
\begin{multline*}
\tc{Ax1GenTwo}:\quad \lfloor(\exists\psi_1\,\psi_2.\,\Phi\psi_1 \wedge \Phi\psi_2 \wedge
  \psi_1\neq\psi_2\\
  {}\wedge \mathit{PosProps}\,\Phi \wedge \mathit{ConjOfPropsFrom}\,\varphi\,\Phi)
  \supset \Pos\varphi\rfloor.
\end{multline*}
This is the reading adopted here, and it is the only change the \emph{Notes}' theories undergo,
with the one exception Section~\ref{sec:plain} announces:\footnote{No claim is made that G\"odel formulated this restriction: neither the 1970 footnote nor his own earlier axiom places any condition on the conjunction (Section~\ref{sec:setting}).  It is a non-degenerate reconstruction, motivated by the wording of the footnote and by the behaviour of the generalised principle under the two instances it licenses.} \tc{Ax1}, \tc{Ax2a},
\tc{Ax2b}, \tc{Ax3} and \tc{Ax4} stand as they state them, and so do the
definitions.  \tc{Ax1} is in fact redundant under it: with
$\Phi=\{\varphi,\psi\}$ the restricted axiom gives the conjunction of a pair, and for
$\varphi=\psi$ that conjunction is $\varphi$ itself, so the instance is trivial
(\tc{Ax1\_of\_Two}).  The two-conjunct case is what
\tc{Ax1} says; what the restriction removes are the degenerate instances.  Either \tc{Ax1Gen} is left alone, and the answers of
Section~\ref{sec:gex} rest on an instance the footnote need not license, or it is
restricted, and they have to be established again --- which is what this section does.
Four neighbouring readings are weighed beside \tc{Ax1GenTwo}, as controls on what the
restriction is doing (Section~\ref{sec:neighbours}):
\tc{Ax1GenOne}, which excludes only the empty conjunction; \tc{Ax1GenBox}, which leaves
the cardinality of $\Phi$ alone and requires instead that its members be positive at every
successor as well;
\tc{Ax1GenR}, which takes $\Phi$ to be a set of properties that does not vary with the world;
and \tc{Ax1GenBoxTwo}, which imposes the restrictions of \tc{Ax1GenTwo} and \tc{Ax1GenBox} at
once.  All five are weakenings of \tc{Ax1Gen}.  What holds below is proved, in \isa{} and in
\lean{}; what fails is refuted by a countermodel Nitpick~\cite{Nitpick} finds, recorded as an
\tc{expect} annotation that the build checks --- and the countermodel for variant~3 is written out in
\lean{} besides, and verified by its kernel.

\subsection{\tc{Ax1GenTwo}: at least two different conjuncts}
\label{sec:ax1gentwo}
\label{sec:two}

Under the two-conjunct reading the first thing to recover is \tc{L}, since the derivation of it
in the \emph{Notes} takes $\Phi:=\Pos$, as G\"odel's own Theorem~1 does
(Section~\ref{sec:setting}), and says nothing about how many members that set has.

\begin{lemma}[\tc{L\_Two}]
\label{lem:ltwo}
In variant 2, \tc{Ax1GenTwo} yields $\val{\Pos G}$ with \tc{Ax2a}, \tc{Ax3} and \tc{Ax4},
hybrid-free.
\end{lemma}

\begin{proof}
Fix $w$.  The universal property is positive at $w$, by \tc{Ax3} and \tc{Ax4}.  If some
further property is positive at $w$, then $\Phi:=\Pos$ has two different members there, and
its conjunction is $G$ by the definition of $G$, so \tc{Ax1GenTwo} applies.  Otherwise the
universal property is the only positive one at $w$; then every positive property holds of
everything, so $G$ does, so ${\sim}G$ is not positive at $w$, and $\Pos G$ follows by
\tc{Ax2a}.
\end{proof}

\tc{Th4} then follows exactly as in the \emph{Notes}, from \tc{Ax2a}, \tc{Ax4} and \tc{L},
with no frame condition.  Requirement~(ii) holds: with \tc{Ax1GenTwo} and \tc{Ax2a} as the only
assumptions, Nitpick refutes both \tc{Th4} and the positivity of the universal property, and
reports each countermodel as genuine, Nitpick's verdict for a model it has checked itself
rather than one its approximations may have produced (session \tc{I}, the theory
described at the head of this section).

The restriction does not hand the argument back to the frame: the collapse of
Section~\ref{sec:gex} survives it, at first under a proviso on necessary existence, and the
rest of this section removes even that.

\begin{theorem}[\tc{G\_ex\_Two\_at}]
\label{thm:gextwo}
In variant 2, under \tc{Ax1GenTwo},
$\val{(\exists\psi.\ \Pos\psi \wedge \psi\neq\lambda x.\top) \supset \exists\actE x.\,G x}$: at
every world at which some property other than the universal one is positive, an actual God-like
being exists.  The proof uses
\tc{Ax1GenTwo}, \tc{Ax2a}, \tc{Ax3}, \tc{Ax4} and no frame condition, and is hybrid-free.
\end{theorem}

\begin{proof}
Let $\psi$ be positive at $w$ with $\psi\neq\lambda x.\top$, and suppose no actual God-like
being exists at $w$.  Put
\begin{align*}
\chi &= \lambda z.\ \neg(\exists\actE y.\,G y) \wedge \psi z,\\
\Phi &= \lambda\psi'.\ \bigl(\neg(\exists\actE y.\,G y) \wedge
   (\psi' = \lambda x.\top \vee \psi' = \psi)\bigr)\\
 &\qquad {}\vee \bigl((\exists\actE y.\,G y) \wedge \psi' = \lambda x.\bot\bigr).
\end{align*}
At $w$ the set $\Phi$ has the two members $\lambda x.\top$ and $\psi$, which are different,
and both are positive there.  At a successor $v$: if an actual God-like being exists at $v$,
then $\Phi$ is $\{\lambda x.\bot\}$ there and its conjunction holds of nothing, as does
$\chi$; if not, then $\Phi$ is $\{\lambda x.\top,\ \psi\}$ and its conjunction is $\psi z$,
which is $\chi z$.  So $\chi$ is necessarily the conjunction of $\Phi$, and \tc{Ax1GenTwo}
makes $\chi$ positive at $w$.  From there the proof of Theorem~\ref{thm:gex} closes it
unchanged, with Lemma~\ref{lem:ltwo} in place of \tc{L}.
\end{proof}

\tc{Ax1GenTwo} asks for two different members of $\Phi$ at the world of evaluation only, not
under the $\Box$ of $\mathit{ConjOfPropsFrom}$: at the successors with an actual God-like
being $\Phi$ is the singleton $\{\lambda x.\bot\}$, and nothing depends on it.  This is where
the world-dependence of $\Phi$ continues to work.

At the worlds the theorem does not reach, the universal property is the only positive one,
and there every individual is God-like: $G x$ says that $x$ has every positive property
(\tc{all\_God\_of\_only\_top}).

\begin{corollary}
\label{cor:gextwo}
$\val{\exists\actE x.\,G x}$ --- and with it \tc{Th3}, \tc{Th5} and \tc{MC} as in
Theorem~\ref{thm:consequences}, with no frame condition --- holds under either of two
conditions: that necessary existence is not the universal property, $E\neq\lambda x.\top$,
since then $E$ is by \tc{Ax3} a witness for the theorem at every world
(\tc{G\_ex\_Two\_of\_E}); or that every world has an actual individual, which covers the
remaining worlds (\tc{G\_ex\_Two}).  Both proofs are hybrid-free.
\end{corollary}

What is left is a world at which the universal property is the only positive one and nothing
is actual.  There every individual is God-like and none of them exists, so the conclusion
would fail.  No such world exists, and \tc{Ax2a} and \tc{Th4} suffice to show it.  By \tc{Ax2a}
every property at such a world is $\lambda x.\top$ or $\lambda x.\bot$
(\tc{OnlyTopDichotomy}), so in particular the property $\lambda x.\,\exists\actE y.\,G y$, which
does not depend on its individual argument.  If it is $\lambda x.\top$, an actual God-like being
exists at the world, and the world has an actual individual.  If it is $\lambda x.\bot$, no
world has an actual God-like being, against \tc{Th4}, which gives every world a successor with
one (\tc{Th4\_Two}).  Hence, under \tc{Ax1GenTwo}, every world has an actual individual
(\tc{ae\_Two}), and the corollary holds with no condition (\tc{G\_ex\_Two\_unconditional}); so
do \tc{Th3}, \tc{Th5} and \tc{MC} (\tc{Th3\_Two\_unconditional}, \tc{Th5\_Two\_unconditional},
\tc{MC\_Two\_unconditional}).  The argument is hybrid-free: it substitutes a property of the
object language and nothing else.

The same world was first excluded here with a nominal, and that argument is kept for what only
it gives.  Applying \tc{Ax2a} to $\lambda x\,t.\ t = u$ forces a single world (\tc{OneWorld});
in a single world with nothing actual the universal property is an essence of every individual,
its inclusion clause being vacuous, so necessary existence is empty, against \tc{Ax3}
(\tc{only\_top\_and\_empty\_domain\_False}); the collapse, applied to $\lambda u.\ u = w$, then
makes the accessibility relation the identity (\tc{R\_Id}), and with it the full \tc{Ax1Gen}
returns (Section~\ref{sec:plain}).  The last two steps have no hybrid-free counterpart here;
Section~\ref{sec:open} gives one.

One case thus carried everything conditional in this section, and it is the same case
throughout: a world whose actual domain is empty.  It is the world the two conditions of
Corollary~\ref{cor:gextwo} guard against, and the model that refutes \tc{Th4} in variant~3 under
the two-conjunct reading (Proposition~\ref{prop:v3two}) has nothing actual at its one world.  In
variant~2 the axioms exclude it, as just shown; in variant~3 as the \emph{Notes} state it they do
not, \tc{Th4} having gone with the empty instance, and Section~\ref{sec:plain} draws the
consequence.  Where the quantifiers are possibilist the case cannot arise, the type of
individuals being non-empty.

In variant 3 the two-conjunct reading is genuinely weaker, and it costs the \emph{Notes} a
theorem they prove there.

\begin{proposition}
\label{prop:v3two}
In variant 3, \tc{Ax1GenTwo} yields neither the positivity of the universal property, nor
\tc{Th4}, nor \tc{G\_ex}, nor \tc{Ax1Gen} --- not even in \textsf{S4}.
\end{proposition}

\begin{proof}
One world, one individual, nothing actual at that world, and the empty property positive and
no other.  All axioms of variant 3 hold, and so does \tc{Ax1GenTwo}, vacuously: two
\emph{different} positive properties cannot exist, $\lambda x.\bot$ being the only positive
one.  The frame is universal, hence reflexive and transitive.  The universal property is not
positive, no God-like being possibly exists, and \tc{Ax1Gen} fails, since its empty instance
would make the universal property positive.
\end{proof}

Nitpick finds a countermodel to each of the four statements (session \tc{F}, variant~3 under
\tc{Ax1GenTwo}), and to \tc{Th4} one that is reflexive and transitive; the model described
here, in which all four fail and the frame is universal, is written out and machine-checked in
\lean{} as the module \tc{Ax1GenTwoVariant3Model}.  The detector of
Section~\ref{sec:criterion} reads this section's results in both systems and flags no
nominal.  Lemma~\ref{lem:ltwo} and Theorem~\ref{thm:gextwo} use none in session \tc{E} or in its
\lean{} module, their witnesses being the universal property, $E$, $G$, ${\sim}G$ and $\Pos$
throughout.  The model module has no \isa{} counterpart, the model there being
Nitpick's, and what the detector flags in the \lean{} module, as foreign to $\mathcal{L}$, are the model's own
definitions of $\Pos$ and $G$ --- semantic objects, which is what a model consists of.
Under the two-conjunct reading \tc{Th4} is therefore not available from the axioms of
variant~3 at all, and the \emph{Notes}' own proof of it there, from \tc{Ax2a}, \tc{L} and
\tc{Ax1Gen}~\cite[Sect.~4.5]{J75}, is one of the derivations the restriction removes.  Their
AFP source states \tc{Th4} as an axiom there, no trusted tactic having replayed the provers'
proof~\cite[Footnote~30]{J75}; the postulate is what lets the development carry on, \tc{Th5},
the positivity of the universal property and \tc{MC} being proved from it in the lines that
follow.

This is also where the three open questions are answered for good.
Section~\ref{sec:gex} answered them from \tc{Ax1Gen}, through the empty conjunction, which
is the instance this section rejects.  Under \tc{Ax1GenTwo} the answers have to be given
again, variant by variant.  For variant~2, question~1: being God-like is positive by
Lemma~\ref{lem:ltwo}; a God-like being possibly exists, since \tc{Th4} follows from \tc{Ax2a},
\tc{Ax3} and \tc{Ax4} once \tc{L} is available; and one actually exists at every world at
which some positive property is not the universal one, by Theorem~\ref{thm:gextwo}; hence
\tc{Th3}, \tc{Th5} and \tc{MC} under the conditions of Corollary~\ref{cor:gextwo}.  Scott's
variant, question~3, behaves the same way.

\begin{theorem}[\tc{Ax1GenTwoScott}]
\label{thm:scott-two}
In Scott's variant with \tc{Ax1GenTwo} in place of \tc{A3}:
\begin{enumerate}
\item[(a)] $\val{\Pos G}$, from \tc{A1}, \tc{A2} and \tc{A5} (\tc{L\_Two});
\item[(b)] an actual God-like being exists at every world at which some positive property is
  not the universal one (\tc{G\_ex\_Two\_at}), and hence at every world once every world has an
  actual individual (\tc{G\_ex\_Two});
\item[(c)] every world does have an actual individual, by the argument that follows
  Corollary~\ref{cor:gextwo}, which transfers with \tc{A1} for \tc{Ax2a} and \tc{Coro} for
  \tc{Th4} (\tc{ae\_Two}); so (b) holds outright (\tc{G\_ex\_Two\_unconditional}), and with it
  \tc{T3} and \tc{MC} (\tc{T3\_Two\_unconditional}, \tc{MC\_Two\_unconditional}).
\end{enumerate}
All with no frame condition and no side condition, and hybrid-free.
\end{theorem}

\begin{proof}
As Lemma~\ref{lem:ltwo} and Theorem~\ref{thm:gextwo}, with \tc{A2} for \tc{Ax4}, \tc{A1} for
\tc{Ax2a}, and \tc{A5} with \tc{A2} for the positivity of the universal property.
\end{proof}

Variant~3, question~2, is different, since Proposition~\ref{prop:v3two} takes \tc{Th4} away.
Taking \tc{Th4} as a hypothesis instead brings the answers back; Section~\ref{sec:plain} shows
that the hypothesis can be dispensed with.

\begin{theorem}[\tc{Ax1GenTwoVariant3Th4}]
\label{thm:v3-two}
In variant~3 with \tc{Ax1GenTwo}, and \tc{Th4} postulated:
\begin{enumerate}
\item[(a)] $\val{\Pos G}$ (\tc{L\_Two}), from \tc{Ax1GenTwo}, \tc{Ax2a}, \tc{Ax3}, \tc{Ax4}
  and \tc{Th4}; the universal property is positive because $E$ is non-empty by \tc{Th4}.
\item[(b)] Suppose that at some world with an actual God-like being some individual is not
  God-like~($\dagger$).  Then an actual God-like being exists at every world
  (\tc{G\_ex\_Two}), hence \tc{Th3}, \tc{Th5}, \tc{MC}, \tc{AllExist}, Lemma~\ref{lem:reach},
  \tc{UniqueEss1} and \tc{UniqueEss2} as in Theorem~\ref{thm:uniqueess-k} --- none of them
  with a frame condition, by Lemma~\ref{lem:serial}.  All of it is hybrid-free.
\item[(c)] Without ($\dagger$), \tc{OneWorld} of Theorem~\ref{thm:v3}(b), which uses no
  conjunction axiom, and \tc{Th4} give \tc{G\_ex} outright (\tc{G\_ex\_Two\_hybrid}) --- with
  nominals.
\end{enumerate}
\end{theorem}

\begin{proof}
(a) as Lemma~\ref{lem:ltwo}.  (b) The construction of Theorem~\ref{thm:gextwo} does not
transfer as it stands: the inclusion $\chi\supset_N{\sim}G$ now requires $\chi$ to be
non-empty, and its $\chi$ --- $\lambda z.\ \neg(\exists\actE y.\,G y)\wedge\psi z$, with an
arbitrary positive $\psi\neq\lambda x.\top$ --- need not be.  Take instead the constant property
$\chi = \lambda z.\ \neg\exists\actE y.\,G y$, which is non-empty at a world without an actual
God-like being, and as the second member of $\Phi$ beside $\lambda x.\top$ the property
$\psi_s = \lambda z.\ G z \vee \neg\exists\actE y.\,G y$.  It is positive by \tc{Ax4} from
$G$, since $G$ is non-empty by \tc{Th4}; it is the universal property at every world without
an actual God-like being, so that the conjunction of $\{\lambda x.\top, \psi_s\}$ is $\chi$
there, while $\Phi = \{\lambda x.\bot\}$ where one exists; and it differs from $\lambda x.\top$
exactly by ($\dagger$).  The rest is as before.  The remaining statements are the proofs of
Section~\ref{sec:v3} with \tc{G\_ex\_Two} for \tc{G\_ex}.  (c) is immediate.
\end{proof}

The side condition ($\dagger$) fails only if every individual is God-like wherever a God-like
being actually exists.  It holds, together with all axioms of variant~3, \tc{Ax1GenTwo} and
\tc{Th4}, in a one-world model with two individuals, one of them God-like and the other not,
which Nitpick finds (session \tc{K}, variant~3 with \tc{Th4} postulated).  Whether it can be dropped from the hybrid-free proof is
left open.

So the answers stand under the two-conjunct reading.  Questions~1 and~3 are answered
outright, since the axioms exclude the one world the two conditions of
Corollary~\ref{cor:gextwo} guard against (\tc{ae\_Two}).  Question~2 is answered either with
the possible existence of a God-like being assumed, which \tc{Ax1GenTwo} no longer supplies, and
hybrid-free under~($\dagger$), or outright once the non-emptiness clause of \NFig{8} is confined to
the essence (Section~\ref{sec:plain}).  The reading excludes the degenerate instances as
instances of the axiom, and the witnesses here are accordingly two-membered; but a singleton
padded with the universal property is two-membered too, and the witness of
Corollary~\ref{cor:gextwo} is such a padding --- $\{\lambda x.\top,\psi\}$ at the world of
evaluation, with $\lambda x.\top$ contributing nothing to the conjunction
(Section~\ref{sec:assessment}).  The one instance that does not return bought nothing
the answers need: where no second member is to be had, \tc{Th4} and \tc{Ax2a} do its work
(\tc{ae\_Two}).  Whether it returns too, so that the two readings are equivalent, is known only
with a nominal (Section~\ref{sec:plain}).

\subsection{Variant 3: where the non-emptiness clause belongs}
\label{sec:plain}

Proposition~\ref{prop:v3two} is not the last word on variant~3.  The \emph{Notes} repair
G\"odel's original essence by requiring the implying property to be non-empty, and they state
that repair inside the abbreviation for necessary property inclusion, which \tc{Ax4} then
inherits.  Nothing about \tc{Ax4} calls for it, and it is exactly what costs variant~3 its
theorems: $E\supset_N\lambda x.\top$ needs $E$ non-empty, which the two-conjunct reading no
longer supplies, so the step that makes the universal property positive in variant~2, from
\tc{Ax3} and \tc{Ax4}, is unavailable.

The clause cannot simply be dropped either.  Its purpose is to block the inconsistency of
G\"odel's 1970 axioms, which arises because without it the empty property necessarily implies
every other and is an essence of everything~\cite[Sects.~4.2, 4.5]{J75}; take it out of the
essence and variant~3 is variant~1 again.  It is needed in the essence, then, and the question is whether
it is needed anywhere else.  Confining it to the essence, and letting \tc{Ax4} use the inclusion
of variant~2, restores everything, \tc{Th4} included, as a theorem
(Figure~\ref{fig:emended}).

What that produces is a new theory, and it should be called one: variant~3$'$.  The
\emph{Notes} ask about
variant~3 as they state it, with one inclusion used in both places; variant~3$'$ keeps their
definitions and their axioms but gives \tc{Ax4} a different inclusion, and no reading of the
conjunction axiom turns the one theory into the other.  So the answer below is an answer about
variant~3$'$.  It bears on the \emph{Notes}' question in two ways.  Variant~3$'$'s \tc{Ax4}
yields theirs, the clause being an added conjunct under the same $\Box$, so every theorem of
their variant~3 still holds under the conjunction axiom it was proved from, and nothing they
prove is given up.  And what the two-conjunct reading costs their variant~3 ---
\tc{Th4}, and with it the rest --- it does not cost variant~3$'$.  The question whether their
variant~3 answers question~2 under that reading with \tc{Th4} a theorem rather than an
assumption stays as Proposition~\ref{prop:v3two} leaves it.

\begin{figure}[t]
\centering
\fbox{\begin{minipage}{\dimexpr\textwidth-2\fboxsep-2\fboxrule\relax}
\noindent A predicate
$\Pos$ on properties, two inclusions kept apart, and the definitions of variant~3:
\begin{align*}
\varphi\supset_P\psi &= \Box(\forall\actE y.\ \varphi y\supset\psi y)
  &&\text{plain inclusion, for \tc{Ax4}}\\
\varphi\supset_N\psi &= \Box\bigl(\varphi\neq(\lambda x.\bot) \wedge
  \forall\actE y.\ \varphi y\supset\psi y\bigr)
  &&\text{inclusion, $\varphi$ non-empty, for \tc{Ess}}\\
G\,x &= \forall\varphi.\ \Pos\varphi\supset\varphi x &&\text{God-like}\\
\varphi\Ess x &= \forall\psi.\ \psi x\supset(\varphi\supset_N\psi)
  &&\text{essence, G\"odel's form, inclusion repaired}\\
E\,x &= \forall\varphi.\ \varphi\Ess x\supset\Box\exists\actE y.\ \varphi y
  &&\text{necessary existence}
\end{align*}
The axioms, all under $\val{\cdot}$:
\begin{align*}
\tc{Ax2a}&:\ \Pos\varphi \vee^{e} \Pos{\sim}\varphi\\
\tc{Ax2b}&:\ \Pos\varphi\supset\Box\Pos\varphi\\
\tc{Ax3}&:\ \Pos\,E\\
\tc{Ax4}&:\ \Pos\varphi \wedge (\varphi\supset_P\psi) \supset \Pos\psi\\
\tc{Ax1GenTwo}&:\ (\exists\psi_1\,\psi_2.\,\Phi\psi_1 \wedge \Phi\psi_2 \wedge \psi_1\neq\psi_2\\
&\qquad {}\wedge \mathit{PosProps}\,\Phi \wedge \mathit{ConjOfPropsFrom}\,\varphi\,\Phi)
  \supset \Pos\varphi
\end{align*}
with $\mathit{PosProps}\,\Phi = \forall\psi.\ \Phi\psi\supset\Pos\psi$ and
$\mathit{ConjOfPropsFrom}\,\varphi\,\Phi = \Box(\forall\actE z.\ \varphi z \leftrightarrow
\forall\psi.\ \Phi\psi\supset\psi z)$ as in Section~\ref{sec:setting}.  The clause of \NFig{8}
occurs in the essence and nowhere else; \tc{Ax1} is left out, being a consequence
(\tc{Ax1\_of\_Two}), though the sources keep it, where the \emph{Notes} have it; and no condition
is placed on the accessibility relation.
\end{minipage}}
\caption{The emended variant~3 in full: everything it assumes, in one place.}
\label{fig:emended}
\end{figure}

\begin{theorem}[\tc{Ax1GenTwoVariant3Plain}]
\label{thm:plain}
In variant~3$'$, under \tc{Ax1GenTwo}:
\begin{enumerate}
\item[(a)] the universal property is positive, from \tc{Ax3} and \tc{Ax4} and no conjunction
  axiom, and with it $\val{\Pos G}$ and \tc{Th4}, exactly as in variant~2
  (Lemma~\ref{lem:ltwo});
\item[(b)] an actual God-like being exists at every world at which some positive property is
  not the universal one (Theorem~\ref{thm:gextwo}), and by the argument that follows
  Corollary~\ref{cor:gextwo} every world has an actual individual (\tc{ae\_Two}), so that one
  exists at every world outright (\tc{G\_ex\_Two\_unconditional});
\item[(c)] hence \tc{Th3}, \tc{Th5} and \tc{MC}, with no frame condition and no side condition
  (\tc{MC\_Two\_unconditional}); the side condition~($\dagger$) of Theorem~\ref{thm:v3-two} does
  not arise, \tc{Th4} being a theorem;
\item[(d)] and, again with no frame condition --- Lemma~\ref{lem:serial} supplying the
  successor --- and no side condition, \tc{AllExist},
  Lemma~\ref{lem:reach}, \tc{EssMemI\_K}, \tc{UniqueEss1} and \tc{UniqueEss2} as in
  Theorem~\ref{thm:uniqueess-k}.
\end{enumerate}
All four parts are hybrid-free, \tc{Th4} is a theorem here rather than a hypothesis, and all
are proved in both systems.
\end{theorem}

\begin{proof}
The proofs of Section~\ref{sec:ax1gentwo} transfer unchanged, since they use \tc{Ax4} only
through the plain inclusion; G\"odel's essence, \tc{Th1} and necessary existence are untouched,
and the proofs of Theorem~\ref{thm:uniqueess-k} go through with the God-like being that (b)
supplies.
\end{proof}

The axioms are consistent, and consistent with \tc{Ax1GenTwo}: Nitpick finds a genuine model of
all of them together (\tc{consistency\_Two}); the one it reports has one world and two
individuals, both actual, with the universal property positive and, of the two remaining
properties, which are each other's complements, exactly one positive.  What
requirement~(ii) objects to is that the conjunction axiom should deliver the positivity of the
universal property and the closure of positivity under necessary coextension; here both come
from \tc{Ax3} and \tc{Ax4}, as they do in variant~2, which the \emph{Notes} accept.  So
variant~3$'$ answers question~2 under the two-conjunct reading --- not with \tc{Th4} assumed,
as Theorem~\ref{thm:v3-two} assumes it to see what follows, but by putting the clause
of \NFig{8} where G\"odel's
definition needs it, and at the price of being a theory the \emph{Notes} do not state.

What the emendation does not do, once nominals are admitted, is keep the degenerate instances
out.  The world (b) excludes was first excluded with a nominal: applying \tc{Ax2a} to
$\lambda x\,t.\ t = u$ forces a single world (\tc{OneWorld}), and there the universal property
is an essence of every individual, so necessary existence is empty, against \tc{Ax3}
(\tc{only\_top\_and\_empty\_domain\_False}, no conjunction axiom being used).  What only that
route gives is what follows: the collapse applies to the nominal $\lambda u.\ u = w$ and makes the
accessibility relation the identity (\tc{R\_Id}), and from that \tc{Ax1GenTwo} yields the full
\tc{Ax1Gen} (\tc{Ax1Gen\_of\_Two}), so that the two readings are equivalent here
(\tc{Two\_equivalent}); both are proved in both systems.
Once the accessibility relation is the identity, the members of $\Phi$ at a successor are its
members at the world of evaluation, and a conjunction over one member or over none is discharged
by \tc{Ax4}'s plain inclusion, from that member's positivity or from the universal property's.
Every instance of \tc{Ax1Gen} is available again, the two degenerate ones included: with
nominals admitted, the restriction has no force here.

It is a two-layer result of the kind Section~\ref{sec:v3} met.  The derivation is not one of the
modal object language: \tc{OneWorld} and \tc{R\_Id} are nominal arguments; the
equivalence has a hybrid-free proof besides (Section~\ref{sec:open}).  Requirement~(ii) is untouched, since it asks what
a conjunction axiom yields with \tc{Ax2a} and nothing else, where the readings do differ
(session \tc{I}).  What gives way is not the requirement but the ambition behind it.  The same
argument, using only \tc{Ax2a}, \tc{Ax3} and \tc{Ax4}, goes through in variant~2, equivalence
included (session \tc{E}).  So the two-conjunct reading is a reading of the object language
only.

In the object language variant~3$'$ is unblemished: \tc{Th4} is a theorem there, the answers to
all three questions are hybrid-free and need neither a frame condition nor a side condition, and
the argument under which the restriction loses its force uses a nominal.  It keeps the essence of
1970 as G\"odel wrote it and repairs the inclusion instead, where variant~2 adds a conjunct to
the essence --- a conjunct G\"odel had written in what Kanckos and Lethen conjecture to be the
latest of his earlier versions, and did not write in 1970 (Lemma~\ref{lem:esscompare}, Section~\ref{sec:setting}).  That he meant the
inclusion to carry the repair is not borne out by his earlier versions: his necessary implication there
is the plain one, and the one earlier version that has an essence has the conjunct in it.  The
move is the \emph{Notes}', and variant~3$'$ is this article's; between them they are the least
that must change for the text of 1970 to prove what it sets out to prove.

\subsection{The \emph{Notes}' results under \tc{Ax1GenTwo}}
\label{sec:corpus}
Under \tc{Ax1GenTwo} --- in variant~2 as the \emph{Notes} state it and in variant~3$'$ ---
nothing of the \emph{Notes} is lost in the actualist setting: every fact of the port that
consumes \tc{Ax1Gen} --- 49, by \tc{\#print axioms} --- is re-established outright, in variant~2
and in the emended variant~3 alike (Appendix~\ref{sec:accounts} reconciles that figure with the
\isa{} account).  The possibilist and mixed-quantifier copies are re-established in four further
\isa{} sessions (\tc{M}--\tc{P}, each with its \lean{} counterpart), with local repairs and two
losses, one in each setting of the emended variant~3; the proofs need less rather than more, the
collapse results doing without the condition \tc{ae} of Corollary~\ref{cor:gextwo} in all four
settings (\tc{ae\_Two}).

What is lost is \tc{AllExist}, that every entity is actual at every world
($\val{\forall y.\ \mathit{existsAt}\,y}$, Theorem~\ref{thm:v3}(c)): it goes in the two settings
of the emended variant~3 whose inclusion is possibilist, and nowhere else, for the same reason in
both.  The proof of
Section~\ref{sec:plain} takes the essence $\lambda z.\ z = y \wedge \neg\mathit{existsAt}\,z$ of a
God-like being and relies on its inclusion clause being vacuous \emph{because that clause
quantifies actualistically}; where the clause is possibilist it is no longer vacuous, and Nitpick
gives a genuine countermodel.  It cannot be repaired there, $\mathit{existsAt}$ occurring in no
axiom of that setting.  Nothing downstream depends on it, the actuality side conditions it
discharged having gone with the same change, so \tc{EssMemI\_K} and the uniqueness results go
through more easily.  In the mixed-quantifier setting the \emph{Notes}' own \tc{Th4} goes
through symmetry (Section~\ref{sec:audit}); under
\tc{Ax1GenTwo} it needs no frame condition, a merely possible God-like being at every world and
the actualist reading of necessary existence doing the work that symmetry did.

\subsection{Neighbouring readings}
\label{sec:neighbours}

\subsubsection{\tc{Ax1GenOne}: at least one conjunct}\label{sec:one}  This reading,
\[
\val{(\exists\psi.\,\Phi\psi \wedge \mathit{PosProps}\,\Phi \wedge
  \mathit{ConjOfPropsFrom}\,\varphi\,\Phi) \supset \Pos\varphi},
\]
excludes only the empty set --- the move that variant~3 makes for property inclusion, whose
implying property must be non-empty. \tc{L} holds ($\Phi:=\Pos$ is non-empty by \tc{Ax3}); with \tc{Ax4},
$\Pos(\lambda x.\top)$ and \tc{Th4} follow as in the \emph{Notes}.  Requirement (ii) holds
as well, with genuine countermodels to \tc{Th4} and to $\Pos(\lambda x.\top)$ from
\tc{Ax1GenOne} and \tc{Ax2a} as the only assumptions.  The collapse survives with no proviso at all: the set of
Lemma~\ref{lem:l2} with $\lambda x.\top$ added is non-empty and positive at every world,
so \tc{G\_ex} follows (\tc{G\_ex\_One}) and with it Theorem~\ref{thm:consequences}.  In
variant~3 no one-individual model refutes anything, and with a nominal everything follows.  The
one-world argument of Theorem~\ref{thm:v3}(b), with Lemma~\ref{lem:serial}, gives
reflexivity, and with it
$\Pos(\lambda x.\top)$, $\neg\Pos(\lambda x.\bot)$, \tc{AllExist}, \tc{G\_ex} and \tc{Th4}
(\tc{PTop\_One\_deriv} to \tc{Th4\_One\_deriv}), and indeed the full \tc{Ax1Gen}
(\tc{One\_gives\_S}).  That last step is the singleton instance closing positivity under
necessary actualist coextension, which with an empty actual domain reaches every property.  So the searches recorded
in session \tc{H} could not have succeeded.  Hybrid-free the question stays open, and that is the
same two-layer situation as Section~\ref{sec:plain}: what separates the readings in variant~3
separates them in the object language only.

\subsubsection{\tc{Ax1GenBox}: members positive at every successor}\label{sec:box}  This reading,
\[
\val{(\mathit{PosProps}\,\Phi \wedge \Box\mathit{PosProps}\,\Phi \wedge
   \mathit{ConjOfPropsFrom}\,\varphi\,\Phi) \supset \Pos\varphi},
\]
attacks the world-dependence instead of the cardinality.  Where a God-like being exists, the
set $\Phi$ of Section~\ref{sec:gex} contains $G$ and ${\sim}G$, and ${\sim}G$ is never
positive where $G$ is.  Requiring the members to be positive at every successor is meant to block
that proof, and it does.  In variant~2, \tc{L}, \tc{Th4} and \tc{Th5}/\tc{MC} in \textsf{KB}
follow as in the \emph{Notes}, while \tc{G\_ex} and \tc{MC} in \textsf{K} and \tc{Th3},
\tc{Th5} in \textsf{S4} have two-world countermodels.  Scott's variant behaves the same
way, so that the countermodels of \cite[Sect.~5.3]{J75} return.  But it fails (ii), for a
reason
that is hard to repair.  The added premiss $\Box\mathit{PosProps}\,\{G\}$ is
$\Box\,\Pos\,G$, which \tc{L} supplies at every world, so \tc{Th4} follows from
\tc{Ax1GenBox} and \tc{Ax2a} alone (\tc{Th4\_rigid}, in both systems; and again as
\tc{Th4\_Box} in session \tc{I}, where the reading enters as a hypothesis).  In variant~3 the
axiom is idle: the base axioms force one
world (Theorem~\ref{thm:v3}(b)), and in a one-world model \tc{Ax1GenBox} gives \tc{Ax1Gen}
back (\tc{Ax1Gen\_of\_refl}), hence \tc{Th5\_S4} and \tc{MC\_S4}, which inherit the nominals of
\tc{OneWorld}.  Lemma~\ref{lem:serial} applies here too, so both of those hold already in
\textsf{K}; the hypotheses and the names are left as they stand, the \textsf{S4} framing being
the question under discussion.

\subsubsection{\tc{Ax1GenR}: rigid sets}\label{sec:rigid}  The world-dependence can also be
attacked by requiring $\Phi$ to be a set of properties \emph{simpliciter}, of type
$(e\to\sigma)\to o$ rather than $(e\to\sigma)\to\sigma$:
\[
\val{(\mathit{PosProps}_{\mathrm{R}}\,\Phi \wedge \mathit{Conj}_{\mathrm{R}}\,\varphi\,\Phi)
  \supset \Pos\varphi},
\]
where the two premisses are those of \tc{Ax1Gen} read at that type, $\Phi\psi$ being a truth
value and not a proposition.  With \tc{L} added, the trivialising consequences do not follow:
\tc{G\_ex} and \tc{MC} in \textsf{K} and \tc{Th3} in \textsf{S4} have countermodels (session
\tc{B}).  Whether \tc{L} itself follows was recorded there as open, the derivation of
Section~\ref{sec:ax1gentwo} being unavailable because $\Pos$ is not rigid.  It does follow, in
\textsf{K}, from \tc{Ax2a} and \tc{Ax2b}, with $\Pos$ rigidified at the world of evaluation,
$\{\psi\mid\Pos\,\psi\,w\}$ --- the satisfaction operator $\mathsf{@}_w\Pos$
(\tc{L\_R}).\footnote{This reading, and \tc{L\_R}, were pointed out in a pre-submission review of
this article.}  So the reading does not escape the hybrid question; it relocates it, and the
detector flags that witness as it flags any other
(Section~\ref{sec:criterion}).

What is asked of a rigid set is that it be given by a condition formed in $\mathcal{L}$
(Definition~\ref{def:hybridfree}), and one such set serves.  The globally positive properties,
$\{\psi\mid\mathsf{A}(\Pos\,\psi)\}$, written with the primitive equality in which
$\Pos\,\psi = \top$ is $\mathsf{A}(\Pos\,\psi)$, give \tc{L} from no axiom of the argument at
all, as \tc{\#print axioms} confirms, provided positivity is the same set of properties at every
world (\tc{L\_R\_rigidpos}).  The detector accepts that set in either spelling, with the
primitive equality or as a quantifier over worlds, the two being one modality
(Appendix~\ref{app:detector}).  \tc{Ax2b} and \tc{P\_succ} fix positivity across every pair of worlds the
relation connects, in either direction, so that proviso holds wherever the frame is weakly
connected (\tc{RigidPos\_of\_link}).  It does not follow from the axioms, with or without the
reading assumed, so a countermodel has to leave the frame not weakly connected; unrelatedness
alone would not do, two worlds that both see a third agreeing on positivity all the same
(\tc{pos\_of\_common\_succ}).  Nor does the proviso trivialise the reading: the three
countermodels above survive it (session \tc{B}).  The object language has two world-free
descriptions of positivity, positive at every world and positive at some world, and both need
the proviso; and the one rigid set that needs none, $\Phi:=\{G\}$, makes the premiss of the
axiom its own conclusion.  What is at issue
here is therefore a condition and not a witness: whether \tc{L} follows under this reading with
no condition on positivity at all.  It does (Section~\ref{sec:open}).

\subsubsection{\tc{Ax1GenBoxTwo}: both restrictions at once}\label{sec:boxtwo}  The two restrictions act on different
dimensions --- the cardinality of $\Phi$ and its world-dependence --- and taking them together
meets all three requirements \emph{and} returns the frame:
\[
\begin{aligned}
\lfloor(\exists\psi_1\,\psi_2.\,&\Phi\psi_1 \wedge \Phi\psi_2 \wedge \psi_1\neq\psi_2 \wedge
  \mathit{PosProps}\,\Phi\\
&{}\wedge \Box\mathit{PosProps}\,\Phi \wedge \mathit{ConjOfPropsFrom}\,\varphi\,\Phi)
  \supset \Pos\varphi\rfloor.
\end{aligned}
\]
Requirement~(i) holds: \tc{L} follows from \tc{Ax2a}, \tc{Ax3} and \tc{Ax4} with no frame
condition, by the argument of Lemma~\ref{lem:ltwo} --- with $\Phi:=\Pos$ the added box premiss
$\Box\mathit{PosProps}\,\Pos$ is trivially true (\tc{L\_BoxTwo}).  The other two requirements
need no new model search, because the axiom carries more premisses than either restriction and
is therefore implied by each (\tc{BoxTwo\_of\_Box}, \tc{BoxTwo\_of\_Two}).  Every model of
\tc{Ax1GenTwo} is a model of it, so the countermodels of session \tc{I} carry over and (ii)
holds.  Every model of \tc{Ax1GenBox} is a model of it, so the countermodels of session \tc{A}
carry over --- \tc{G\_ex} and \tc{MC} in \textsf{K}, \tc{Th3} and \tc{Th5} in \textsf{S4}.  Under
this reading, then, the collapse is not derivable in \textsf{K} or in \textsf{S4}, and the
\textsf{S4} question has the answer the \emph{Notes} expected.  In \textsf{KB} it is derivable, as
it is under every tabulated reading (\tc{MC\_BoxTwo}); the rigid reading is not tabulated,
\tc{L} there needing a satisfaction operator or positivity world-independent.\footnote{The combination was pointed out in a pre-submission review of this article.  It
costs what \tc{Ax1GenBox} costs: $\Box\mathit{PosProps}$ is a premiss nothing in G\"odel's
footnote suggests, so what recommends the combination is the three requirements and not the
text.}

In variant~3 and variant~3$'$ the combination behaves as the two-conjunct reading
does, with one question left open.  In variant~3 it costs what that reading costs: the one-world
model of Proposition~\ref{prop:v3two} is a model of \tc{Ax1GenTwo} and hence of the combination,
so \tc{Th4}, \tc{G\_ex} and $\Pos(\lambda x.\top)$ fail there too (Nitpick finds this in session \tc{F};
\tc{Ax1GenBoxTwo} holds in \tc{Ax1GenTwoVariant3Model} besides).  The collapse is
untouched, needing no conjunction axiom at all (Theorem~\ref{thm:v3}(b)).  In variant~3$'$,
\tc{L} and \tc{Th4} hold as in Theorem~\ref{thm:plain}(a), the box premiss being trivial for
$\Phi:=\Pos$ (\tc{L\_BoxTwo}, \tc{Th4\_BoxTwo}, from \tc{Ax2a}, \tc{Ax3} and \tc{Ax4} with no
frame condition).  The collapse holds there too, but only with a nominal.  The construction of
Theorem~\ref{thm:gextwo} does not transfer, since its $\Phi$ is $\{\lambda x.\bot\}$ at a
successor with an actual God-like being and $\lambda x.\bot$ is not positive there; what survives
is the route of Section~\ref{sec:v3}, which uses no conjunction axiom at all and is therefore open
to every reading of one.  \tc{NoNecExist} rests on the definitions and reflexivity alone, and the
step from it to a single world on \tc{Ax2a}, \tc{Ax3} and \tc{Ax4} --- here with the plain
inclusion, which makes that step easier than in variant~3 (\tc{NoNecExist},
\tc{OneWorld\_abs}, \tc{MC\_OneWorld}).  Whether \tc{G\_ex} follows hybrid-free is open: Nitpick
finds no countermodel and gives up (session \tc{L}).

\subsection{Assessment}
\label{sec:assessment}

Against the three requirements, and against the collapse, the readings stand as
Table~\ref{tab:readings} shows.

\begin{table}[t]
\centering\footnotesize\setlength{\tabcolsep}{3pt}
\caption{The readings of the conjunction axiom against the three requirements, against the
  collapse, and against the axiom they emend.}
\label{tab:readings}
\begin{tabular}{llllll}
\toprule
 & \tc{Ax1Gen} & \tc{Ax1GenBox} & \tc{Ax1GenOne} & \tc{Ax1GenTwo} & \tc{Ax1GenBoxTwo}\\
\midrule
(i) \tc{L}                        & holds & holds & holds & holds & holds\\
(ii)                              & fails & fails & holds & holds & holds\\
degenerate instances excluded     & none & none & empty only$^\S$ & both$^\P$ & both$^\P$\\
weaker than \tc{Ax1Gen}, variant 3 & --- & (no)$^{\mathrm{a}}$ & (no)$^{\mathrm{a}}$ & yes & yes\\
premiss beyond the footnote        & none & $\Box\mathit{PosProps}$ & none & none
                                   & $\Box\mathit{PosProps}$\\
\midrule
collapse, variant 2               & derivable & refuted; der.\ in \textsf{KB}
                                  & derivable & derivable & refuted; der.\ in \textsf{KB}\\
collapse, variant 3               & derivable & (derivable)$^{\mathrm{a}}$ & (derivable)$^{\mathrm{a}}$ & (derivable)$^{\mathrm{a}}$ & (derivable)$^{\mathrm{a}}$\\
collapse, variant 3$'$            & derivable & (derivable)$^{\mathrm{a}}$ & derivable & derivable & (derivable)$^{\mathrm{a}}$\\
\tc{Th4}, variant 2               & derivable & derivable & derivable & derivable & derivable\\
\tc{Th4}, variant 3               & derivable & derivable & (derivable)$^{\mathrm{a}}$ & refuted$^\ddagger$ & refuted$^\ddagger$\\
\tc{Th4}, variant 3$'$            & derivable & derivable & derivable & derivable & derivable\\
\bottomrule
\end{tabular}
\end{table}
The first two rows answer the requirements of this section, the next three what each reading
takes away from \tc{Ax1Gen} and what it adds; the rest report derivability from
the axioms alone.  Verdicts are for \textsf{K} unless a logic is named and are hybrid-free unless marked:
a verdict in parentheses, marked $^{\mathrm{a}}$, passes through the one-world argument of
Theorem~\ref{thm:v3}(b) and is therefore not one of the object language, where the question stays
open; $^\ddagger$ a refutation that needs a single individual, two distinguishable ones giving
$\Pos(\lambda x.\top)$ from \tc{Ax2a} and \tc{Ax4} alone; $^\S$ \tc{Ax1GenOne} excludes
$\Phi=\emptyset$ but licenses $\Phi=\{\lambda x.\top\}$, which has the same conjunction
(\tc{Conj\_empty\_iff\_top}) and is the instance the derivation of Section~\ref{sec:gex} needs
(\tc{G\_ex\_of\_Singleton}); $^\P$ as instances of the axiom: the singleton instance returns
by padding (below).

The fourth row is for variant~3 because that is where the
cardinality readings come apart: \tc{Ax1GenTwo} is separated from \tc{Ax1Gen} by a genuine model
there (\tc{Two\_strictly\_weaker}, session \tc{F}), \tc{Ax1GenOne} implies it
(\tc{One\_gives\_S}), and \tc{Ax1GenBoxTwo} inherits the separation from \tc{Ax1GenTwo}
(\tc{BoxTwo\_of\_Two}).  In variant~2 the two coincide (\tc{Two\_equivalent}; hybrid-free,
Section~\ref{sec:open}).  Rows~(i)
and~(ii) are for variant~2: requirement~(ii) is tested in a theory whose only axiom is \tc{Ax2a}
and does not turn on the variant, whereas \tc{L} in variant~3 holds only with \tc{Th4}
postulated (Theorem~\ref{thm:v3-two}(a)) or with nominals (session \tc{F}).  Variant~3$'$ is the emendation of
Section~\ref{sec:plain}, with the clause of \NFig{8} confined to the essence, and there
\tc{Th4} is a theorem under every reading; in the collapse the readings keep the differences they
have in variant~3, except that \tc{Ax1GenTwo} and \tc{Ax1GenOne} reach it hybrid-free and with no side
condition (Theorem~\ref{thm:plain}(c)).  Only its \tc{Ax1GenTwo} entries are separate results
(Theorem~\ref{thm:plain}): \tc{Ax1GenOne} inherits them, licensing every instance
\tc{Ax1GenTwo} licenses; \tc{Ax1Gen} inherits them because variant~3$'$ proves every theorem of variant~3;
and \tc{Ax1GenBox} by the argument that makes it idle in variant~3, the axioms forcing one
world.  In variant~3 the collapse is hybrid-free and frame-free under \tc{Ax1Gen}
(\tc{MC\_K}), while under \tc{Ax1GenBox} it passes through the one-world argument of
Theorem~\ref{thm:v3}(b), with nominals.  With \tc{Th4} taken as a hypothesis in variant~3 instead,
Theorem~\ref{thm:v3-two} applies.

The combined reading
(Section~\ref{sec:neighbours}) is entered only where its verdicts follow from the transfers:
being implied by each of the two restrictions, it inherits every countermodel either of them
has, and it keeps \tc{L}, \tc{Th4} and the \textsf{KB} collapse
(\tc{L\_BoxTwo}, \tc{Th4\_BoxTwo}, \tc{MC\_BoxTwo}).  Its entries for variant~3 and
variant~3$'$, and the one question left open there, are those of Section~\ref{sec:neighbours}; the
collapse survives in variant~3 under every reading, since \tc{OneWorld} uses no conjunction
axiom at all.  The rigid reading is not tabulated:
of its verdicts only \tc{L} --- with a satisfaction operator, or with a witness of the object
language once positivity is world-independent --- and the countermodels of session \tc{B} are
known.

Both cardinality restrictions meet the three requirements in variant~2, and of the two the
two-conjunct reading is the one the footnote's wording favours.  It differs from the other in
dropping the singleton instance --- but only as an instance of the conjunction axiom.  What that
instance says is the closure of positivity under necessary coextension, and that follows from
\tc{Ax4} alone, coextension being inclusion both ways (\tc{Singleton\_of\_Ax4}).  The instance
itself returns under \tc{Ax1GenTwo} as soon as \tc{Ax3} and \tc{Ax4} are there: a singleton
padded with the universal property has two different members and the same conjunction.  The one
case the padding misses is the singleton whose member is that property itself, where the padding
adds nothing (\tc{Singleton\_of\_Two}, and \tc{Singleton\_of\_Two\_at} where $\Phi$ is a
singleton only at the world of evaluation, which is the case the derivations here use).  That
case is the empty instance under another name, the two having the same conjunction
(\tc{Conj\_empty\_iff\_top}).  So the two degenerate instances leave one residue between them,
and that is all the reading takes away in variant~2.  And that residue is not idle.  Granted a $\Phi$ with one member at every world ---
$\lambda x.\top$ where no actual God-like being exists, $\lambda x.\bot$ where one does ---
the derivation of Section~\ref{sec:gex} goes through verbatim with that set in place of the
empty one, so that \tc{G\_ex} is unconditional in \textsf{K} and hybrid-free
(\tc{G\_ex\_of\_Singleton}).  The world where that instance would be needed, the universal
property being the only positive one and no second member to be had, is the world
Corollary~\ref{cor:gextwo} guards against, and under \tc{Ax1GenTwo} it does not exist
(\tc{ae\_Two}); so the answers do without the instance.  The instance itself does return, so that the two readings are equivalent: with a nominal here
(\tc{Ax1Gen\_of\_Two}, \tc{Two\_equivalent}), hybrid-free in Section~\ref{sec:open}.

Why not the one-conjunct reading, then?  On Table~\ref{tab:readings} it meets the three
requirements as well, and being the stronger of the two it licenses every instance
\tc{Ax1GenTwo} licenses, so it inherits Theorem~\ref{thm:plain} entire and gives \tc{Th4} in
variant~3, with a nominal, where \tc{Ax1GenTwo} refutes it.  But its surplus is exactly the
residue just described.  The \emph{Notes}' own derivation of \tc{G\_ex} needs the singleton
$\{\lambda x.\top\}$ and nothing more (\tc{G\_ex\_of\_Singleton}); that reading licenses it, so
the derivation of Section~\ref{sec:gex} runs under it verbatim, \tc{G\_ex} unconditional in
\textsf{K} and hybrid-free (\tc{G\_ex\_One}, Section~\ref{sec:one}).  And in variant~3 the
restriction is void, \tc{Ax1GenOne} implying the full \tc{Ax1Gen} (\tc{One\_gives\_S}), where
\tc{Ax1GenTwo} is separated from it by a genuine model (\tc{Two\_strictly\_weaker}, session
\tc{F}).  Requirement~(ii) is a test with \tc{Ax2a} alone, and the one-conjunct reading passes it
while keeping, under another name, the instance the test was devised to exclude.  What
recommends it is a precedent in the definitions of variant~3 --- excluding only the empty set is
what that variant asks of property inclusion --- and against that stands the footnote's word: a
conjunction has summands to conjoin.  That word is the ground of the choice between the two
cardinality readings, and it should be seen for what it is.  On requirements~(i)--(iii) the two
readings are level, and in variant~3$'$ they prove the same theorems; what separates them is
the wording of the footnote and whether the degenerate instances remain instances of the axiom.
This article declares that ground --- the wording --- and does not derive it.

And why not the combination?  It meets the three requirements as well, and it is the only reading
that does so while giving the \textsf{S4} question an answer: removing the world-dependence
of $\Phi$ refutes the collapse in \textsf{K} and in \textsf{S4}, whereas under \tc{Ax1Gen} and
\tc{Ax1GenTwo} \tc{Th3} holds in \textsf{K} already and \textsf{S4} is idle
(Table~\ref{tab:verdicts}).  But it displaces the collapse rather than removing it:
\tc{MC\_BoxTwo} assumes symmetry, so in \textsf{KB} --- contained in the \textsf{S5} of
G\"odel's own argument --- the collapse is derivable under the combination as under every
tabulated reading, and in variant~3 it needs no conjunction axiom at all.  It also answers less.
Being implied by each restriction (\tc{BoxTwo\_of\_Box}, \tc{BoxTwo\_of\_Two}) it is the
weakest of the three, and in variant~3$'$ it leaves \tc{G\_ex} open, the construction of
Theorem~\ref{thm:gextwo} not transferring (Section~\ref{sec:boxtwo}) and Nitpick finding no
countermodel (\tc{G\_ex\_BoxTwo}, session \tc{L}), whereas \tc{Ax1GenTwo} settles it with no
condition and hybrid-free (Theorem~\ref{thm:plain}(b)).  And requirement~(iii) asks none of
this: the collapse comes with the maximality of the positive
properties and not with the conjunction axiom, so a reading that blocks it answers a charge this
diagnosis does not lay at the axiom's door.  The combination is the reading to prefer if what one wants vindicated is the \emph{Notes}'
expectation about \textsf{S4}; \tc{Ax1GenTwo} is the one to prefer if what one wants followed is the text of 1970.

The strength of \tc{Ax1Gen} thus has two sources, and they come apart.  The
derivation that violates (ii) needs both degenerate instances, so excluding the empty one is
already enough, and both cardinality readings do that (session \tc{I}); there the padding is not
to be had either, \tc{Ax2a} alone not making the universal property positive.  The
world-dependence of $\Phi$ is what produces the collapse, and the readings that remove it are
the three that attack it rather than the cardinality: \tc{Ax1GenBox}, \tc{Ax1GenR} and their
combination with the two-conjunct restriction (Section~\ref{sec:neighbours}).  The price
is~(ii) for \tc{Ax1GenBox}, and for the combination, which meets~(ii), the premiss
$\Box\mathit{PosProps}\,\Phi$ that G\"odel's footnote does not ask for; and the removal helps
only where the collapse does not already come
from the definitions, as it does in variant~3.  Table~\ref{tab:instances} reads both sources
off the \lean{} proofs: which set each result hands to the conjunction axiom, and which it
reaches through a lemma.  Table~\ref{tab:verdicts} then collects what each
of the three questions receives under each reading, and on what condition.

\begin{table}[tp]
\centering\small\setlength{\tabcolsep}{3pt}
\caption{The set each result hands to the conjunction axiom, read off the \lean{} proofs,
  where it stands written out in the source, by \tc{tools/ax1gen\_instances.py}: $\bullet$ where the proof applies the axiom to it,
  $\circ$ where it reaches it through a lemma it uses.  ``Empty'' and ``singleton'' are the two
  degenerate instances of requirement~(ii); ``world-dep.'' is a set whose membership turns on
  whether a God-like being exists at the world of evaluation; ``all positive'' is the set of
  positive properties, the instance that gives \tc{L} and is not degenerate.  The other
  quantifier settings agree with the ones shown.}
\label{tab:instances}
\begin{tabular}{llcccc}
\toprule
 & & \multicolumn{4}{c}{set handed to the axiom}\\
\cmidrule(l){3-6}
theory & result & empty & singleton & world-dep. & all positive\\
\midrule
\multicolumn{6}{l}{\emph{\tc{Ax1Gen}, as the \emph{Notes} state it}}\\
\tc{GoedelVariantHOML2} & \tc{L} & --- & --- & --- & $\bullet$\\
\tc{GoedelVariantHOML2} & \tc{Th4} & $\bullet$ & $\bullet$ & --- & $\circ$\\
\tc{GoedelVariantHOML2} & \tc{Th5} & $\circ$ & $\circ$ & --- & $\circ$\\
\tc{GoedelVariantHOML2} & \tc{MC} & $\circ$ & $\circ$ & --- & $\circ$\\
\tc{GoedelVariantHOML2} & \tc{PisUFilter} & $\circ$ & $\circ$ & --- & $\circ$\\
\tc{GoedelVariantHOML3} & \tc{Th4} & $\bullet$ & $\bullet$ & --- & $\circ$\\
\tc{GoedelVariantHOML3} & \tc{PisUFilter} & $\circ$ & $\circ$ & --- & $\circ$\\
\tc{ThereIsNoEvil2} & \tc{NecNoEvil} & $\bullet$ & --- & --- & ---\\
\tc{GoedelVariantHOML2inS4oneFile} & \tc{G\_ex} & --- & --- & $\bullet$ & $\circ$\\
\tc{GoedelVariantHOML2inS4oneFile} & \tc{Th3} & --- & --- & $\circ$ & $\circ$\\
\tc{GoedelVariantHOML2inS4oneFile} & \tc{MC} & --- & --- & $\circ$ & $\circ$\\
\midrule
\multicolumn{6}{l}{\emph{\tc{Ax1GenOne}: at least one conjunct}}\\
\tc{Ax1GenOneVariant2} & \tc{L\_One} & --- & --- & --- & $\bullet$\\
\tc{Ax1GenOneVariant2} & \tc{Th4\_One} & --- & --- & --- & $\circ$\\
\tc{Ax1GenOneVariant2} & \tc{G\_ex\_One} & --- & --- & $\bullet$ & $\circ$\\
\tc{Ax1GenOneVariant2} & \tc{MC\_One} & --- & --- & $\circ$ & $\circ$\\
\midrule
\multicolumn{6}{l}{\emph{\tc{Ax1GenTwo}: at least two different conjuncts}}\\
\tc{Ax1GenTwoVariant2} & \tc{L\_Two} & --- & --- & --- & $\bullet$\\
\tc{Ax1GenTwoVariant2} & \tc{Th4\_Two} & --- & --- & --- & $\circ$\\
\tc{Ax1GenTwoVariant2} & \tc{G\_ex\_Two\_at} & --- & --- & $\bullet$ & $\circ$\\
\tc{Ax1GenTwoVariant2} & \tc{MC\_Two} & --- & --- & $\circ$ & $\circ$\\
\tc{Ax1GenTwoVariant3Plain} & \tc{L\_Two} & --- & --- & --- & $\bullet$\\
\tc{Ax1GenTwoVariant3Plain} & \tc{G\_ex\_Two\_at} & --- & --- & $\bullet$ & $\circ$\\
\tc{Ax1GenTwoVariant3Plain} & \tc{MC\_Two} & --- & --- & $\circ$ & $\circ$\\
\midrule
\multicolumn{6}{l}{\emph{\tc{Ax1GenBox}: members positive at every successor}}\\
\tc{Ax1GenBoxVariant2} & \tc{L} & --- & --- & --- & $\bullet$\\
\tc{Ax1GenBoxVariant2} & \tc{PTop} & $\bullet$ & --- & --- & ---\\
\tc{Ax1GenBoxVariant2} & \tc{Serial} & $\bullet$ & --- & --- & ---\\
\tc{Ax1GenBoxVariant2} & \tc{Th4\_rigid} & $\circ$ & $\bullet$ & --- & $\circ$\\
\tc{Ax1GenBoxVariant2} & \tc{MC} & --- & --- & --- & $\circ$\\
\bottomrule
\end{tabular}

\end{table}

\begin{table}[tp]
\centering\small\setlength{\tabcolsep}{3pt}
\caption{The three questions of the \emph{Notes}, answered reading by reading.  $\star$ marks
  the reading this article adopts (Section~\ref{sec:assessment}).  Verdicts are
  for \textsf{K} unless a logic is named; ``free'' means no frame condition and no side
  condition; a dash stands where the entry is a refutation, for which the hybrid-free question
  does not arise.  In the reading column a comma lists readings and a parenthesis names the
  variant; in the verdict column a semicolon separates what holds under different suppositions,
  and the cells to its right are read in the same order.  ($\dagger$) is the side condition of
  Theorem~\ref{thm:v3-two}, and ``none with a nominal'' that it is discharged, but only by a
  hybrid argument.
  Variant~3$'$ is the emendation of Section~\ref{sec:plain}.  Readings not studied for a question
  are not listed, and the rigid reading is not tabulated (Section~\ref{sec:rigid}).}
\label{tab:verdicts}
\begin{tabular}{>{\raggedright\arraybackslash}p{2.5cm}>{\raggedright\arraybackslash}p{3.3cm}>{\raggedright\arraybackslash}p{2.6cm}>{\raggedright\arraybackslash}p{1.3cm}>{\raggedright\arraybackslash}p{1.9cm}}
\toprule
reading & verdict & condition & hybrid-free & where\\
\midrule
\multicolumn{5}{l}{\emph{Question 1 (variant 2): \tc{Th3} in \textsf{S4}}}\\
\tc{Ax1Gen} & holds in \textsf{K} & free (empty instance) & yes & Thm.~\ref{thm:consequences}\\
\tc{Ax1GenOne} & holds in \textsf{K} & free & yes & Sect.~\ref{sec:one}\\
$\star$~\tc{Ax1GenTwo} & holds in \textsf{K} & free & yes & Cor.~\ref{cor:gextwo}\\
\tc{Ax1GenBox}, \tc{Ax1GenBoxTwo} & refuted in \textsf{S4}; holds in \textsf{KB} & --- & --- & Sects.~\ref{sec:box}, \ref{sec:boxtwo}\\
\midrule
\multicolumn{5}{l}{\emph{Question 2 (variant 3): \tc{Th3} in \textsf{S4}; \tc{UniqueEss1}, \tc{UniqueEss2} in \textsf{S5}}}\\
\tc{Ax1Gen} & all three hold in \textsf{K} & free (empty instance) & yes & Sect.~\ref{sec:v3}, Thm.~\ref{thm:uniqueess-k}\\
no conjunction axiom & \tc{Th3} holds and there is one world & reflexivity & no & Thm.~\ref{thm:v3}\\
\tc{Ax1GenTwo} (variant~3 as stated) & \tc{Th4} refuted; with \tc{Th4} postulated, all three hold in \textsf{K} & ($\dagger$); none with a nominal & yes; no & Prop.~\ref{prop:v3two}, Thm.~\ref{thm:v3-two}\\
$\star$~\tc{Ax1GenTwo} (variant~3$'\!$) & \tc{Th4} a theorem and all three hold in \textsf{K} & free & yes & Thm.~\ref{thm:plain}\\
\tc{Ax1GenOne} (variant~3$'\!$) & as \tc{Ax1GenTwo}, by inheritance & free & yes & Sect.~\ref{sec:assessment}\\
\tc{Ax1GenOne}, \tc{Ax1GenBox}, \tc{Ax1GenBoxTwo} (variant~3) & hold (through the one-world argument) & reflexivity & no & Sects.~\ref{sec:one}, \ref{sec:box}, \ref{sec:boxtwo}\\
\midrule
\multicolumn{5}{l}{\emph{Question 3 (Scott's variant with \tc{Ax1Gen} for \tc{A3}): \tc{T3} in \textsf{S4}}}\\
\tc{Ax1Gen} & holds in \textsf{K} & free (empty instance) & yes & Thms.~\ref{thm:gex-scott}, \ref{thm:consequences}\\
$\star$~\tc{Ax1GenTwo} & holds in \textsf{K} & free & yes & Thm.~\ref{thm:scott-two}\\
\tc{Ax1GenBox} & refuted in \textsf{S4} (the \emph{Notes}' countermodels return) & --- & --- & Sect.~\ref{sec:box}\\
\bottomrule
\end{tabular}
\end{table}
In every row the \textsf{S4} question has one of two fates.  Where the world-dependence of
$\Phi$ is kept, the frame drops out and the answer is yes, in \textsf{K}, free under \tc{Ax1Gen}
and under \tc{Ax1GenTwo} alike --- except in variant~3 as the \emph{Notes} state it, where
\tc{Ax1GenTwo} needs \tc{Th4} postulated and the side condition~($\dagger$).  Where that dependence is
removed, \textsf{S4} becomes a question with an answer, and the answer is the countermodel the
\emph{Notes} expected.  Every yes has a hybrid-free proof; only the removal of~($\dagger$)
needs a nominal.

The section's recommendation is therefore \tc{Ax1GenTwo}, with the clause of \NFig{8} confined
to the essence.  Under that reading the answers of Sections~\ref{sec:gex} and~\ref{sec:v3} hold
for all three questions --- in \textsf{K}, hybrid-free, with neither a frame condition nor a side
condition --- and neither degenerate instance is licensed.  No other reading keeps the answers
while excluding both.  \tc{Ax1Gen} keeps them, but licenses both degenerate instances and reaches
the answers through one of them.  \tc{Ax1GenOne} excludes the empty conjunction and licenses
$\{\lambda x.\top\}$, which is that instance renamed.  \tc{Ax1GenBox} excludes neither and
fails~(ii) besides.  The combination excludes both and meets~(ii), but gives question~1 the
opposite answer: \tc{Th3} is refuted in \textsf{S4} under it, as under \tc{Ax1GenBox}, which
refutes question~3 there as well --- the countermodels the \emph{Notes} expected
(Section~\ref{sec:box}).  Both of those readings carry the premiss
$\Box\mathit{PosProps}\,\Phi$, which G\"odel's footnote does not ask for and which is what
takes away the world-dependence of $\Phi$ that the answers rest on.  The rigid
reading is not in the comparison, nothing being known of its verdicts on the three questions
beyond \tc{L} and the countermodels of session \tc{B} (Section~\ref{sec:rigid}).

\tc{Ax1GenTwo} costs the \emph{Notes} nothing in the actualist setting
(Section~\ref{sec:corpus}), and the text of 1970 less than the \emph{Notes}' own repair does:
the clause enters the essence alone, where their repair puts it into \tc{Ax4} besides.  The degenerate
instances do return under the recommended reading: in variant~2 and in variant~3$'$ \tc{Ax1GenTwo}
and the unrestricted axiom are equivalent (\tc{Ax1Gen\_of\_Two}, \tc{Two\_equivalent}), proved through
\tc{OneWorld} and \tc{R\_Id} here (Section~\ref{sec:plain}) and in the object language by Lindheim
(Section~\ref{sec:open}), whereas in variant~3 as the \emph{Notes} state it a genuine model
separates them (\tc{Two\_strictly\_weaker}, session \tc{F}).  No model finder over this embedding
could have settled the hybrid-free question either way: the implication is a theorem, so there
was no countermodel to find (Section~\ref{sec:criterion-open}).

\section{The audit of the \emph{Notes}}
\label{sec:audit}

The three questions are settled.  What the criterion of Section~\ref{sec:criterion} adds ---
that every term a proof substitutes at a type of Definition~\ref{def:hybridfree} is a term of
$\mathcal{L}$ --- is a check on the \emph{Notes}' development as a whole.  Their proofs were
found by automation over a shallow embedding, and nothing in that setting requires a witness to
stay inside $\mathcal{L}$; whether any left it is a question about those proofs, not about the
statements.  Two runs of the detector answer it; Appendix~\ref{app:runs} says how each
development was read and gives the figures in full.

The \isa{} run reads the \emph{Notes} themselves: the 30 theories of their session, rebuilt with
proof recording, which gives 427 facts once the axioms and the definitional equations are
counted in.  The \textsc{nominal} flag fires on none of them: nowhere in the \emph{Notes} does a
proof substitute a nominal where the criterion looks.  Of the 427 facts, 313 are certified
outright.  The 114 that are not carry, without exception, a term the prover's reconstruction
substituted rather than one the argument needs --- in 112 a Skolem or choice term of an \tc{smt}
or \tc{metis} replay, in two a constant of HOL's library, listed by kind in
Appendix~\ref{app:runs} --- and the explicit proofs of the same statements, in the port and in
the cross-check (session \tc{IsarAudit}), are certified every one.  So every statement the
\emph{Notes} prove has a proof \isa{} certifies: its own for 313 of them, one written for the
purpose for the other 114.  An earlier version of this article gave 403 and 24: the reading
of the \isa{} detector then counted, besides the developments' own facts, the witnesses of a
list of quantifier rules, and passed over the lemmas the automated provers leave in a theory,
under which almost every prover's term stands (Appendix~\ref{app:runs}).

The \lean{} run reads the port, whose proofs are written out explicitly against the same
statements: 30 modules, 399 facts, of which 294 are the proofs themselves and 105 the equations
\lean{} generates for the definitions; the 294 are the results the \emph{Notes} prove, five of
them stated there with a proof that an external prover had reported and that no trusted tactic replayed.  No proof in the port is flagged, not even for the
existence predicate, and none depends on \tc{sorry}; its 252 counted instantiations are every
one a term of $\mathcal{L}$.  This matters for the dependency accounting of the port: what
\tc{\#print axioms} reports there are dependencies of arguments conducted in the modal logic,
not artefacts of the embedding.  The two runs are about different proofs, the originals and
proofs written for the purpose, and each is complete over its own; where they differ is in how
much each records, and that tracks how a proof was made, not which system made it.

\subsection{\tc{Th4} in \textsf{K}: a frame condition the \emph{Notes} do not need}
\label{sec:th4k}

One flag the runs raised, against a proof of this article rather than one of the \emph{Notes}',
led to a result in the modal logic, at the one place where the \emph{Notes} record a frame
condition they do not need.  The \emph{Notes} and the port prove the possible existence of a God-like
being, \tc{Th4}, in the mixed-quantifier setting of the second variant
(\tc{GoedelVariantHOML2AndersonQuant}) with symmetry (\cite[Footnote~25]{J75}: the proof
``requires symmetry of the accessibility relation (or modal axiom B), as an additional
dependency'').
A first proof that symmetry can be dispensed with used a nominal, and went by contradiction:
suppose no
successor of $w$ carries an actual God-like being.  By \tc{Ax1Gen}
some successor $v$ of $w$ carries a possibilist God-like being.  The property ``being at a
successor of $w$'', $\lambda x\,u.\ w\,r\,u$ --- that is, $\Diamond^{-1}\mathsf{w}$ --- is
positive at $w$ (\tc{Ax4} from the universal property), and hence at $v$ (\tc{Ax2b}).  No
successor of $v$ is a successor of $w$, else \tc{Th2} would place an actual God-like being
there, against the supposition.  So at $v$, \tc{Ax4} turns that property into the empty one, against \tc{Ax2a}.
The detector flags the witness \textsc{access} and \textsc{rigid} (\tc{Th4\_K\_hybrid} in the sources).
The flag prompted the search for a hybrid-free proof, and there is one, simpler than the first.

\begin{proposition}
\label{prop:th4mixed}
In the mixed-quantifier setting of variant~2, \tc{Th4} holds in \textsf{K}: from \tc{Ax1Gen},
\tc{Ax2a}, \tc{Ax2b}, \tc{Ax3} and \tc{Ax4}, with no frame condition and with a hybrid-free proof
(\tc{Th4\_K}).
\end{proposition}

\begin{proof}
The argument of Section~\ref{sec:gex}, run with possibilist quantifiers, yields a
possibilist God-like being at $w$ \emph{itself}. By \tc{Ax3} and \tc{Th1} it has necessary
existence, so \tc{Th2} places actual God-like beings at every successor of $w$.  And $w$
has a successor, since at a dead end \tc{Ax1Gen} makes every property positive, against
\tc{Ax2a} --- Lemma~\ref{lem:serial} is the same step with \tc{Ax3} and \tc{Ax4} in place of
\tc{Ax1Gen}.
\end{proof}

The proof (\tc{Th4\_K}) consumes \tc{Ax1Gen}, \tc{Ax2a}, \tc{Ax2b}, \tc{Ax3},
\tc{Ax4}, no frame condition, and only object-language witnesses, and it is checked in both
systems, in \tc{GoedelVariantHOML2AndersonQuantTh4inK}, over an embedding in
logic \textsf{K}.  Footnote~25 of the \emph{Notes} is thereby
corrected in the object language: symmetry is not needed for \tc{Th4} in the
mixed-quantifier setting, and neither is reflexivity.  The footnote attributes the
dependency to \tc{Th2}; in the AFP sources, as in the port, \tc{Th2} is frame-free and the
citation sits on \tc{Th4}.

\section{Discussion}
\label{sec:discussion}

The article has settled its questions in a particular way, and this section asks what that way
is worth.  Two of G\"odel's own commitments come first: how much \tc{Ax1Gen} asserts, and how
little the modal collapse needs.  The criterion itself comes next, in three parts: what a
hybrid-free certificate does establish, what it gives over general models and what it leaves
open, and why it is a method and not only a verdict about these theories.  Then come what each
of the two proof assistants contributed, what the \textsf{S4} framing of the \emph{Notes}'
questions turns out to show, and where this work stands in the framework it uses.  What remains
open is collected at the end.

\subsection{The strength of \tc{Ax1Gen}} Reading G\"odel's 1970 footnote that the
conjunction axiom should cover ``any number of summands'' as a third-order axiom of the
embedding is due to Anderson and Gettings~\cite{AndersonGettings1996} and to
Fitting~\cite{Fitting2002}; the form is G\"odel's own, his Axiom~2 of 1952/53 quantifying over a
second-order collection of properties and placing no restriction on it
(Section~\ref{sec:setting}).  The \emph{Notes} adopt it as \tc{Ax1Gen} and introduce it to derive
$\Pos\,G$ (Sect.~4.4), and in
Sects.~4.5 and~4.6 their provers use it directly, for \tc{Th4} and for \tc{NecNoEvil}.  They
remark that ``it is therefore possible to omit \tc{Ax1Gen} \dots\ and simply postulate
\tc{L} as an axiom''~\cite[Sect.~4.4]{J75}, which is what Scott's \tc{A3} does.  The axiom
says more, at two levels.  The first is the two degenerate instances of requirement~(ii),
$\Phi=\emptyset$ and $\Phi=\{G\}$, where no world-dependence is involved.  Neither follows
from the binary \tc{Ax1}, and with \tc{Ax2a} they already give the possible existence of a
God-like being without \tc{Ax4}.  G\"odel states the lever himself, beside Theorem~9 of his
1952/53 version: a positive property that is impossible makes every property positive,
``$\Pos\varphi$ . ${\sim}M(\exists x)\varphi(x)\supset(\psi)\Pos\psi$'' in his
notation~\cite[p.~1019]{KanckosLethen2021}, and there it comes from his Axiom~1, the \tc{Ax4}
of the \emph{Notes}.  In variant~3, whose weakened \tc{Ax4} does not apply to
an empty $G$, this is what makes \tc{Th4} provable at all under \NFig{8} as it stands;
Section~\ref{sec:plain} locates that obstacle in where the clause of \NFig{8} sits rather than in
the restriction the two-conjunct reading imposes.  \tc{Ax1Gen} is thus
also an extensionality principle for positivity: its empty instance counts, and
positivity does not distinguish necessarily coextensive properties.  That closure \tc{Ax4}
has by itself, so what \tc{Ax1Gen} adds here it adds over \tc{Ax1} and not over the
theory (Section~\ref{sec:assessment}).  Neither half of that principle survives
\tc{Ax1GenTwo}, which licenses neither degenerate instance.  The second level is
Theorem~\ref{thm:gex}.  The set $\Phi$ ranges over world-dependent collections of
properties, so the axiom can transport a fact about the current world (here, that no
God-like being exists) into the positivity of a property, which \tc{Ax4} then turns
against $\Pos\,G$.  The step needs neither symmetry nor reflexivity: the God-like being is
not obtained at a successor world and brought back but at the world itself, and in this respect
the derivation of Theorem~\ref{thm:gex} hardly argues modally.  The operators are not idle --- the inclusion in \tc{Ax4} and
$\mathit{ConjOfPropsFrom}$ are both boxes --- but nothing is transported across the
accessibility relation.

One more thing about that derivation is worth recording.  The set $\Phi$ of
Lemma~\ref{lem:l2} is a comprehension instance at the type of sets of properties, a
third-order instantiation: a world-dependent set of properties whose condition contains a box
and, through $G$, a quantifier over properties.  The provers the \emph{Notes} ran on these
questions did not find it, and the census of Appendix~\ref{app:runs} shows how exceptional it is
in the dataset, where nearly every instantiation is a constant, a complement, or the universal
or the empty property.  A comprehension instance of higher order is also what makes Boolos's
curious inference short~\cite{Boolos1987}; whether provers find that proof, and what help they
need, is the question of~\cite{WhoFinds}, where they did once suitable shorthand notations had
been supplied.  The derivation here is a second concrete case, from metaphysics, in which the
instance was found by hand after the provers had stopped.
What is shown is that this proof uses such an instance and that no prover found one, not that
the result needs one: whether a witness of lower order, or one with a quantifier-free condition,
would do is not settled here.

\subsection{Modal collapse without symmetry}
Modal collapse, $\val{\varphi\supset\Box\varphi}$, is the best-known objection to
G\"odel's argument, raised by Sobel~\cite{Sobel1987}, and arguably a consequence G\"odel
accepted: Kova\v{c} argues from G\"odel's own texts that it is an intended consequence of the
system and not a defect of it~\cite{Kovac2012}.  Formally it comes with the
maximality of the positive
properties~\cite{MuehlenbeckBenzmueller2026,Benzmueller2026comment}; the ultrafilter structure
of the positive properties, and which variants avoid the collapse, were analysed on the computer
in~\cite{BenzmuellerFuenmayor2020}.  The \emph{Notes} show it to follow in \textsf{KB}.
Theorem~\ref{thm:consequences}(b) shows that with \tc{Ax1Gen} it follows in \textsf{K},
and Theorem~\ref{thm:consequences}(c) draws the semantic conclusion: every world sees only
itself.  In variant 3 there is only one world altogether (Theorem~\ref{thm:v3}(b)), and
everything that exists exists actually (Theorem~\ref{thm:v3}(c)).  The axioms thus erase
the distinction between possibilist and actualist quantification, which the \emph{Notes}
take care to maintain.  The erasure belongs to that variant with those quantifiers and
not to the axioms as such: under the emendation of Section~\ref{sec:plain}, with the quantifier
in the inclusion read possibilistically, \tc{AllExist} has a genuine countermodel.  Theorem~\ref{thm:nonecexist} locates the source in the
definitions.

\subsection{What the criterion establishes, and what it leaves open}
\label{sec:criterion-open}
Twice in this article a first proof used a nominal ---
\tc{Th4} of the mixed-quantifier setting (Section~\ref{sec:th4k}) and the uniqueness of
essences in variant~3 (Theorem~\ref{thm:uniqueess-k}) --- and a hybrid-free proof was found in
its place.  Both hybrid
proofs were correct: they establish the theorem in every model of HOL, standard or Henkin.  What
the hybrid-free proof adds is what the introduction said it is: a syntactic certificate about the
terms a proof substitutes (Proposition~\ref{prop:sound}).  A soundness theorem for a class of
models closed under the operations of $\mathcal{L}$ alone has to cover more than that criterion
checks: a proof term carries terms besides its witnesses, the bodies and motives the detector
passes over, and those have to denote in such a model too.  The semantic question therefore has
two parts, what such models are and whether the proofs certified here meet that stronger
condition.  Both are answered below for the theorems of this article; for the witness criterion
as such the answer is no.  An example separating provability in the embedding from
truth in such models is given below (Example~\ref{ex:separating}); no result this article calls
hybrid-free is false in a model of that kind, and for the statements still proved with a nominal
the question stays open (Section~\ref{sec:open}).

The question is not idle, because $\mathcal{L}$ is strong over standard models.  With
quantification over propositions and the primitive equality,
\[
\mathsf{D}\varphi \;:=\; \exists\psi.\ \neg\psi \wedge \bigl((\psi\wedge\varphi) \neq \bot\bigr)
\]
says at $w$ that $\varphi$ holds at some world other than $w$, in every model whose
propositional domain separates any two worlds: the difference modality of de
Rijke~\cite{deRijke1992}, here from the global modality~\cite{GorankoPassy1992} and
propositional quantification.  With it an atom, and satisfaction at an atom, are definable,
\[
\mathsf{Atom}\,\psi \;:=\; (\psi \neq \bot) \wedge \neg\bigl((\psi \wedge \mathsf{D}\psi) \neq \bot\bigr),
\qquad
@_\psi\varphi \;:=\; (\psi\wedge\varphi) \neq \bot \quad\text{for an atom } \psi,
\]
the first saying that $\psi$ holds somewhere and nowhere else.  Over standard models, then,
$\mathcal{L}$ has the expressive power of the hybrid language --- not a closed term for each
world, the terms being countably many, but atoms, quantification over them and satisfaction at
them --- and the certificate distinguishes nothing about what a proof may say of worlds.  (For
the related result that the nominal-free sentences of the hybrid language with $\downarrow$ and
$@$ are expressible with propositional quantifiers already without the global modality, see ten
Cate~\cite{tenCate2006}.)  The certificate's content is relative to general models whose propositional domain
lacks such atoms; and those are not the models of the completeness theorem below, whose Henkin
domains must contain every singleton.  That is the gap the semantic question names; the
models in question are the following.

\begin{definition}[$\mathcal{L}$-general model]
\label{def:lgeneral}
An $\mathcal{L}$-general model is a Kripke model $(W, r, D, \mathit{existsAt})$ with a domain
$P_\tau$ for each type $\tau$ of $\mathcal{L}$ --- $P_e = D$, $P_\sigma \subseteq 2^W$,
$P_{\tau_1\to\tau_2} \subseteq P_{\tau_1} \to P_{\tau_2}$ --- on which the signature is
interpreted and which are closed under $\mathcal{L}$: every term of $\mathcal{L}$ of type
$\tau$ with parameters from the domains denotes in $P_\tau$.  It is a model of a set $T$ of
sentences of $\mathcal{L}$ --- the axioms of a variant, say --- if every sentence of $T$ holds
at every world, its quantifiers ranging over the $P_\tau$.  A standard model is one with every
$P_\tau$ full.  A Henkin model of HOL, whose domains at every type of HOL are closed under
every term of HOL, restricts to an $\mathcal{L}$-general model.
\end{definition}

The definition has two ancestors.  At type $\sigma$ the $P_\sigma$ are the admissible
propositions of a general frame --- a Kripke frame with a set of subsets of $W$ closed under the
connectives and the modality, over which the propositional quantifiers range --- the semantics
Fine~\cite{Fine1970} gave propositionally quantified modal logic, a subject that begins with
Bull~\cite{Bull1969}; for its current state see Fritz~\cite{Fritz2024}.  At the higher types it
is Henkin's general model for the theory of types~\cite{Henkin1950}, in the form Andrews gave
it~\cite{Andrews1972}, where closure under the terms is required of the domains rather than
assumed to follow from their being domains of functions.  Here closure is required for the terms of
$\mathcal{L}$ only, so an $\mathcal{L}$-general model is weaker than a Henkin model of HOL; the
model of Example~\ref{ex:separating} below is one and not the other.

\begin{proposition}[soundness over $\mathcal{L}$-general models]
\label{prop:lgeneral}
Let $T$ be a set of sentences of $\mathcal{L}$ and $\pi$ a proof in HOL of $\val{\varphi}$
from $T$, with $\varphi$ a term of $\mathcal{L}$.  Suppose that $\pi$ uses no constant of a
type of $\mathcal{L}$ outside the signature except through its definition, and that every term
of $\pi$ whose denotation is used as an element of a domain of a type of $\mathcal{L}$ is, after
unfolding the defined notions and folding back the connectives, a term of $\mathcal{L}$ with
parameters of types of $\mathcal{L}$.  The positions in question are the argument of a constant
of the signature, the argument of a variable of such a type and the instance of a quantifier at
such a type, wherever in $\pi$ they lie --- in a hypothesis, a motive, a side of an equation or
the statement itself.  Then $\val{\varphi}$ holds in every $\mathcal{L}$-general model of $T$.
\end{proposition}

\begin{proof}
Fix an $\mathcal{L}$-general model $M$ of $T$.  Interpret the types of HOL standardly over the
worlds and individuals of $M$: the type of worlds by $W$, the individuals by $D$, $o$ by the two
truth values, and every function type by the full function space.  So every term of HOL
denotes, by the usual recursion modulo $\beta$- and $\eta$-conversion, and a term of a type
$\tau$ of $\mathcal{L}$ denotes a function on the full domains whether or not it lies in
$P_\tau$.  Three things are relativised to $M$.  A quantifier at a type $\tau$ of $\mathcal{L}$
ranges over $P_\tau$; every other quantifier ranges over its full domain.  Equality at a type of
$\mathcal{L}$ is agreement on $P$: identity at $e$, and at $\sigma$, whose argument type is the
type of worlds, not a type of $\mathcal{L}$, so that agreement on every world is identity; and
at a function type $\tau_1\to\tau_2$ of $\mathcal{L}$ the relation ``$f$ and $g$ agree, in this
sense, on every element of $P_{\tau_1}$''.  And a constant of the signature, as well as the value
of a variable at a function type $\tau\to\tau'$ of $\mathcal{L}$, which is a function on
$P_\tau$, is extended to the full domain by restriction: its value at an argument $F$ of the
full domain is its value at $F{\restriction}P_\tau$ whenever that restriction is an element of
$P_\tau$, and arbitrary otherwise, so that arguments which agree on $P$ receive the same value;
the connectives of the embedding and the defined constants are interpreted by their
definitions.  Under this interpretation the sentences
of $T$ hold, since $M$ is a model of $T$ with its quantifiers at types of $\mathcal{L}$ ranging
over the $P_\tau$.  What the hypothesis buys is this: wherever $\pi$ uses a term of a type of
$\mathcal{L}$ as an element of $P$ --- as the instance of a quantifier at that type, as the
argument of a constant of the signature, or as the argument of a variable of a type of
$\mathcal{L}$, which ranges over $P$; after unfolding the defined notions, an application of a
term of a type of $\mathcal{L}$ is one of these --- that term is a term of $\mathcal{L}$ with
parameters of types of $\mathcal{L}$, the parameters range over $P$, and by closure it denotes in
$P_\tau$.  Now the rules of HOL that $\pi$ uses.  $\forall$-elimination at a type of
$\mathcal{L}$, from $\forall x.\ A$ to $A[t/x]$: $t$ denotes in $P_\tau$, the range of the
quantifier; at every other type the range is full.  $\exists$-introduction likewise.
$\forall$-introduction and $\exists$-elimination: a derivation of $A$ for an arbitrary $x$ holds
for every element of the full domain, hence for every element of $P_\tau$.  Equality is
reflexive, and substitutive: from $s = t$ and $A(s)$ to $A(t)$, where $A$ applies $s$ to an
argument that argument is an element of $P$ and $s$, $t$ agree there; where $A$ hands $s$ to a
constant of the signature or to a variable of a type of $\mathcal{L}$, $s$ and $t$ are elements
of $P$ and equal as such, and a constant or variable extended by restriction gives the same
value to arguments that agree on $P$; and where $s$ is of type $e$ or $\sigma$ the two agree
outright.  A context of a type outside $\mathcal{L}$, a variable $h : \tau \to i$ say, could
tell $s$ from $t$, but it cannot be used: the Leibniz instance $\lambda x.\ h\,x = h\,s$ would
be an instance of a quantifier at a type of $\mathcal{L}$ with a parameter of a type outside
$\mathcal{L}$, which the hypothesis forbids.
Extensionality of functions at a type $\tau_1\to\tau_2$ of $\mathcal{L}$, from
$\forall x.\ f x = g x$ to $f = g$: the relativised premise says that $f$ and $g$ agree on
$P_{\tau_1}$, which is the relativised conclusion; at every other type it is the standard
extensionality of functions.  Extensionality of propositions, the Boolean rules and excluded
middle hold because $o$ is interpreted by the truth values.  Choice at a type other than those of
$\mathcal{L}$ is standard choice; a choice term at a type of $\mathcal{L}$ is not a term of
$\mathcal{L}$, so by the hypothesis it never stands where an element of $P$ is required, and
where it stands it denotes like any other term.  $\beta$- and $\eta$-conversion hold because the
interpretation is modulo them.  So every line of $\pi$ holds under the interpretation, and so
does the last, $\val{\varphi}$; and since $\varphi$ is a term of $\mathcal{L}$, whose quantifiers
range over the $P_\tau$, that is $\val{\varphi}$ holding in $M$.
\end{proof}

The hypothesis is a whole-term criterion, and it is checked.  Both detectors have a whole-term
mode --- \tc{\#hybrid\_all} in \lean{}; in \isa{} the third reading of \tc{HybridAudit.ML}, added
in the present version --- which examines every term at a position of the kind the hypothesis
names, wherever it lies in the proof term or the statement --- inside the hypotheses, the motives
and the propositions handed to excluded middle, which the witness mode passes over --- and
traverses the rest; which positions are of that kind it reads off the declared statement of the
rule applied, not off a list of rules (Appendix~\ref{app:detector}).  Over the 55 modules of this development and of the port it
flags no theorem that the witness mode does not flag or reach by inheritance
(\tc{wholeterm-report.txt}).  The \isa{} reading flags nothing more over the nine theories of
this article and the sixteen sessions, and over the \emph{Notes} four theorems more, \tc{Inconsistency}
and \tc{T3} in two theories each, whose proofs by \tc{smt} leave abstraction terms inside a
context (\tc{isabelle-report-whole.txt}); each of the four has a proof written out in the
cross-check, which the reading passes.  So every result this article calls hybrid-free meets the
hypothesis of Proposition~\ref{prop:lgeneral} on both sides --- the clause on constants is the
declaration check of Appendix~\ref{app:detector} --- and holds in every $\mathcal{L}$-general
model of its variant; and so does every result of the \emph{Notes} the witness mode certifies,
through the port and through the replays of the cross-check.

The witness criterion alone does not give Proposition~\ref{prop:lgeneral}, and the gap between
it and the whole-term condition is not hypothetical.  A formula such as $P\,(\lambda u.\ u = w)\,w$, introduced by excluded
middle and discharged again, has no denotation in the model, and the witness mode does not see
it.  An earlier version of this article conjectured that such a formula cannot be put to use
without a flagged instantiation.  It can.  Bj\"orn Lindheim proved the separating sentence of
Example~\ref{ex:separating} in the embedding with the nominal standing at one place only: as the
point at which two predicates on propositions, proved equal by extensionality, are applied
(\tc{congrFun}), the $\forall$-elimination inside extensionality.\footnote{Bj\"orn Lindheim, with his autoresearch
harness (Claude Opus~5.5 agents, audited by Claude), personal communication, 29 October 2026; the
probe and the findings of Section~\ref{sec:open} are deposited
as~\cite{Lindheim2026a,Lindheim2026b,Lindheim2026c,Lindheim2026d}, with the repository fixed at commit
\tc{c62ee37}.}  His
proof is \tc{probe\_separating\_congrfun} of the \lean{} suite; the step itself, in both
systems, is \tc{probe\_congrfun\_point}.  The detectors of both systems as they stood in that
version accepted that step: each gated the arguments of the library's rules by a list
of rules with hand-named witness positions, and the point of \tc{congrFun}, of \tc{fun\_cong} in
\isa{}, was on neither list.  Since then both read the schema positions of a rule off its
declared statement (Appendix~\ref{app:detector}) and flag the step; the figures of this article are
those of the corrected detectors.  The step can be moved, though, to a position no
reading of instantiations examines: the same sentence is proved with the nominal inside the
context handed to \tc{congrArg}, to \tc{arg\_cong} in \isa{}, a schema position
(\tc{whole\_separating\_context} in both suites; the \isa{} proof writes the sentence out,
since a lemma about the nominal instantiates a definition at it there, which that reading
counts, where \lean{} unfolds by conversion and leaves no argument behind).  The witness mode of either
system passes that proof, and the whole-term mode of either refuses it.  That settles the
question against the conjecture: the witness criterion is the weaker notion, and the whole-term
condition is the one the certificate rests on.  The certificate has content, as the example
shows: the axioms, written in \textsf{K} with the global modality, have an object-language
reading, and a sentence of that language is excluded.

\begin{example}[a separating sentence]
\label{ex:separating}
Let $\varphi := \mathsf{A}\,\exists p.\ (\mathsf{Atom}\,p \wedge p)$, ``every world lies in an
atom'', a sentence of $\mathcal{L}$ with $\mathsf{Atom}$ as above.  It holds in every standard
model, and the embedding proves it with a nominal: at a world $w$ take $p := \lambda u.\ u = w$;
$p$ holds at $w$, $p \neq \bot$, and $\mathsf{D}p$ fails at the one world of $p$, since every
$\psi$ that meets $p$ contains $w$.  The detector flags the witness \textsc{nominal} and
\textsc{rigid}.  Now let $W$ be the infinite $0$-$1$-sequences, let every world see every world,
let $D$ be a single individual that exists at every world, and $\Pos$ the empty set at every
world.  Let $G$ be the group of \emph{flips}, the maps $x \mapsto x \oplus d$ for a $0$-$1$-sequence
$d$, acting on $W$ and, by conjugation, on the full type hierarchy over $W$ and $D$; the flips
that leave the first $n$ coordinates alone form the subgroup $G_n$.  Call an object of type
$\tau$ \emph{finitely supported} if some $G_n$ fixes it, and let $P_\tau$ consist of the
hereditarily finitely supported objects: at a function type, the finitely supported functions
from $P_{\tau_1}$ to $P_{\tau_2}$.  This is the Fraenkel--Mostowski
construction~\cite[Ch.~15]{Jech2003} over $W$, in the form Pitts gives it~\cite{Pitts2013}.
Three facts.  (a)~The $P_\tau$ form an $\mathcal{L}$-general model $M$ of the empty set of
sentences: a term of $\mathcal{L}$ with finitely supported parameters denotes a finitely
supported object, by induction on terms, the case of a quantifier at a type $\tau$ because $G_n$
permutes $P_\tau$ and, once $n$ supports every parameter, fixes the parameters, so that the set
defined is fixed by $G_n$.  (b)~At type $\sigma$, finitely supported means \emph{finitary}:
membership depends on the first $n$ coordinates alone, since the flip of every coordinate beyond
$n$ at which $x$ and $y$ differ carries $x$ to $y$.  (c)~$M$ has no atoms: for $p \in P_\sigma$
with support $n$ and $w \in p$, flipping coordinate $n$ of $w$ gives a $w' \in p$, so that
$q := p \cap \{x \mid x_n \neq w_n\} \in P_\sigma$ is nonempty, inside $p$ and omits $w$; thus
$\mathsf{D}p$ holds at $w$ and $\mathsf{Atom}\,p$ fails.  In particular no singleton lies in
$P_\sigma$: $M$ is not a Henkin model of HOL, whose domains, as noted above, contain every
singleton.  Hence $\varphi$ is false at every world of $M$, and by
Proposition~\ref{prop:lgeneral} no proof of $\varphi$ meets the proposition's hypothesis: the
whole-term condition excludes a sentence the embedding proves, and the certificate, which checks
that condition, refuses every proof of it.  The propositional part of this --- a deep
embedding of $\mathcal{L}$ with its quantifiers over propositions, the family $P_\sigma$ and its
closure, the absence of atoms, and $\varphi$ false in $M$ and true in every standard model ---
is formalised in both systems (\tc{Separation.thy}, \tc{Separation.lean}); the higher types rest
on (a).  Two remarks.  No model with finitely many worlds separates: there the intersection of
the $p \in P_\sigma$ containing $w$ is a finite intersection, lies in $P_\sigma$ and is an atom
containing $w$, so $\varphi$ holds; a separating model is infinite.  And $M$ is not a model of the axioms of variant~2, whatever $\Pos$ is taken to be: the modal collapse holds there
hybrid-free (Theorem~\ref{thm:consequences}(b)), so $X \subseteq \Box X$ for every
$X \in P_\sigma$, which fails for every finitary $X$ other than $\emptyset$ and $W$, since under the total relation $\Box X$ is empty unless $X = W$.  With the identity in place of
the total relation, $M$ is a model of those axioms (Section~\ref{sec:open}).  Nothing in
Sections~\ref{sec:gex}--\ref{sec:neighbours} is touched; what the example shows is that the
questions Section~\ref{sec:open} collects have content.  A frame with
an order separates as well: the irrationals under $<$, with $G$ the order-automorphisms of the
reals that map the rationals onto themselves, where a finitely supported set of type $\sigma$ is
a finite union of gaps between rationals of the support, Cantor's back-and-forth acting
transitively on the irrationals of each gap; that model is not formalised.
\end{example}

The soundness and completeness of the embedding~\cite{R45,J26} are relative to Henkin
semantics, not to $\mathcal{L}$-general models.  The Henkin model of~\cite[Def.~4.2]{J26} takes
the modal model's propositional domain $P$ for the domain of type $i\to o$.  Every term of HOL
must denote in a Henkin model, the singletons $\lambda u.\ u = w$ among them, since that type
theory has equality as a primitive at every type; so $P$ must already contain the nominals and
be closed under the operations of HOL.  Of the modal models the definition asks only that $P$
contain the sets a formula of the modal language defines, and a singleton set of worlds need not
be one of those, although the term denoting it has type $i\to o$.  The clause quoted in the introduction
is therefore what the construction requires rather than a gloss on it, and the theorem statements
do not carry it: \cite[Thm.~4.4]{J26} is stated without it and is to be read with the clause its
construction supplies.  Whether it would also hold without the clause is a separate question,
and nothing here turns on it.  The other theorem, \cite[Thm.~4.9]{J26}, does not need it, since the
modal model it builds from a Henkin model takes that type's whole domain for $P$.  The clause
bites only where something quantifies at that type, so the first-order and propositional
fragments, \cite[Cor.~4.11 and~4.12]{J26}, are untouched.

Nitpick searches standard models of HOL of bounded size, Henkin models in particular, so its
countermodels refute a statement in HOL and hence in the embedding.  A countermodel to the
object-language reading of a statement would have to omit the nominal properties from its
domains, which no Henkin model of the embedding does; no model finder over the shallow embedding
can produce one, and for the statements treated here none is needed, since the theorems hold.

What the nominal proofs establish beyond that is genuine and is recorded (\tc{OneWorld},
\tc{NoNecExist}): the models of HOL leave the axioms of variant~3 exactly one world.  No formula of
$\mathcal{L}$ this article answers rests on a nominal; where a nominal remains, at the five
places Section~\ref{sec:open} collects, two of them under the recommended reading, it is not
known whether it is needed or merely convenient.  Which answers are known only with a nominal
is in Table~\ref{tab:verdicts}, reading by reading.  The frame statements \tc{OneWorld} and \tc{NoNecExist} are hybrid by
nature, and so is what is drawn from them alone: \tc{Th3} of variant~3 without the conjunction
axiom (Theorem~\ref{thm:v3}(a)), and \tc{Th5\_S4}, \tc{MC\_S4} under \tc{Ax1GenBox}.

\subsection{A method, not only a verdict}  The criterion of
Section~\ref{sec:criterion} asks of a proof term only which instantiations it makes at the
types of Definition~\ref{def:hybridfree}, so it is not tied to these arguments.  It applies to
any development over the shallow embedding, and to embeddings of other non-classical logics
whose object language is a fragment of the host.  A replay in an explicit proof calculus settles
the question by construction --- in~\cite{KanckosPaleo2017} for a rigid higher-order modal logic,
a restriction the embedding does not impose and this article accordingly studies
(Section~\ref{sec:rigid}).  Set beside that, the criterion is the cheaper and the weaker route.
It certifies that what a proof substitutes stays within
the object language, not that the theorem has an object-level derivation: a proof it rejects
may still have one, as the two cases above show.  What it offers in exchange is scale (here the
427 facts of the \emph{Notes} and the 399 of the port), and it costs nothing to re-run when a
development changes.  The audit of the \emph{Notes} also shows something about the method itself.
A proof term carries terms of two origins: those the argument needs, and those automation
introduced --- a Skolem function of a world, a constant of the library.  An audit has to tell
them apart, since a flag on the second kind says nothing against the argument.  The 114 proofs the
detector did not certify are all of the second kind, and a proof written for the purpose replaces
each (Section~\ref{sec:audit}, Appendix~\ref{app:runs}).

\subsection{What each system contributed}
\label{sec:contributed}
Everything this article establishes is established in both systems, the model of
Example~\ref{ex:separating} in its propositional part.  Each of the sixteen \isa{}
sessions has a \lean{} module for the same setting, and every result stated in the one is stated
in the other; only the countermodels are Nitpick's, and the \lean{} modules record them as
placeholders.  What is on one side only is on it because the systems differ --- a model search,
for want of a model finder on the other side, or an auxiliary step of a proof, which follows the
proof and not the statement --- and never a result (Appendix~\ref{app:verification}).  The
criterion of Section~\ref{sec:criterion} is implemented for both and the two detectors agree on
what they flag; the account of which postulates a theorem uses, which the \lean{} side gives
through \tc{\#print axioms}, the \isa{} side gives by folding over the recorded proof.  What
\lean{} adds is an independent check: a second kernel, a second embedding, proofs
written out explicitly against the same statements, and proof terms in which every instantiation
stands as written.  Nothing here requires that second system --- what it buys is that a mistake
in one embedding or one kernel would have to be repeated in the other to pass unnoticed.

The dependence is not symmetric: the refutations are Nitpick's, and \lean{} has no comparable
facility; that is why the one countermodel needed on that side is written out explicitly and
checked (\tc{Ax1GenTwoVariant3Model}).  Nothing in this article that refutes
a statement could have been obtained in \lean{}, and the proof search that produced the \isa{}
proofs in the first place --- \tc{sledgehammer} with a replay by \tc{metis} or \tc{smt} --- has no
counterpart there either; where the port needed those proofs, they were written out explicitly.  Both
observations are about the tools as this study used them, not verdicts on the systems.

\subsection{What the \textsf{S4} restriction shows}
The \emph{Notes} ask whether \textsf{S4} will do in place of \textsf{KB}.  The two are
incomparable: \textsf{S4} adds reflexivity and transitivity, \textsf{KB} symmetry.  What
the question turns on is that \textsf{S4} does without the symmetry their proof of
\tc{Th3} uses.  The answer is that the question dissolves: in the variants with \tc{Ax1Gen} the results do
not depend on the frame at all, and the frame is trivialised by the axioms.  The restriction
this article adopts does not bring the frame back either: under \tc{Ax1GenTwo} the answers are
frame-free as well (see Table~\ref{tab:deps-two} on p.~\pageref{tab:deps-two}).  \textsf{S4}
becomes a question with an answer only where the world-dependence of $\Phi$ is attacked instead
of its cardinality --- under
\tc{Ax1GenBox}, under \tc{Ax1GenR}, or under the two together with the two-conjunct restriction.
And there the answer is the one the \emph{Notes} expected: a countermodel.

\subsection{The place of this work in LogiKEy}
The \emph{Notes}, the port and this article apply the method of the LogiKEy
framework~\cite{J48}: an object logic is embedded shallowly in HOL, and the reasoning and model
finding of a proof assistant are used on it as they stand.  The application is computational
metaphysics~\cite{J45}: a metaphysical theory is the object theory, and the assistant supplies every verdict
about it --- theorem, countermodel and open question alike.  Kirchner, Benzm\"uller and Zalta
carry the same method to the theory of \emph{Principia Logico-Metaphysica}~\cite{J50}.
The framework's strength is that the
embedding gives away nothing of HOL: the same tools reason about the object logic, about the
theory it is used for, and about the semantics that interprets both.  Where the object logic
quantifies over propositions or properties, as G\"odel's does, the same freedom lets a proof
substitute what the object language has no term for, and that is where Section~\ref{sec:criterion}
draws its line.  The \emph{Notes} turn out to have a proof within that line for every statement
they prove; where a proof has none --- the nominal arguments of Section~\ref{sec:v3} that leave
a single world --- the detector says so, and the result stands as a fact about the models rather
than a theorem of the logic.  Both kinds of result have their place in the framework; what was
missing, for its higher-order applications, was the means to tell them apart after the fact.

\subsection{What stays open}
\label{sec:open}
Three questions are left open above, and they are collected here.

The first is what nominals still do.  Two things are open under \tc{Ax1GenTwo}: whether the side
condition~($\dagger$) of Theorem~\ref{thm:v3-two} can be dropped in variant~3 with a hybrid-free proof, and
whether the equivalence of \tc{Ax1GenTwo} with \tc{Ax1Gen} in variants~2 and~3$'$ has a
hybrid-free proof (Section~\ref{sec:plain}).  Two more are open under readings this article does
not recommend: whether \tc{G\_ex} follows hybrid-free in variant~3$'$ under the combination
(Section~\ref{sec:boxtwo}), and whether the \emph{Notes}' results follow hybrid-free in
variant~3 under the one-conjunct reading (Section~\ref{sec:one}).

The second is whether \tc{L} follows under the rigid reading with no condition on positivity
(Section~\ref{sec:rigid}).  The third is what remains of the semantic question behind the
criterion.  Proposition~\ref{prop:lgeneral} gives soundness over $\mathcal{L}$-general models
for proofs that meet the whole-term condition, the results of this article meet it, and
Example~\ref{ex:separating} shows a sentence of $\mathcal{L}$ the embedding proves and the
condition excludes.  The witness criterion alone does not imply the condition: the separating sentence has a proof
with the nominal in a context alone, which the witness mode passes
(Section~\ref{sec:criterion-open}); that is why the certificate checks the condition itself.

Three parts of these questions --- the second part of the first, the second question entire, and
the last part of the third --- were answered, within days, by Bj\"orn Lindheim's harness
(footnote in Section~\ref{sec:criterion-open}), in \lean{} files that build against this
development with its postulates only and that both modes of the detector
pass~\cite{Lindheim2026b,Lindheim2026c,Lindheim2026a,Lindheim2026d}.  The equivalence of
\tc{Ax1GenTwo} with \tc{Ax1Gen} has a hybrid-free proof in variants~2 and~3$'$
(\tc{Ax1Gen\_of\_Two\_hf} and \tc{Two\_equivalent\_hf} in \tc{TwoV2.lean} and
\tc{TwoV3Plain.lean}): the singleton and the empty instance are carried to every successor by
modal collapse at sentences of $\mathcal{L}$ built from $\Phi$ and its one member, where the proofs
of Section~\ref{sec:plain} had used \tc{R\_Id}; and the equivalence holds in every
$\mathcal{L}$-general model of either variant (\tc{two\_equivalent\_v2}, \tc{two\_equivalent\_v3p}
in \tc{DeepE.lean}, over a deep calculus proved sound and complete for the models of
Definition~\ref{def:lgeneral} with a nonempty domain of individuals, \tc{sound\_complete\_all} in
\tc{Canon.lean}).  \tc{L} follows under the rigid reading with no condition on
positivity, in \textsf{K} (\tc{L\_R\_hf} in \tc{RigidHF.lean}): for every $\psi$ the guarded
property $\lambda x.\ P\psi \supset \psi x$ is globally positive by \tc{Ax2a}, \tc{Ax2b} and
\tc{Ax4} and agrees with $\psi$ wherever $\psi$ is positive, which gives the globally positive
properties the conjunction clause that Section~\ref{sec:rigid} could give them only with
positivity the same at every world (\tc{L\_R\_rigidpos}).  And a separating model exists for variant~2: the
model of Example~\ref{ex:separating} with the identity in place of the total relation
(\tc{sep\_v4} in \tc{DeepV4.lean}), in which $\Box X = X$, so that the collapse and the axioms of
variant~2 hold --- unrestricted \tc{Ax1Gen} and \tc{Ax1GenTwo} alike (\tc{sep\_v4\_ax1gen} in
\tc{TranslateTwo.lean}) --- while the finitely supported propositions still have no atom and the
separating sentence fails at every world.  What stays open of the first question is the side
condition~($\dagger$) and the two readings this article does not recommend; of the third, whether a
separating model exists for variant~3.

\section{Verification}
\label{sec:verification}

Every theorem above is machine-checked twice, in \isa{} and in \lean{}, the model of
Example~\ref{ex:separating} in both for the propositional part of $\mathcal{L}$, and every countermodel is
Nitpick's, certified by a build rather than asserted here.  Over the eight theories both systems have, the two detectors flag the
same results and certify every other fact (Appendix~\ref{app:verification}); the ninth, which
only one of them has, carries a countermodel, and no theorem stated here rests on it alone.  Appendix~\ref{app:verification}
gives the details: what the two formalisations are, what the detector of Section~\ref{sec:criterion} reads
when it is turned on this article's own theories rather than the \emph{Notes}', how the account of
which postulates a proof uses is taken on both sides, how to read the two dependency tables, and
where the material is.

The \lean{} formalisation, the hybrid-witness detector and parts of the \isa{} formalisation were
developed by the author in continuous interaction with an AI coding assistant (Anthropic Claude),
which also supported the drafting of this article under the author's direction.  The detector was corrected several times while this article
was written, five of them on the strength of an external reader's reviews, one on a reader's proof
(Section~\ref{sec:criterion-open}) and the rest by the author; its verdicts changed with the earlier
corrections and with the last, the reading of every rule, as Appendix~\ref{app:runs} records, and
not with those in between; every
figure here is from the instrument as shipped.  The author checked and
revised every claim and is
responsible for the content.

\appendix
\appendixheadings
\section*{Appendices}
\addcontentsline{toc}{section}{Appendices}

\section{The detector}
\label{app:detector}

The criterion of Section~\ref{sec:criterion} is applied by two programs, one per system.  This
appendix gives first what both check --- the grammar of $\mathcal{L}$ and the flags a term
outside it receives --- and then how each reaches the proof terms it checks, where the two
differ.  Neither arrangement is the norm from which the other departs: \isa{} records proof terms
only on request, because recording is expensive, and the audit pays that cost without
difficulty; \lean{} keeps them by default.

\subsection{The grammar and the flags}
In a proof term, axioms, theorems and hypotheses are applied to arguments, and the
instantiations sought are those arguments whose type is one of
Definition~\ref{def:hybridfree}'s, together with the terms a \tc{let} binds at such a type.
Bare variables among them are not witnesses and are passed over; what is left is a
$\lambda$-abstraction such as $\lambda x.\top$, or a closed term such as $G$ or
${\sim}\varphi$.

The detector reads them off and checks each against the grammar of $\mathcal{L}$.  Every
defined notion other than a connective of the embedding is unfolded first.  The grammar is a
whitelist.  The term must be built from four kinds of thing: (i)~the connectives of the
embedding, where the polymorphic ones among them --- the quantifiers and the equalities --- are
taken at a type of $\mathcal{L}$; the type of worlds is admitted under a quantifier, which is
then the universal or the existential modality, and flagged under an equality; (ii)~the raw propositional
connectives, equality and quantification, likewise restricted, and the bottom and top of a
lifted type, which is what those connectives unfold to; (iii)~the constants the developments
declare, at a type of $\mathcal{L}$; and (iv)~variables of a type of $\mathcal{L}$, bound in
the term or outside it, and worlds bound in the term.  The types of $\mathcal{L}$ are those
built from $e$, $o$ and $\sigma$ by $\to$, with $\sigma = i\to o$ taken as a unit:
$e\to\sigma$ and $(e\to\sigma)\to o$ are among them, $i$, $i\to\sigma$ and $e\to i$ are
not.  And the body
of a binder over a world may mention no other world bound inside the term: a lifted term speaks
of its own world of evaluation.  Whatever falls outside is flagged, and the flag says why: an
equation between worlds (\textsc{nominal}); the accessibility relation (\textsc{access}); a
world variable bound outside the term, that is, a satisfaction operator (\textsc{rigid}); a
second world under a binder over worlds (\textsc{worldq}); and anything else --- a choice
operator, a binder, variable or constant whose type is not a type of $\mathcal{L}$, such as a
world-indexed family $i\to\sigma$, a Skolem function of a world $e\to i$, a polymorphic
definition instantiated at $i$, or a type copy of $i$ (\textsc{foreign}).  Under these rules
validity $\val{F}$ and a quantifier over worlds whose body speaks of the bound world alone ---
$\forall v.\,F\,v$ or $\exists v.\,F\,v$, raw or in the embedding's notation --- are the
universal and the existential modality, $\mathsf{A}F$ and $\mathsf{E}F$, and are passed as
such.  What \textsc{worldq} catches is a body that also mentions a world bound outside the
binder, as the Leibniz identity of worlds $\forall\psi.\ \psi\,v \to \psi\,u$ does, whether
under a quantifier, inside a property, or as an equation at $\sigma$.  The check is syntactic;
Proposition~\ref{prop:sound} says what it establishes.

\subsection{The two implementations}
The two systems store their proof terms differently, and the detectors reach them accordingly.

\isa{} records proofs under \tc{record\_proofs = 2}, compactly and without the instantiated
terms; \tc{Thm.reconstruct\_proof\_of} recovers them, supplying the sort hypotheses of the
theorem's type variables that a recorded proof also proves.  The connectives of the embedding are
\tc{abbreviation}s and so unfolded in the internal term; they, and only they, are folded back
first, with \isa{}'s reverse rules re-formed so that a quantifier connective folds under a formula
body, which \tc{Proof\_Context.contract\_abbrevs} does not do; any other abbreviation, in
particular one of an audited theory, stays expanded, so that it cannot hide a witness.  The world type, the accessibility relation and the
existence predicate are identified by name, \tc{i}, \tc{R} and \tc{existsAt}, each theory of the
\emph{Notes} declaring its own.  So are the connectives, which the grammar must tell from a
defined notion: by a list of their base names in \isa{}, by the \tc{M} prefix in the root
namespace in \lean{}, the convention both developments follow.  No other constant is identified
by name.  The
\isa{} detector (\tc{HybridAudit.ML}) traverses the reconstructed proof of every theorem and the
nested lemmas automation creates, reconstructing each nested body with its \tc{MinProof} leaves
replaced by placeholders, so that a leaf does not discard the body around it; the leaves are
counted, attributed to the rule they are a premiss of, and their statements scanned.  Every
collected term at one of those types is paired with the rule that introduced it --- the head of its
spine (\tc{Proofterm.any\_head\_of}) --- and with its position among that rule's arguments; which
positions count is stated in Appendix~\ref{app:runs}: the schema positions of a rule --- the
positions that are not counted --- are told from its instance positions by the rule's own
statement, as on the \lean{} side.  Definitions of the library are unfolded inside a witness, to
a depth of 32 as in \lean{} --- \tc{Let}, the combinators \tc{metis} and \tc{blast} leave in
their clause terms, \tc{fun\_app} of the \tc{smt} replay --- so that a term of $\mathcal{L}$
behind one is not flagged for it; in an earlier version only \tc{Let} was.  Definitions of the audited theory are not unfolded; instead the
right-hand side of every definition and abbreviation of the audited
theory is scanned by the grammar and reported, and a flagged definition marks every witness that
uses its constant --- in the theories of the \emph{Notes} and in the nine of this article none is
flagged, and Appendix~\ref{app:verification} says what it flags in the sixteen sessions.
A second traversal collects, for each theorem, the facts of the audited theories
its proof uses, which gives the inheritance verdict; a third collects the postulates it depends on
(\tc{AxiomDeps.ML}), the counterpart of \lean{}'s \tc{\#print axioms}.

In \lean{} no folding is needed, the connectives being definitions there --- a choice of device,
not of system: the \emph{Notes} introduce the connectives by \tc{abbreviation} and $G$, the essence
and necessary existence by \tc{definition}, whereas the port makes all of them definitions.  A
metaprogram (\tc{HybridAudit.lean}, shorter than the \isa{} detector by the reconstruction and
the folding that the one needs and the other does not) traverses the proof term of every theorem
of a module.  In a first mode it lists the instantiations (the \emph{census}): every
argument of an applied proof, other than a bare variable, whose type is one of
those of Definition~\ref{def:hybridfree}, and every \tc{let}-bound term at such a type.  In a second mode it examines each of them --- a
$\lambda$-abstraction, a constant, or a partial application such as $(=)\,w$, the nominal in
$\eta$-reduced form, which is $\eta$-expanded first --- except at the \emph{schema positions} of
the rule applied, which it reads off the rule's declared type: an argument declared as a function
over a type parameter of the rule (the body of $\exists v{:}i$ and of every quantifier lemma of the
library, the motive of a cast or an induction, the context of a congruence, a continuation), or
declared at a type parameter and occurring in the rule's statement only as a side of an equation
(the terms a rewrite relates).  Those are formulas with a hole and the terms a rule of equality
relates, for which the soundness proof of Section~\ref{sec:criterion-open} needs no element of a
domain; every other argument --- the witness of \tc{Exists.intro}, the point at which \tc{congrFun}
applies two functions, the value a \tc{have} binds, and every argument of a fact of the
development, which has no schema positions --- is an instance and is examined.  It applies the grammar of
Section~\ref{sec:criterion} with its five flags; the
existence predicate outside the actualist quantifiers is reported for information only, being
definable as $\exists\actE y.\ y = z$.  Inside a witness every definition other than a connective of the
embedding is unfolded, to a depth of 32, so that a nominal cannot hide behind a helper definition
or behind $\neq$; the connectives themselves are not, since $\Box\varphi$ mentions the
accessibility relation in its definition and is a modal formula all the same.  A quantifier of the
embedding instantiated at the type of worlds quantifies over worlds behind object-language notation
without producing a witness of world-lifted type; the \lean{} detector therefore counts such
instantiations over every statement and proof term as well and reports the count, which is zero in
the port and here.  In \isa{} an abbreviation leaves no trace in the term, so a frame condition
$\forall x.\ x\,r\,x$ and a lifted quantifier over worlds are the same term there; the \isa{}
detector counts over the witnesses, where the grammar decides, and the count is zero in the
\emph{Notes} as well.  A theorem is hybrid if its own proof is flagged,
and inherits hybridness if it is not itself flagged but its proof uses one that is; both detectors
report that closure.  A third mode of each detector --- \tc{\#hybrid\_all} in \lean{},
\tc{run\_theory\_all\_whole} in \isa{} --- examines every term of a type of
Definition~\ref{def:hybridfree} that stands as the argument of a constant of the signature, of
a variable of such a type, or of a proof, anywhere in a proof term and in the statement, and
traverses the rest --- the body of a raw quantifier over worlds, a case distinction over the
frame --- which is evaluated over the worlds and needs no domain; it reports the theorems
flagged there that the witness mode neither flags nor reaches by inheritance
(Section~\ref{sec:criterion-open}, \tc{wholeterm-report.txt}, \tc{isabelle-report-whole.txt}).

Proposition~\ref{prop:sound} is about what the detector accepts, and the detector is a program.
What the two implementations assume --- and what any development they are turned on would have
to satisfy --- is this.  The connectives are identified by name, as above: a definition of the root namespace
whose name begins with \tc{M} is taken for a connective in \lean{} and is therefore not
unfolded, and one outside the \isa{} detector's list is not folded back.  In the developments
audited here every root-namespace \tc{M} definition is a connective of the embedding; a reader
who turns the detectors on another development has to check that first.  Everything else is
checked rather than assumed: a constant or a type the grammar does not know is
\textsc{foreign}, not passed; every proposition-typed constant of a \lean{} module is audited,
whatever its keyword or name --- a proof hidden in a \tc{def}, or in an auxiliary declaration
the elaborator splits off --- and the inheritance closure runs through imported modules; the
\isa{} proof terms are complete, and the definitions and abbreviations of an audited theory are
scanned as well.\footnote{Three earlier assumptions --- that only \tc{theorem}s are traversed,
that the closure runs over theorems, and that it is taken within a module --- and eight terms
outside $\mathcal{L}$ that an earlier, pattern-based grammar let through were pointed out in
pre-submission reviews of this article, with probes written for them; the grammar has been a
whitelist since.  The \isa{} detector took a \tc{Let} for a constant outside $\mathcal{L}$ until
a term of $\mathcal{L}$ behind one was tried; it unfolds a \tc{Let} before the scan now, as the
\lean{} detector reads through a \tc{let}.}

Both are run, on every audit, against the same suite of probes (\tc{lean/HybridProbes.lean},
\tc{audit/probes/HybridProbes.thy}).  On the \isa{} side 45 witnesses leave $\mathcal{L}$
and must be flagged: nominals in $\eta$-reduced form, behind a definition, a double negation,
a \tc{let}; the accessibility relation applied partially and under the converse modality;
satisfaction operators, at a lifted type and at a set of properties \emph{simpliciter}; the
difference modality, and the Leibniz identity of worlds under a quantifier over worlds, in five
spellings (plain, behind a definition, as a redex, under the universal quantifier and as
validity); a second world under a binder --- the Leibniz nominal of the evaluation world inside a
property, as an equation at $\sigma$, under the lifted inequality, under a connective, through
a \tc{let}, and a binder over a world with a constant body; a satisfaction operator at the
evaluation world inside a property; a world-indexed family as a parameter and as a declared
constant; a quantifier over a type copy of $i$; a polymorphic definition instantiated at $i$; a
quantifier of the embedding at the type of worlds over a world-indexed family; a Skolem function
of a world, as a variable and as a declared constant; a choice over worlds; an arbitrary world;
a hybrid definition imported from a theory outside the audited theories; the separating
sentence of Example~\ref{ex:separating}, proved as the embedding proves it, with the nominal as
the witness of its existential; and a nominal at the point of \tc{fun\_cong}, bare and in that
sentence proved as Lindheim proved it (Section~\ref{sec:criterion-open}).  And 43 terms of
$\mathcal{L}$ must not be flagged, among them a raw equation between individuals,
an equation between propositions and one between properties, a raw disjunction, the existence
predicate under the actualist quantifier, rigid sets given by validity and by an equation, the
$\eta$-expanded unfoldings of the connectives at the world of evaluation, the universal and
the existential modality written as quantifiers over worlds, in the embedding's notation and as
validity, plain and over a compound body, the difference modality and the atom it defines
written in $\mathcal{L}$, terms of $\mathcal{L}$ behind a \tc{let} and behind a definition, a term of
$\mathcal{L}$ at the point of \tc{fun\_cong}, the evaluation world free in the sides of an equation
handed to \tc{arg\_cong}, and the body of a modality under the quantifier rules.  Four further
probes are for the whole-term reading, with the term outside $\mathcal{L}$ where the witness
reading traverses and does not examine: in the context handed to \tc{arg\_cong}, bare and in the
separating sentence, in that of a substitution (\tc{subst}; \tc{Eq.ndrec} in \lean{}), and in a hypothesis of the statement.
The \lean{} suite has 48 and 43: the same cases, the same four whole-term probes, and three
probes of the traversal that have no \isa{} counterpart: a proof hidden in
a \tc{def}, one behind an internal name, and hybridness inherited through an import.  A wrong verdict on any probe fails the run, as a Nitpick \tc{expect}
annotation fails a build (\tc{probes-report.txt} on each side).

A second check reads the environment rather than the proofs.  The witness scan examines the
terms a proof substitutes; what it does not see is what a development declares around them ---
a constant or a type whose type mentions the worlds and occurs only in statements, a postulate
with a hybrid wording, an import from outside the audited theories, an oracle.  The declaration
check (\tc{\#hybrid\_decls} in \lean{}, run by \tc{run\_audit.py}; \tc{HybridAudit.decls} in
\isa{}, run by the session \tc{DeclAudit}) lists, per module and per theory, every import from
outside the audited theories, every declared type, every constant --- declared, defined or
abbreviated --- whose type is not a type of $\mathcal{L}$, every postulate with its statement,
and in \isa{} every oracle a theorem depends on; and it compares the list with a whitelist of
what the embedding needs: the types $i$, $e$, $\sigma$ and $\tau$, the accessibility relation
$R$, and the postulates, which are listed and not counted.  The run fails on anything else, and
it fails if the mutation suite does not catch its three planted violations --- a constant with
the worlds in its type, a type copy of the worlds, an import from outside the audited theories.  Over the 55
modules of the port and this development, and over the 55 \isa{} theories, the check
reports five findings outside the whitelist on the \lean{} side and four on the
\isa{} side, every one of them a definition this article names for what it is: the nominal
\tc{Pt}, ``being $x$ at $w$'', of the one-world arguments of sessions \tc{F}, \tc{H},
\tc{O} and \tc{P} and their counterparts, and in \lean{} the link relation \tc{Link}
of the rigid reading's frame lemma, which \isa{} spells with the library's closure operator
(\tc{decls-report.txt} on each side; the harness expects exactly these and fails on any
other).  The \isa{} report prints terms with the abbreviations of every theory loaded in the
session, so a term of one session can appear under another's name there, as \tc{TwoMembers}
does in Appendix~\ref{app:verification}; the check compares names, not printed terms.  What a
run certifies is thereby relative to three things, each of them mechanical: the grammar, the
declaration list, and the postulates the dependency account names.

\section{The two runs in detail}
\label{app:runs}

Section~\ref{sec:audit} reports what the audit found; this appendix says how the two
developments were read.  Four parts concern the \isa{} run --- what its proof store makes
readable, which of the terms in it the criterion counts, the tally that results, and the
cross-check that answers a flag on an automated step by writing the step out --- then comes the
\lean{} census, and last the two runs are compared.  The detail belongs beside the figures because of the reservation of
Section~\ref{sec:criterion-open}: a certificate covers the terms a run could read, and what a run
could not read is therefore reported here too.

\subsection{What the \isa{} run reads}  \isa{} stores a proof compactly, and the detector reconstructs the proof
of every theorem and the nested bodies that automation leaves in it.  A recorded body can contain
\tc{MinProof} leaves --- sub-derivations \isa{} does not keep even at \tc{record\_proofs = 2}:
equations of the rewriter's scaffolding, the reflexivity contradictions \tc{metis} closes on,
congruences of the simplifier --- and \tc{Proofterm.reconstruct\_proof} returns \tc{MinProof} for
the \emph{whole} body at the first such leaf.  Each leaf is therefore replaced by a placeholder the
reconstruction can go around; the leaf is counted and attributed to the rule it is a premiss of, and its statement,
which the reconstruction determines, is scanned like a witness.  Read this way, the 427 facts carry
48\,222 instantiations at those types (63\,872 rule--term pairs), and 102\,656
sub-derivations remain unread: 47\,338
leaves standing bare, 28\,214 premisses of \tc{Pure.combination}, 26\,721 of \tc{Pure.transitive},
208 of \tc{HOL.eq\_reflection}, 112 of \tc{iffD2}, 50 of \tc{disjE}, and thirteen under \tc{impI},
\tc{impCE}, \tc{iffI} and \tc{Pure.equal\_elim}.  Of their statements, 28 mention an equation between
worlds, and those 28 are two forms of one clause: the \tc{smt} proofs of \tc{Th1} and \tc{Th2} of
variant~3 leave a veriT clause with a Skolem constant for a world, in the four theories of that
variant and in the results that inline those proofs.  No other unread statement mentions such an
equation.  That clause
is a statement the prover made, not a term the argument substituted, and the reading below does not
count it; it is reported so that the reader can see what the certificate does not cover.

\subsection{What is counted}  Not every term at one of those types in a proof term is an instantiation the
argument makes.  Over those pairs, the machinery of the proof method --- the equational rules of
Pure and HOL, the body positions of the quantifier rules, the clause steps and Skolemisation lemmas
of \tc{metis} and \tc{smt} and the anonymous clause lemmas they leave in the theory --- carries
33\,377 of the 33\,419 that are flagged, and all 164 occurrences of an equation between worlds: veriT's
Skolem terms $\mathsf{SOME}\ v.\ \ldots$ for worlds, the prover's names for worlds it needs, not
terms the argument supplies.  The criterion of Section~\ref{sec:criterion} asks what a proof
substitutes for the property and proposition variables of the argument, and the detector counts
exactly that: a term counts if it is substituted for a variable of one of the developments' own
facts --- an axiom, a definition or the equation the unfolder derives from it, a theorem of the
theory under audit --- in any position, or for a variable of any other fact the proof applies ---
a rule of HOL, a lemma of the library, a clause lemma the prover left in the theory --- outside
that fact's schema positions, which the detector reads off the fact's statement: a variable that
is a function over a type variable (the body of \tc{spec}, \tc{allE} and \tc{exI}, the two functions
of \tc{fun\_cong} and of Pure's combination rule, a motive, a continuation), or one of a type
variable's type that occurs in the statement only as a side of an equation (the terms the
rewriter's rules relate).  So the witness of \tc{spec}, \tc{allE} and \tc{exI} counts, the point of
\tc{fun\_cong} counts, and the equational scaffolding of the rewriter does not.  The schema
positions are descended into, so that a fact applied inside them is still read, but are not
counted themselves.  The unfiltered reading is shipped beside the counted one, so
that a reader can see what the machinery contributes and check the boundary.  That boundary is a
judgement and not a theorem: it says which rule positions the criterion asks about, it is drawn
by the same program whose verdicts it conditions, and it has not been validated independently of
that program.  \tc{Th4} of the
mixed-quantifier variant, the proof \tc{metis} found with symmetry as an added premiss, shows the
boundary at work: its own clausification terms arrive under the names of the lemmas that introduced
them --- \tc{Meson.all\_forward}, \tc{Meson.ex\_forward}, \tc{HOL.spec}, the Skolem term under
\tc{Hilbert\_Choice.choice} --- and a detector that reads the name tells them from an instantiation
of a property variable; what flags \tc{Th4} is not its own proof but \tc{L}'s, which it uses.

\subsection{The tally}  Under this reading the 427 facts carry 1\,961 counted instantiations (2\,205
rule--term pairs): 851 witnesses of \tc{allE}, \tc{spec} and \tc{exI}, 364 arguments of the
definitional equations, 98 of the named facts of the developments --- \tc{Ax1Gen} twenty times,
\tc{Ax4} fourteen, \tc{MC} nine, Scott's \tc{A2} seven, and the rest fall to the Hilbert axioms,
comprehension, the quantifier and empty-property lemmas, and the modal schemes of the two test theories ---
528 of the clause lemmas the automated provers leave in the theories, 287 of the lemmas of the
\tc{smt} replay, 75 of local hypotheses, and two of a combinator lemma of \tc{metis}.  An earlier version counted 1\,177 (1\,313 pairs): the first three kinds, under a reading that passed over
every other rule (Appendix~\ref{app:detector}).
Of the 427, 115 carry a flagged witness, 114 of them not certified.  Fourteen carry a Skolem
function of a world.  They go back to one proof text, used in two theories: \tc{L}, ``the God-like property is positive'',
is proved in the possibilist and the mixed-quantifier setting of variant~2 by \tc{smt} with the
default solver, and that proof instantiates a universally quantified property with a Skolem
function applied to a world, $v_{5,3}\ a\ b\ c$.  That term is the prover's name for a property
that varies with the world, which is what the \textsc{rigid} flag is for.  The six results of
each theory whose proofs pass through \tc{L}
(\tc{Th4}, \tc{Th5}, \tc{MC}, \tc{Filter}, \tc{UltraFilter}, \tc{UniqueEss3}) carry the same
witness, since the detector reads the proof of \tc{L} where it is used.  Skolemisation is the
prover's device for eliminating a quantifier; the explicit proofs of the same fourteen
statements, in both systems, substitute the set of positive properties and nothing else.  In the
actualist setting of variant~2, where \tc{L} is proved by \tc{smt (verit)}, the reconstruction
substitutes choice terms of its own, $\mathsf{SOME}\ v.\ \ldots$, under the lemmas it leaves in the
theory, and \tc{L} and what depends on it carry those; so do the \tc{smt} proofs of \tc{Th1} and \tc{Th2},
wherever the \emph{Notes} prove them so, and what runs through them --- \tc{Monotheism}, \tc{Filter}
and \tc{UltraFilter} --- and in \tc{EvilDerivable} \tc{T3} to \tc{T5}: 112 facts in all.  Twelve carry a constant of HOL's library in a witness position
(\textsc{foreign}): the proof by \tc{auto} of the test lemma \tc{EqPrimLeib}, in both test
theories, instantiates the Leibniz quantifier with $x \in \{a\}$, the library's rendering of
equality with $a$, and those two carry nothing else; \tc{UltraFilter} of Scott's variant in the possibilist setting, proved by
\tc{blast}, instantiates a property quantifier with $\mathit{listsp}\ A\ []$, a list predicate that
is true of the empty list and stands for the universal property; and \tc{UniqueEss3}, proved by
\tc{auto} through \tc{MC} in nine theories, leaves a term built from the converse of a relation and
an uninstantiated variable, which denotes nothing the argument uses.  Each of these denotes
something the object language has (an equality, the universal property, an arbitrary property),
and none is a term of it.  The one remaining flag is informational: a witness of \tc{PosProps}
in Scott's variant carries the existence predicate, which is object-definable.  Nothing inherits,
the descent having already attributed \tc{L}'s witness to every proof that runs through it.

\subsection{The cross-check}  A proof written out in Isar carries the author's instantiations in
place of those \tc{smt} or \tc{metis} find, so a flag on an automated step can be answered by
writing the step out.  \tc{MC} written out without automation (\tc{MCbyHand}, session
\tc{MCByHand}) carries two counted instantiations, both object-level.  The 34 theories of the
cross-check (session \tc{IsarAudit}) carry 137 theorems and 344 counted instantiations; two
theorems are flagged, \tc{Th1} and \tc{Th2} of the theory re-run over \textsf{K}, whose
proofs are the \emph{Notes}' own \tc{smt} proofs copied verbatim and carry what those carry.
Six replay single results of the \emph{Notes}; one re-runs a theory of theirs over logic
\textsf{K} with its thirteen proofs copied verbatim; two prove the results of variant~2 in
the possibilist and in the mixed setting; ten prove, one each, the statements whose proofs
substitute a library term (the tenth, Scott's possibilist \tc{UltraFilter}, with four further replays beside it,
\tc{T3}, \tc{MC}, \tc{Monotheism} and \tc{Filter}); and fifteen, added in the present version, prove the 90 further
statements the reading of every rule left uncertified, one theory per theory of the
\emph{Notes} that has them, each citing the axioms, the definitions, the certified facts and earlier replays
only (\tc{Replay2026V1\_\ldots}, \tc{V2\_\ldots}, \tc{V3\_\ldots},
\tc{EvilDerivable}, \tc{Scott\_\ldots}).  The four theorems the whole-term reading alone flags
(Section~\ref{sec:criterion-open}), \tc{Inconsistency} in two theories and Scott's \tc{T3} in two, are
replayed in the theories of their variants as well.  The
account of which postulates a theorem uses (\tc{AxiomDeps.ML}, Tables~\ref{tab:deps} and~\ref{tab:deps-two} and
Section~\ref{sec:discussion}) folds over the recorded proof atoms without reconstructing anything,
and does not depend on the reading.

\subsection{The \lean{} census}  Of the named results of the \emph{Notes}' theories, the only ones
without a counterpart of the same name in the port are two that the port renames and 72 of the 77
that the \emph{Notes} close with \tc{oops}, which abandons the proof and keeps the statement.  The
77 are 26 model probes, all but one of them on the statement \tc{True}, 35 with a Nitpick
countermodel, 10 left open, the five the port proves, and one repetition of a statement the same
theory proves; the statements of questions~1 and~2 are among those left open.  The statements the \emph{Notes} leave unproved are
recorded as placeholders that no theorem uses.  The 252 instantiations (counted per theorem) fall
into 28 shapes: the complement ${\sim}\varphi$ of a property variable (49), $G$ itself (38), the
empty and the universal property $\lambda x.\bot$, $\lambda x.\top$ (21, 19) and their complements
(6, 7), necessary existence $E$ (19), the constant property $\lambda x.\varphi$ of a proposition
variable (12, the witness of modal collapse and of comprehension), self-identity and
self-difference $\lambda x.\ x = x$, $\lambda x.\ x\neq x$ (10 each), and Leibniz identity with a
fixed individual $\lambda z.\ z\equiv x$ (9, the witness of \tc{UniqueEss3}).  An earlier count,
318 in 44 shapes, had taken in the sides of the definitional equations and of rewrites, which are
not instantiations (Appendix~\ref{app:detector}).  The most complex instantiations are the two properties of the inconsistency
proof, ``having the empty essence implies that necessarily some actual thing is empty'' and
$\lambda x.\Box\bot$, and the set $\lambda\psi.\ \psi = G$ handed to \tc{Ax1Gen} in the proof of
\tc{Th4}.  The full list is \tc{census.txt} in the sources.

\subsection{The two runs compared}

The two runs are about different proofs --- the \isa{} run about the originals, the \lean{} run
about proofs written for the purpose --- and each is complete over its own proofs; neither replaces the
other.  Where they differ is in how much each records, 1\,961 counted instantiations against 252,
and the difference is one of degree: the port's proofs apply the source facts directly, while
\isa{}'s automation reaches the same facts through the quantifier rules, whose witnesses are
counted.  The difference tracks how a proof was made, not which system made it.

\section{Verification in detail}
\label{app:verification}

Section~\ref{sec:verification} states what is verified; this appendix says how.  It gives the two
formalisations, what the detector of Section~\ref{sec:criterion} reports when it is turned on this
article's own theories rather than the \emph{Notes}', how the account of which postulates a proof
uses is taken on both sides, how to read the two dependency tables, and where the material is.

\subsection{The two formalisations}
The results are formalised twice.  In \isa{}, following the \emph{Notes}, there are four
theories, one per development, each self-contained and carrying its own copy of the
embedding of Section~\ref{sec:setting}; they are restricted to \textsf{S4} for the three
developments of that section and to \textsf{K} for the mixed-quantifier setting of
Section~\ref{sec:th4k}.  Four more carry the copies of Section~\ref{sec:copies}, each with the
frame of the \emph{Notes}' theory, which the answers there do not use.  A ninth,
\tc{EssenceComparison}, carries the embedding and no axiom of the argument at all, so that the
two essences can be compared where nothing else holds (Lemma~\ref{lem:esscompare}).  It is the
one theory here without a \lean{} counterpart, since what separates the two definitions is a
countermodel and \lean{} has no model finder; no theorem this article states is proved there
alone.  The new proofs are written out in Isar,
with the automation of the \emph{Notes} --- \tc{blast}, \tc{metis}, \tc{smt} --- closing the
individual steps.  Neither choice is cosmetic: keeping their automation makes the new proofs the
same kind of proof as the originals, and the Isar structure keeps each step small enough to be read
off the proof term one at a time.  In \lean{}~\cite{Lean4} there are eight modules, mirroring the
theories item by item, with explicit proof terms.  \lean{}'s \tc{\#print axioms} reports for every
theorem the postulates its proof depends on.  Since the frame conditions are postulated as
named axioms, this certifies the ``no frame condition'' claims above, and the detector of
Section~\ref{sec:criterion} certifies the ``hybrid-free'' ones.  Tables~\ref{tab:deps}
and~\ref{tab:deps-two} list both.

\subsection{The detector over these theories}
The hybrid-free claims are certified on both sides.  On the \isa{} side the nine theories were
rebuilt with \tc{record\_proofs = 2} for the purpose (session \tc{OwnAudit}) and read in full as
Appendix~\ref{app:runs} describes.  Their 206 facts are the axioms, the definitional equations
and the proofs together, as over the \emph{Notes}.  Of them the detector flags eighteen as hybrid
and certifies 188, among them all 75 of the four further theories, the four of
\tc{EssenceComparison}, and six with the informational flag for
the existence predicate.  Seven of their proofs, by \tc{smt} in an earlier version, were
written out when the reading of every rule found the choice terms their reconstruction
substitutes.  The eighteen
are these.  \tc{NoNecExist} and the
thirteen results of variant~3 whose proofs run through it --- \tc{Th3}, \tc{Th5}, \tc{MC},
\tc{MCdual}, \tc{Triv}, \tc{EssDetermined}, \tc{EssMem}, \tc{EssMemI}, \tc{EssSupp},
\tc{Ess\_iff\_Ess2}, \tc{UniqueEss1}, \tc{UniqueEss2}, \tc{UniqueEss3} --- carry its two nominals $\lambda z\,u.\ u = v$
and $\lambda z\,u.\ z = x \wedge u = w$ directly, since the detector reads the proof of
\tc{NoNecExist} where it is used.  \tc{OneWorld} in each of the three \textsf{S4} theories
instantiates the collapse or the essence at a nominal.  And \tc{Th4\_K\_hybrid} carries the
accessibility relation at a fixed world, as it was written to.  The \lean{} detector flags
\tc{NoNecExist}, the three \tc{OneWorld} and \tc{Th4\_K\_hybrid} directly and reports the
thirteen as inheriting; the two sets coincide, and both are what Table~\ref{tab:deps} records.

In figures, over the four theories of the three questions the two
detectors flag the same eighteen results, and each certifies every other fact of those
theories: 109 of them on the \isa{} side, and 175 on the \lean{} one, whose modules carry the
auxiliary steps and the model placeholders besides.  In the four further theories both flag
nothing.  Over the \emph{Notes} the \isa{} run reads the original proofs, certifies 313 of the
427 and names what it does not (Section~\ref{sec:audit}) --- a reading the port, whose proofs
were written for the purpose, cannot give.  Two statements Nitpick leaves undecided --- \tc{L}
and \tc{Th4} for variant~3 under \tc{Ax1GenTwo}, the latter where every world has an actual
individual --- are theorems with nominals, in both systems (\tc{L\_Two\_hybrid},
\tc{Th4\_Two\_actual\_hybrid}, session \tc{F}); and \tc{L\_Box\_inst}, which names a set that
an automation proof leaves implicit and that a \lean{} proof term writes out, is the one named
fact on one side only (Appendix~\ref{sec:accounts}).  The refutations are Nitpick's: 78 model
searches, each certified by the build through an \tc{expect} annotation, 42 of them producing a
genuine model --- 33 a countermodel, seven a witness to the consistency of a setting, two a
model of a side condition.

The sixteen sessions have been read the same way (session \tc{AltAudit}), so that the results
of Sections~\ref{sec:emendation} and~\ref{sec:audit} are certified on the \isa{} side as well
as on the \lean{} one.  Of their 532 facts it flags 154 and calls 85 hybrid, and the flags
fall where those sections say the nominals are.  Among the flagged are \tc{OneWorld}, \tc{R\_Id},
\tc{Ax1Gen\_of\_Two}, \tc{Two\_equivalent} and \tc{only\_top\_and\_empty\_domain\_False},
together with the nominal variants kept beside the hybrid-free ones
(\tc{ae\_Two\_nominal}, \tc{G\_ex\_Two\_unconditional\_nominal}); \tc{ae\_Two},
\tc{G\_ex\_Two\_at}, \tc{G\_ex\_Two\_unconditional} and the rigid set of \tc{L\_R\_rigidpos},
under either spelling of that set, carry only terms of $\mathcal{L}$.  Five of their proofs, by
\tc{smt} in an earlier version (\tc{L\_Box} in three sessions, \tc{T2} in two), were written out
when the reading of every rule found the choice terms their reconstruction substitutes, in them
and in the thirteen results that run through them.
Of the 207 definitions read it flags 50, of four kinds: the frame conditions, which name the
accessibility relation (34); the abbreviation \tc{TwoMembers}, which, when the sixteen are read
in one session, resolves to another session's copy of itself and is therefore reported as foreign
(11; the sessions build separately, as Appendix~\ref{sec:sources} describes); the nominal
\tc{Pt}, ``being $x$ at $w$'', of the one-world arguments of sessions \tc{F}, \tc{H}, \tc{O}
and \tc{P} (4); and \tc{RigidPos} of session \tc{B}, a quantifier over worlds that is no
modality (1).  The abbreviations that state the readings are not among them: they are
validities of terms of $\mathcal{L}$.

Four proofs of these theories are written out in Isar where the \emph{Notes} use automation,
for reasons the run over the \emph{Notes} shows.  The \tc{smt} proof of \tc{L} carries a Skolem
function of a world (Section~\ref{sec:audit}), and copied into three of the four theories of the
three questions, that proof would carry the function under twenty results, \tc{Th4\_K} among them.  The \tc{smt}
proofs of \tc{Th1} and \tc{Th2} leave among their unread sub-derivations the veriT clause with
a Skolem constant for a world described in Appendix~\ref{app:runs}; copied, they would leave
it under every result that uses them --- in variant~3 under \tc{MC\_K} and the uniqueness
results, in the sixteen sessions under Theorem~\ref{thm:plain}(c) and~(d) among others ---
and a reader of the reports would find the mark of an unread clause on results this article
calls hybrid-free.  \tc{L}, \tc{Th1} and \tc{Th2} are therefore written out, in one proof text
each, in every theory and session that states them.  The clause now occurs only under
\tc{OneWorld} of variant~3, whose witness is a nominal in any case, and under the two results
that session \tc{L} derives from that \tc{OneWorld} (\tc{OneWorld\_abs}, \tc{MC\_OneWorld}).  And \tc{UniqueEss3}
of variant~2, left to
\tc{sledgehammer}, runs through \tc{OneWorld} and instantiates \tc{MC} with the nominal
$\lambda u.\ u = w$, where the \lean{} proof instantiates the essence at Leibniz identity with $x$.
Written out with the instantiation the argument makes --- \tc{L} from \tc{Ax1Gen} at the set of
positive properties, \tc{Th1} and \tc{Th2} from \tc{Ax2a}, \tc{Ax2b} and \tc{Ax3} at $\psi$, $G$
and its complement, \tc{UniqueEss3} as in \lean{} --- all four are certified.  Where the
\emph{Notes}' automation is kept it is sometimes handed more than it needs: the \tc{metis} call
that proves Lemma~\ref{lem:l2} is also given \tc{Ax2a}, which that argument does not use, as the
\lean{} module's dependencies show.

\subsection{The dependency and instance accounts}
\label{sec:accounts}
The dependency account is available on both sides as well.  A recorded \isa{} proof term names
every postulate its proof appeals to, so folding over its atoms gives what \tc{\#print axioms}
gives, once the primitive rules of the meta logic, the defining equations and the class arities
are set aside (\tc{AxiomDeps.ML}, run in the same two sessions).  Where both were taken the two
agree on what a proof uses, and every entry of the two tables but one is borne out by the \isa{}
account as well.  The tables report the \lean{} dependencies, as their captions say, and
Appendix~\ref{sec:readtables} gives the one entry where the two systems' proofs differ.

The instance account, too, is available on both sides for the theories of this article: where a
conjunction axiom is applied there, the \isa{} proof names the set, at twenty places; \tc{L} is
proved from \tc{Ax1Gen} at $\Phi:=\Pos$ in all eight theories that state it, in one proof text,
and the remaining twelve places name the set each argument needs.  In the
sixteen sessions, where the reading of the conjunction axiom is a hypothesis and automation
closes the goal, a proof that names the set stands beside it (\tc{L\_Box\_inst} in sessions
\tc{A}, \tc{C} and \tc{D}), the automation proof being kept.  For the
\emph{Notes}' own theories Table~\ref{tab:instances} is the port's reading and can be nothing
else, their proofs being other proofs, found by automation --- which is what the \isa{} run is
there to read.  The same holds of the 49 facts that consume \tc{Ax1Gen}
(Section~\ref{sec:corpus}), where the \isa{} account gives 36, and the two reconcile exactly.
Of the 42 that are results the \emph{Notes} prove, the \isa{} account reaches 33: it does not
reach \tc{Th4}, \tc{Th5} and \tc{MC} in the three quantifier settings of variant~3, which the
\emph{Notes} obtain from their postulated \tc{Th4} and not from the conjunction axiom.  It adds
\tc{UniqueEss3} in the three settings of variant~2, whose proof there runs through \tc{MC}
where the port's instantiates the essence at Leibniz identity, as above --- 33 and three making
36.  The remaining seven of the 49 are not facts of their session at all: the six auxiliaries
the port names, and \tc{NecNoEvil}, which their source closes without a proof.  No statement
is in question on either side; what differs is which axiom a proof of it appeals to.

The dependency account is what settles, for the
\emph{Notes}, that symmetry is the only frame condition their results use, with one exception: it occurs in
61 of them, transitivity only in their test theories, and reflexivity there and in \tc{Th3} of
their \textsf{S4} variant-1 theory, which the \emph{Notes} show inconsistent.

In the sources the suffix \tc{\_K} marks the hybrid-free proofs of
Section~\ref{sec:v3}, which stand beside the hybrid proofs of the same statements rather than
replacing them, so that both are in Table~\ref{tab:deps} and either can be inspected; by
Lemma~\ref{lem:serial} they are frame-free.  Both developments build in seconds (Isabelle2025-2; Lean~4.33.1, no library
beyond the core) and are available with this article.

\subsection{Reading Tables~\ref{tab:deps} and~\ref{tab:deps-two}}
\label{sec:readtables}
Every entry the detector flags is a proof of this article, not a proof the \emph{Notes} give; the one statement among them that the \emph{Notes} prove,
\tc{Th4} of the mixed setting, also has the hybrid-free proof \tc{Th4\_K}.  The claim about the
\emph{Notes} in Section~\ref{sec:audit} is therefore untouched by these flags.  ``Original'' marks
this article's \textsf{S4} proof of a statement, as against its \tc{\_K} proof.  The \textsf{S4}
originals stand in the third block, and \tc{Rrefl} occurs there only; everywhere else
Lemma~\ref{lem:serial} replaces it, and \tc{Rtrans} occurs nowhere.  ``Inherits'' means that
the proof uses a flagged result; where \tc{metis} copied the flagged result's witnesses into the proof
term, the dependent result is itself flagged directly, as \tc{Th3} and the original \tc{MC} of variant~3 are.  Only
for that original \tc{MC} do the two systems differ in what a proof uses: the \isa{} proof runs
through \tc{G\_ex} and so takes \tc{Ax1Gen} besides, where the \lean{} proof runs through
\tc{OneWorld}.  In variant~2, \tc{UniqueEss1} holds hybrid-free with no axiom and
\tc{UniqueEss2} is refuted, as in the \emph{Notes}.

\begin{table}[tp]
\centering\footnotesize\setlength{\tabcolsep}{2pt}
\caption{Postulates used by the \lean{} proofs under \tc{Ax1Gen} as the \emph{Notes} state it, as
  reported by \tc{\#print axioms}, and the verdict of the hybrid-witness detector.  Axiom names
  drop their \tc{Ax} prefix; signature constants, the non-emptiness of $e$ and of $i$ and
  \lean{}'s three standard axioms are omitted.  \tc{Rrefl} is reflexivity.  $^\S$ In variant~3 these are
  \tc{Th3\_K}, \tc{Th5\_K} and \tc{MC\_K} (Section~\ref{sec:v3}).}
\label{tab:deps}
\begin{tabular}{lllll}
\toprule
 & variant 2 & variant 3 & Scott + \tc{1Gen} & hybrid?\\
\midrule
\multicolumn{5}{l}{\emph{The answers to the three open questions}}\\
\tc{G\_ex} & \tc{1Gen 2a 4} & \tc{1Gen 2a 4} & \tc{1Gen A1 A2} & no\\
\tc{Th3}/\tc{T3} & \tc{1Gen 2a 4} & \tc{1Gen 2a 4}$^\S$ & \tc{1Gen A1 A2} & no\\
\tc{Th5} & \tc{1Gen 2a 4} & \tc{1Gen 2a 4}$^\S$ & --- & no\\
\tc{MC} & \tc{1Gen 2a 2b 4} & \tc{1Gen 2a 2b 4}$^\S$ & \tc{1Gen A1 A2 A4} & no\\
\tc{UniqueEss3} & none & --- & --- & no\\
\tc{UniqueEss3\_K} & --- & none & --- & no\\
\tc{Serial} & --- & \tc{2a 3 4} & --- & no\\
\tc{AllExist}, \tc{Reach\_K}, \tc{EssMemI\_K} & --- & \tc{1Gen 2a 3 4} & --- & no\\
\tc{UniqueEss1\_K}, \tc{UniqueEss2\_K} & --- & \tc{1Gen 2a 2b 3 4} & --- & no\\
\midrule
\multicolumn{5}{l}{\emph{Their possibilist and mixed-quantifier copies (Section~\ref{sec:copies})}}\\
\tc{G\_ex}, \tc{Th3} (poss.) & \tc{1Gen 2a 4} & \tc{1Gen 2a 4} & --- & no\\
\tc{G\_exP} (mixed) & --- & \tc{1Gen 2a 4} & --- & no\\
\tc{Reach}, \tc{EssMemI} (poss., mixed) & --- & \tc{1Gen 2a 3 4} & --- & no\\
\tc{UniqueEss1/2} (poss., mixed) & --- & \tc{1Gen 2a 2b 3 4} & --- & no\\
\midrule
\multicolumn{5}{l}{\emph{The collapse results, which need nominals}}\\
\tc{OneWorld} (relative) & \tc{1Gen 2a 2b 4} & --- & \tc{1Gen A1 A2 A4} & \textsc{nominal}\\
\tc{NoNecExist} & --- & \tc{Rrefl} & --- & \textsc{nominal}\\
\tc{Th3} (via \tc{NoNecExist}) & --- & \tc{3 Rrefl} & --- & \textsc{nominal}\\
\tc{OneWorld} (absolute) & --- & \tc{2a 3 4 Rrefl} & --- & \textsc{nominal}\\
\tc{EssDetermined}, \tc{UniqueEss1/2} & --- & \tc{1Gen 2a 3 4 Rrefl} & --- & inherits\\
\tc{Th5} (variant 3, original) & --- & \tc{1Gen 2a 3 4 Rrefl} & --- & inherits\\
\tc{MC} (variant 3, original) & --- & \tc{2a 3 4 Rrefl} & --- & \textsc{rigid}\\
\midrule
\multicolumn{5}{l}{\emph{The \textsf{K}-proofs of Section~\ref{sec:th4k}: variant~2, mixed quantifiers}}\\
\tc{Serial} & \tc{1Gen 2a} & --- & --- & no\\
\tc{Th4\_K\_hybrid} & \tc{1Gen 2a 2b 3 4} & --- & --- & \textsc{access rigid}\\
\tc{Th4\_K} & \tc{1Gen 2a 2b 3 4} & --- & --- & no\\
\bottomrule

\end{tabular}
\end{table}

\begin{table}[tp]
\centering\footnotesize\setlength{\tabcolsep}{2pt}
\caption{The same under \tc{Ax1GenTwo}, with the conventions of Table~\ref{tab:deps}.
  $^\dagger$ marks the one hypothesis that is not an axiom, the side condition of
  Theorem~\ref{thm:v3-two}.  That every world has an actual individual is not a hypothesis
  anywhere here: the variant-2, Scott and variant-3$'$ entries use it, and it is a theorem of the
  axioms each of them lists (\tc{ae\_Two}).}
\label{tab:deps-two}
\begin{tabular}{lllll}
\toprule
 & variant 2 & variant 3 & Scott + \tc{1Gen} & hybrid?\\
\midrule
\multicolumn{5}{l}{\emph{Under \tc{Ax1GenTwo} (Section~\ref{sec:ax1gentwo}): variant~2, variant~3
  as the \emph{Notes} state it, Scott's variant}}\\
\tc{L\_Two} & \tc{Two 2a 3 4} & \tc{Two 2a 3 4 Th4} & \tc{Two A1 A2 A5} & no\\
\tc{G\_ex\_Two}, \tc{Th3}/\tc{T3}, \tc{Th5} & \tc{Two 2a 3 4} & \tc{Two 2a 3 4 Th4}$^\dagger$ & \tc{Two A1 A2 A5} & no\\
\tc{MC\_Two} & \tc{Two 2a 2b 3 4} & \tc{Two 2a 2b 3 4 Th4}$^\dagger$ & \tc{Two A1 A2 A4 A5} & no\\
\tc{AllExist\_Two}, \tc{Reach\_Two} & --- & \tc{Two 2a 3 4 Th4}$^\dagger$ & --- & no\\
\tc{UniqueEss1/2\_Two} & --- & \tc{Two 2a 2b 3 4 Th4}$^\dagger$ & --- & no\\
\tc{G\_ex\_Two\_hybrid} & --- & \tc{2a 3 4 Th4} & --- & inherits\\
\midrule
\multicolumn{5}{l}{\emph{Variant~3$'$ (Section~\ref{sec:plain}): the clause confined to the
  essence, where \tc{Th4} is a theorem}}\\
\tc{PTop} & --- & \tc{3 4} & --- & no\\
\tc{L\_Two}, \tc{Th4\_Two} & --- & \tc{Two 2a 3 4} & --- & no\\
\tc{G\_ex\_Two}, \tc{Th3\_Two}, \tc{Th5\_Two} & --- & \tc{Two 2a 3 4} & --- & no\\
\tc{MC\_Two} & --- & \tc{Two 2a 2b 3 4} & --- & no\\
\tc{AllExist}, \tc{Reach\_K}, \tc{EssMemI\_K} & --- & \tc{Two 2a 3 4} & --- & no\\
\tc{UniqueEss1\_K}, \tc{UniqueEss2\_K} & --- & \tc{Two 2a 2b 3 4} & --- & no\\
\bottomrule

\end{tabular}
\end{table}

\begin{table}[htbp]
\centering\footnotesize\setlength{\tabcolsep}{4pt}
\caption{The sixteen sessions of \tc{alternatives/} and their \lean{} modules.}
\label{tab:sessions}
\begin{tabular}{llll}
\toprule
 & \lean{} module & & \lean{} module\\
\midrule
\tc{A} & \tc{Ax1GenBoxVariant2}   & \tc{I} & \tc{Th4FromConjunction}\\
\tc{B} & \tc{Ax1GenRigidVariant2} & \tc{J} & \tc{Ax1GenTwoScott}\\
\tc{C} & \tc{Ax1GenBoxVariant3}   & \tc{K} & \tc{Ax1GenTwoVariant3Th4}\\
\tc{D} & \tc{Ax1GenBoxScott}      & \tc{L} & \tc{Ax1GenTwoVariant3Plain}\\
\tc{E} & \tc{Ax1GenTwoVariant2}   & \tc{M} & \tc{Ax1GenTwoVariant2Poss}\\
\tc{F} & \tc{Ax1GenTwoVariant3}   & \tc{N} & \tc{Ax1GenTwoVariant2AndersonQuant}\\
\tc{G} & \tc{Ax1GenOneVariant2}   & \tc{O} & \tc{Ax1GenTwoVariant3PlainPoss}\\
\tc{H} & \tc{Ax1GenOneVariant3}   & \tc{P} & \tc{Ax1GenTwoVariant3PlainAndersonQuant}\\
\bottomrule
\end{tabular}
\end{table}

\subsection{Sources}
\label{sec:sources}
The eight \isa{} theories with \lean{} counterparts, \tc{EssenceComparison} beside them, and
everything else named here,
accompany this article as ancillary files:
\begin{center}
\tc{GoedelVariantHOML2inS4oneFile}\quad
\tc{GoedelVariantHOML3inS4oneFile}\\[2pt]
\tc{ScottVariantHOMLAx1GeninS4oneFile}\quad
\tc{GoedelVariantHOML2AndersonQuantTh4inK}\\[2pt]
\tc{GoedelVariantHOML2possInS4oneFile}\quad
\tc{GoedelVariantHOML3possInS4oneFile}\\[2pt]
\tc{GoedelVariantHOML3possUniqueEssOneFile}\\[2pt]
\tc{GoedelVariantHOML3AndersonQuantUniqueEssOneFile}\\[2pt]
\tc{EssenceComparison}
\end{center}
The script that reads Table~\ref{tab:instances} off the \lean{} proofs, whose terms write the
set out where they apply the axiom, is
\tc{tools/ax1gen\_instances.py}, with its full report beside it.  Both detectors are included with their reports: \tc{HybridAudit.ML}, with the \isa{}
session that runs it over the theories of the \emph{Notes};
\tc{HybridAudit.lean}, with its report over the port and over these theories; and the
census of both.  Two further \isa{} sessions carry the runs of Section~\ref{sec:audit} outside the
\emph{Notes}' own session: over the 34 theories of the cross-check, and over
\tc{MCbyHand}, the proof of modal collapse written out here without automation.  The emendations of
Section~\ref{sec:emendation} have their own files.
Sixteen \isa{} sessions (\tc{alternatives/}, \tc{A}--\tc{P}) carry one reading of the conjunction
axiom and one quantifier setting each, as the \tc{README} of that directory sets out.  All 78
Nitpick searches of those sessions carry their outcome as an \tc{expect} annotation: 42 expecting a genuine model,
33 of them a countermodel and nine a satisfying one, 23 expecting none within the given
cardinalities, 13 recording that the search gives up.  The build therefore checks, on every run, that Nitpick still reports the recorded status under
the search configuration the session fixes; the positive derivations there are Isar proofs.  Two
further searches are annotated the same way outside those sessions, the two probes of
\tc{EssenceComparison} that separate the two essences (Lemma~\ref{lem:esscompare}).  The certificate is that build and nothing else: a
countermodel or a satisfying model is reproduced by running these sessions under
Isabelle2025-2 with \tc{isabelle build -D .}, and neither the \lean{} modules, which have no
model finder, nor the reports shipped beside them replace that run.  A statement the
\emph{Notes} refute is recorded in a \lean{} module as an \tc{example} closed by a placeholder,
so that the build reports it as using \tc{sorry}; no theorem of any module depends on it, as
\tc{\#print axioms} confirms.  A \lean{} module accompanies each session, and
Table~\ref{tab:sessions} names them all.  The results this article
states by name are in session \tc{E} (Lemma~\ref{lem:ltwo}, Theorem~\ref{thm:gextwo}), \tc{J}
(Theorem~\ref{thm:scott-two}), \tc{K} (Theorem~\ref{thm:v3-two}), \tc{F}
(Proposition~\ref{prop:v3two}, with the model written out in
\tc{Ax1GenTwoVariant3Model}), \tc{L} (Theorem~\ref{thm:plain}, parts (a)--(c)), and \tc{A},
\tc{C}, \tc{D} for the neighbouring readings --- \tc{A} also for the combination of the two
restrictions and \tc{B} for the rigid one (Section~\ref{sec:neighbours}).
The separating sentence of Example~\ref{ex:separating} is formalised in
\tc{separation/Separation.thy} and \tc{lean/Separation.lean}, a session and a lake target of
their own outside Table~\ref{tab:sessions}, which build in seconds.
The structure-preserving \lean{} port of the complete dataset of the
\emph{Notes}, over which the \lean{} run of Section~\ref{sec:audit} works, is described separately~\cite{Port};
the \emph{Notes}' own \isa{} sources, which the \isa{} run of that section reads, are an entry of
the Archive of Formal Proofs~\cite{NotesAFP}; the run reads the AFP release of 6 February 2026, the one that
accompanies Isabelle2025-2.  Neither is redistributed here, and the \tc{README} of the
supplementary material names the two downloads, the two commands that run the audits over them,
the figures each must yield, and a hash of every source file each run read.

Everything this article achieves in \lean{} it achieves in \isa{} as well.  The converse does not
hold: the proof search that found the \isa{} proofs, and above all the model finder that produced
every refutation in this article, have no counterpart on the other side
(Section~\ref{sec:contributed}).

\section*{Acknowledgments}
I thank Dana Scott, with whom the \emph{Notes} were written, and Bruno Woltzenlogel Paleo, David
Fuenmayor and Cordelia M\"uhlenbeck, my co-authors on earlier work on the argument.  The shallow embedding in
higher-order logic in which the audited proofs are written is older than this application of it, and is not mine
alone: I thank Chad E. Brown, whose encoding of hybrid logic in HOL it descends from; Lawrence
C. Paulson, with whom it was introduced for quantified multimodal logics; Max Wisniewski and
Alexander Steen, who carried it to the quantified nominal case; and Daniel Kirchner, Edward N.
Zalta, Xavier Parent and Leendert van der Torre, my co-authors on later embeddings and on the
LogiKEy framework.  The question of this article arose in discussion at the
conference \emph{Kurt G\"odel, \guilsinglright Gott\guilsinglleft{} und \guilsinglright
Teufel\guilsinglleft} of the Kurt-G\"odel-Forschungsstelle of the Berlin-Brandenburg Academy of
Sciences and Humanities, Berlin, 29--31 July 2026; I thank its organisers, Eva-Maria Engelen and
Alexander Englert, and the audience of my talk there.  I thank Bj\"orn Lindheim, whose autoresearch
harness read an earlier version of this article, found the step of Section~\ref{sec:criterion-open},
and answered three questions of Section~\ref{sec:open}.

An AI coding assistant (Anthropic Claude) was used, under the author's direction and in
continuous interaction, in the formalisation and in the drafting of this article.
Section~\ref{sec:verification} says in what way, and what was done to check its output.
The author is responsible for the content.

\providecommand{\bibfont}{}\renewcommand{\bibfont}{\small}
\setlength{\bibsep}{6pt}
\bibliographystyle{plainnat}
\bibliography{refs}

\begin{thebibliography}{51}
\providecommand{\natexlab}[1]{#1}
\providecommand{\url}[1]{\texttt{#1}}
\expandafter\ifx\csname urlstyle\endcsname\relax
  \providecommand{\doi}[1]{doi: #1}\else
  \providecommand{\doi}{doi: \begingroup \urlstyle{rm}\Url}\fi

\bibitem[Anderson(1990)]{Anderson1990}
C.~Anthony Anderson.
\newblock Some emendations of {G{\"o}del's} ontological proof.
\newblock \emph{Faith and Philosophy}, 7\penalty0 (3):\penalty0 291--303, 1990.
\newblock \doi{10.5840/faithphil19907325}.

\bibitem[Anderson and Gettings(1996)]{AndersonGettings1996}
C.~Anthony Anderson and Michael Gettings.
\newblock {G}{\"o}del's ontological proof revisited.
\newblock In Petr H{\'a}jek, editor, \emph{G{\"o}del '96: Logical Foundations
  of Mathematics, Computer Science and Physics}, volume~6 of \emph{Lecture
  Notes in Logic}, pages 167--172. Springer, Berlin, 1996.
\newblock Reprinted by Cambridge University Press, 2017,
  \doi{10.1017/9781316716939.011}.

\bibitem[Andrews(1972)]{Andrews1972}
Peter~B. Andrews.
\newblock General models and extensionality.
\newblock \emph{Journal of Symbolic Logic}, 37\penalty0 (2):\penalty0 395--397,
  1972.
\newblock \doi{10.2307/2272982}.

\bibitem[Areces and ten Cate(2007)]{ArecesTenCate2007}
Carlos Areces and Balder ten Cate.
\newblock Hybrid logics.
\newblock In Patrick Blackburn, Johan van Benthem, and Frank Wolter, editors,
  \emph{Handbook of Modal Logic}, volume~3 of \emph{Studies in Logic and
  Practical Reasoning}, pages 821--868. Elsevier, 2007.
\newblock \doi{10.1016/S1570-2464(07)80017-6}.

\bibitem[Benzm{\"u}ller(2019)]{J41}
Christoph Benzm{\"u}ller.
\newblock Universal (meta-)logical reasoning: Recent successes.
\newblock \emph{Science of Computer Programming}, 172:\penalty0 48--62, 2019.
\newblock \doi{10.1016/j.scico.2018.10.008}.

\bibitem[Benzm{\"u}ller(2026{\natexlab{a}})]{Benzmueller2026comment}
Christoph Benzm{\"u}ller.
\newblock A comment on modal collapse and ultrafilters in {G}{\"o}del's
  ontological argument.
\newblock arXiv:2608.07578, 2026{\natexlab{a}}.
\newblock \url{https://arxiv.org/abs/2608.07578}.

\bibitem[Benzm{\"u}ller(2026{\natexlab{b}})]{Port}
Christoph Benzm{\"u}ller.
\newblock G{\"o}del's and {Scott's} variants of the ontological argument in
  {Lean~4} and {TPTP THF}.
\newblock arXiv:2609.26806, 2026{\natexlab{b}}.

\bibitem[Benzm{\"u}ller and Fuenmayor(2020)]{BenzmuellerFuenmayor2020}
Christoph Benzm{\"u}ller and David Fuenmayor.
\newblock Computer-supported analysis of positive properties, ultrafilters and
  modal collapse in variants of {G}{\"o}del's ontological argument.
\newblock \emph{Bulletin of the Section of Logic}, 49\penalty0 (2):\penalty0
  127--148, 2020.
\newblock \doi{10.18778/0138-0680.2020.08}.

\bibitem[Benzm{\"u}ller and Paulson(2008)]{B9}
Christoph Benzm{\"u}ller and Lawrence~C. Paulson.
\newblock Exploring properties of normal multimodal logics in simple type
  theory with {LEO-II}.
\newblock In Christoph Benzm{\"u}ller, Chad~E. Brown, J{\"o}rg Siekmann, and
  Richard Statman, editors, \emph{Reasoning in Simple Type Theory: Festschrift
  in Honor of {Peter B. Andrews} on His 70th Birthday}, Studies in Logic,
  Mathematical Logic and Foundations, pages 386--406. College Publications,
  2008.

\bibitem[Benzm{\"u}ller and Paulson(2009)]{R45}
Christoph Benzm{\"u}ller and Lawrence~C. Paulson.
\newblock Quantified multimodal logics in simple type theory.
\newblock SEKI Report SR-2009-02, Saarland University, 2009.
\newblock arXiv:0905.2435.

\bibitem[Benzm{\"u}ller and Paulson(2013)]{J26}
Christoph Benzm{\"u}ller and Lawrence~C. Paulson.
\newblock Quantified multimodal logics in simple type theory.
\newblock \emph{Logica Universalis}, 7\penalty0 (1):\penalty0 7--20, 2013.
\newblock \doi{10.1007/s11787-012-0052-y}.

\bibitem[Benzm{\"u}ller and Scott(2025{\natexlab{a}})]{J75}
Christoph Benzm{\"u}ller and Dana Scott.
\newblock Notes on {G{\"o}del's} and {Scott's} variants of the ontological
  argument.
\newblock \emph{Monatshefte f{\"u}r Mathematik}, 208\penalty0 (4):\penalty0
  569--611, 2025{\natexlab{a}}.
\newblock \doi{10.1007/s00605-025-02078-x}.

\bibitem[Benzm{\"u}ller and Scott(2025{\natexlab{b}})]{NotesAFP}
Christoph Benzm{\"u}ller and Dana~S. Scott.
\newblock Notes on {G}{\"o}del's and {S}cott's variants of the ontological
  argument ({Isabelle/HOL} dataset).
\newblock \emph{Archive of Formal Proofs}, 2025{\natexlab{b}}.
\newblock ISSN 2150-914x.
\newblock
  \url{https://isa-afp.org/entries/Notes_On_Goedels_Ontological_Argument.html},
  Formal proof development.

\bibitem[Benzm{\"u}ller and Woltzenlogel~Paleo(2014)]{ECAI2014}
Christoph Benzm{\"u}ller and Bruno Woltzenlogel~Paleo.
\newblock Automating {G}{\"o}del's ontological proof of {G}od's existence with
  higher-order automated theorem provers.
\newblock In Torsten Schaub, Gerhard Friedrich, and Barry O'Sullivan, editors,
  \emph{ECAI 2014 -- 21st European Conference on Artificial Intelligence},
  volume 263 of \emph{Frontiers in Artificial Intelligence and Applications},
  pages 93--98. IOS Press, 2014.
\newblock \doi{10.3233/978-1-61499-419-0-93}.

\bibitem[Benzm{\"u}ller and Woltzenlogel~Paleo(2016)]{IJCAI2016}
Christoph Benzm{\"u}ller and Bruno Woltzenlogel~Paleo.
\newblock The inconsistency in {G}{\"o}del's ontological argument: A success
  story for {AI} in metaphysics.
\newblock In Subbarao Kambhampati, editor, \emph{Proceedings of the
  Twenty-Fifth International Joint Conference on Artificial Intelligence (IJCAI
  2016)}, pages 936--942. AAAI Press, 2016.
\newblock URL \url{https://www.ijcai.org/Proceedings/16/Papers/137.pdf}.

\bibitem[Benzm{\"u}ller et~al.(2004)Benzm{\"u}ller, Brown, and Kohlhase]{J7}
Christoph Benzm{\"u}ller, Chad~E. Brown, and Michael Kohlhase.
\newblock Higher-order semantics and extensionality.
\newblock \emph{Journal of Symbolic Logic}, 69\penalty0 (4):\penalty0
  1027--1088, 2004.
\newblock \doi{10.2178/jsl/1102022211}.

\bibitem[Benzm{\"u}ller et~al.(2020)Benzm{\"u}ller, Parent, and van~der
  Torre]{J48}
Christoph Benzm{\"u}ller, Xavier Parent, and Leendert van~der Torre.
\newblock Designing normative theories for ethical and legal reasoning:
  {LogiKEy} framework, methodology, and tool support.
\newblock \emph{Artificial Intelligence}, 287:\penalty0 103348, 2020.
\newblock \doi{10.1016/j.artint.2020.103348}.

\bibitem[Benzm{\"u}ller et~al.(2024)Benzm{\"u}ller, Fuenmayor, Steen, and
  Sutcliffe]{WhoFinds}
Christoph Benzm{\"u}ller, David Fuenmayor, Alexander Steen, and Geoff
  Sutcliffe.
\newblock Who finds the short proof?
\newblock \emph{Logic Journal of the {IGPL}}, 32\penalty0 (3):\penalty0
  442--464, 2024.
\newblock \doi{10.1093/jigpal/jzac082}.

\bibitem[Blackburn(2000)]{Blackburn2000}
Patrick Blackburn.
\newblock Representation, reasoning, and relational structures: A hybrid logic
  manifesto.
\newblock \emph{Logic Journal of the {IGPL}}, 8\penalty0 (3):\penalty0
  339--365, 2000.
\newblock \doi{10.1093/jigpal/8.3.339}.

\bibitem[Blanchette and Nipkow(2010)]{Nitpick}
Jasmin~Christian Blanchette and Tobias Nipkow.
\newblock Nitpick: A counterexample generator for higher-order logic based on a
  relational model finder.
\newblock In Matt Kaufmann and Lawrence~C. Paulson, editors, \emph{Interactive
  Theorem Proving (ITP 2010)}, volume 6172 of \emph{LNCS}, pages 131--146.
  Springer, 2010.
\newblock \doi{10.1007/978-3-642-14052-5_11}.

\bibitem[Boolos(1987)]{Boolos1987}
George Boolos.
\newblock A curious inference.
\newblock \emph{Journal of Philosophical Logic}, 16\penalty0 (1):\penalty0
  1--12, 1987.
\newblock \doi{10.1007/BF00250612}.

\bibitem[Brown(2005)]{Brown2005}
Chad~E. Brown.
\newblock Encoding hybrid logic into higher-order logic.
\newblock Slides of an invited talk at LORIA, Nancy; archived at
  \url{https://web.archive.org/web/20081123061324/http://mathgate.info/cebrown/papers/hybrid-hol.pdf},
  accessed 14 October 2026, April 2005.

\bibitem[Bull(1969)]{Bull1969}
R.~A. Bull.
\newblock On modal logic with propositional quantifiers.
\newblock \emph{Journal of Symbolic Logic}, 34\penalty0 (2):\penalty0 257--263,
  1969.
\newblock \doi{10.2307/2271103}.

\bibitem[de~Moura and Ullrich(2021)]{Lean4}
Leonardo de~Moura and Sebastian Ullrich.
\newblock The {Lean} 4 theorem prover and programming language.
\newblock In Andr{\'e} Platzer and Geoff Sutcliffe, editors, \emph{Automated
  Deduction -- CADE 28}, volume 12699 of \emph{LNCS}, pages 625--635. Springer,
  2021.
\newblock \doi{10.1007/978-3-030-79876-5_37}.

\bibitem[de~Rijke(1992)]{deRijke1992}
Maarten de~Rijke.
\newblock The modal logic of inequality.
\newblock \emph{Journal of Symbolic Logic}, 57\penalty0 (2):\penalty0 566--584,
  1992.
\newblock \doi{10.2307/2275293}.

\bibitem[Fine(1970)]{Fine1970}
Kit Fine.
\newblock Propositional quantifiers in modal logic.
\newblock \emph{Theoria}, 36\penalty0 (3):\penalty0 336--346, 1970.
\newblock \doi{10.1111/j.1755-2567.1970.tb00432.x}.

\bibitem[Fitting(2002)]{Fitting2002}
Melvin Fitting.
\newblock \emph{Types, Tableaus, and {G}{\"o}del's {G}od}, volume~12 of
  \emph{Trends in Logic}.
\newblock Kluwer, 2002.
\newblock \doi{10.1007/978-94-010-0411-4}.

\bibitem[Fritz(2024)]{Fritz2024}
Peter Fritz.
\newblock \emph{Propositional Quantifiers}.
\newblock Elements in Philosophy and Logic. Cambridge University Press, 2024.
\newblock \doi{10.1017/9781009177740}.

\bibitem[G{\"o}del(2004)]{GoedelNotes}
Kurt G{\"o}del.
\newblock Appendix {A}: Notes in {Kurt G{\"o}del's} hand.
\newblock In Jordan~H. Sobel, editor, \emph{Logic and Theism: Arguments for and
  Against Beliefs in God}, pages 144--145. Cambridge University Press, 2004.
\newblock \doi{10.1017/CBO9780511497988}.
\newblock Notes dated 1970.

\bibitem[Goranko and Passy(1992)]{GorankoPassy1992}
Valentin Goranko and Solomon Passy.
\newblock Using the universal modality: Gains and questions.
\newblock \emph{Journal of Logic and Computation}, 2\penalty0 (1):\penalty0
  5--30, 1992.
\newblock \doi{10.1093/logcom/2.1.5}.

\bibitem[H{\'a}jek(2002)]{Hajek2002}
Petr H{\'a}jek.
\newblock A new small emendation of {G}{\"o}del's ontological proof.
\newblock \emph{Studia Logica}, 71\penalty0 (2):\penalty0 149--164, 2002.
\newblock \doi{10.1023/A:1016583920890}.

\bibitem[Henkin(1950)]{Henkin1950}
Leon Henkin.
\newblock Completeness in the theory of types.
\newblock \emph{Journal of Symbolic Logic}, 15\penalty0 (2):\penalty0 81--91,
  1950.
\newblock \doi{10.2307/2266967}.

\bibitem[Holliday(2019)]{Holliday2019}
Wesley~H. Holliday.
\newblock A note on algebraic semantics for {S5} with propositional
  quantifiers.
\newblock \emph{Notre Dame Journal of Formal Logic}, 60\penalty0 (2):\penalty0
  311--332, 2019.
\newblock \doi{10.1215/00294527-2019-0001}.

\bibitem[Jech(2003)]{Jech2003}
Thomas Jech.
\newblock \emph{Set Theory}.
\newblock Springer Monographs in Mathematics. Springer, Berlin, {Third
  Millennium} edition, 2003.
\newblock \doi{10.1007/3-540-44761-X}.

\bibitem[Kanckos and Lethen(2021)]{KanckosLethen2021}
Annika Kanckos and Tim Lethen.
\newblock The development of {G}{\"o}del's ontological proof.
\newblock \emph{The Review of Symbolic Logic}, 14\penalty0 (4):\penalty0
  1011--1029, 2021.
\newblock \doi{10.1017/S1755020319000479}.

\bibitem[Kanckos and Woltzenlogel~Paleo(2017)]{KanckosPaleo2017}
Annika Kanckos and Bruno Woltzenlogel~Paleo.
\newblock Variants of {G}{\"o}del's ontological proof in a natural deduction
  calculus.
\newblock \emph{Studia Logica}, 105\penalty0 (3):\penalty0 553--586, 2017.
\newblock \doi{10.1007/s11225-016-9700-1}.

\bibitem[Kirchner et~al.(2019)Kirchner, Benzm{\"u}ller, and Zalta]{J45}
Daniel Kirchner, Christoph Benzm{\"u}ller, and Edward~N. Zalta.
\newblock Computer science and metaphysics: A cross-fertilization.
\newblock \emph{Open Philosophy}, 2\penalty0 (1):\penalty0 230--251, 2019.
\newblock \doi{10.1515/opphil-2019-0015}.

\bibitem[Kirchner et~al.(2020)Kirchner, Benzm{\"u}ller, and Zalta]{J50}
Daniel Kirchner, Christoph Benzm{\"u}ller, and Edward~N. Zalta.
\newblock Mechanizing {Principia Logico-Metaphysica} in functional type theory.
\newblock \emph{Review of Symbolic Logic}, 13\penalty0 (1):\penalty0 206--218,
  2020.
\newblock \doi{10.1017/S1755020319000297}.

\bibitem[Kova{\v c}(2012)]{Kovac2012}
Sre{\'c}ko Kova{\v c}.
\newblock Modal collapse in {G}{\"o}del's ontological proof.
\newblock In Miros{\l}aw Szatkowski, editor, \emph{Ontological Proofs Today},
  pages 323--344. Ontos Verlag, 2012.
\newblock \doi{10.1515/9783110325881.323}.

\bibitem[Lindheim(2026{\natexlab{a}})]{Lindheim2026a}
Bj{\"o}rn Lindheim.
\newblock A complete calculus for the {$\mathcal{L}$}-general models of the
  ontological argument.
\newblock Zenodo, version 1, 2026{\natexlab{a}}.
\newblock \textsc{doi}: 10.5281/zenodo.23795768.

\bibitem[Lindheim(2026{\natexlab{b}})]{Lindheim2026b}
Bj{\"o}rn Lindheim.
\newblock Nominal-free proofs of the theorems {Benzm{\"u}ller}'s detector
  flags.
\newblock Zenodo, version 1, 2026{\natexlab{b}}.
\newblock \textsc{doi}: 10.5281/zenodo.23795772.

\bibitem[Lindheim(2026{\natexlab{c}})]{Lindheim2026c}
Bj{\"o}rn Lindheim.
\newblock G{\"o}del's lemma {L} under the rigid reading, with no condition on
  positivity.
\newblock Zenodo, version 1, 2026{\natexlab{c}}.
\newblock \textsc{doi}: 10.5281/zenodo.23795777.

\bibitem[Lindheim(2026{\natexlab{d}})]{Lindheim2026d}
Bj{\"o}rn Lindheim.
\newblock From {Lean} proofs to certified derivations in the ontological
  embedding.
\newblock Zenodo, version 1, 2026{\natexlab{d}}.
\newblock \textsc{doi}: 10.5281/zenodo.23795779.

\bibitem[M{\"u}hlenbeck and Benzm{\"u}ller(2026)]{MuehlenbeckBenzmueller2026}
Cordelia M{\"u}hlenbeck and Christoph Benzm{\"u}ller.
\newblock On the maximality of positive properties and modal collapse in
  variants of {G}{\"o}del's ontological proof for the existence of {God}.
\newblock \emph{Logic and Logical Philosophy}, 2026.
\newblock \doi{10.12775/LLP.2026.007}.
\newblock Online first, pp.~1--21.

\bibitem[Nipkow et~al.(2002)Nipkow, Paulson, and Wenzel]{IsabelleHOL}
Tobias Nipkow, Lawrence~C. Paulson, and Markus Wenzel.
\newblock \emph{Isabelle/HOL: A Proof Assistant for Higher-Order Logic}, volume
  2283 of \emph{LNCS}.
\newblock Springer, 2002.
\newblock \doi{10.1007/3-540-45949-9}.

\bibitem[Pitts(2013)]{Pitts2013}
Andrew~M. Pitts.
\newblock \emph{Nominal Sets: Names and Symmetry in Computer Science},
  volume~57 of \emph{Cambridge Tracts in Theoretical Computer Science}.
\newblock Cambridge University Press, 2013.
\newblock \doi{10.1017/CBO9781139084673}.

\bibitem[Scott(2004)]{ScottNotes}
Dana Scott.
\newblock Appendix {B}: Notes in {Dana Scott's} hand.
\newblock In Jordan~H. Sobel, editor, \emph{Logic and Theism: Arguments for and
  Against Beliefs in God}, pages 145--146. Cambridge University Press, 2004.
\newblock \doi{10.1017/CBO9780511497988}.
\newblock Notes dated 1972.

\bibitem[Sobel(1987)]{Sobel1987}
Jordan~H. Sobel.
\newblock {G}{\"o}del's ontological proof.
\newblock In Judith~Jarvis Thomson, editor, \emph{On Being and Saying: Essays
  for Richard Cartwright}, pages 241--261. MIT Press, 1987.

\bibitem[ten Cate(2006)]{tenCate2006}
Balder ten Cate.
\newblock Expressivity of second order propositional modal logic.
\newblock \emph{Journal of Philosophical Logic}, 35\penalty0 (2):\penalty0
  209--223, 2006.
\newblock \doi{10.1007/s10992-005-9012-9}.

\bibitem[Wisniewski and Steen(2015)]{WisniewskiSteen2014}
Max Wisniewski and Alexander Steen.
\newblock Embedding of quantified higher-order nominal modal logic into
  classical higher-order logic.
\newblock In Christoph Benzm{\"u}ller and Jens Otten, editors, \emph{ARQNL
  2014. Automated Reasoning in Quantified Non-Classical Logics}, volume~33 of
  \emph{EPiC Series in Computing}, pages 59--64. EasyChair, 2015.
\newblock \doi{10.29007/dzc2}.

\bibitem[Wisniewski et~al.(2016)Wisniewski, Steen, and Benzm{\"u}ller]{W56}
Max Wisniewski, Alexander Steen, and Christoph Benzm{\"u}ller.
\newblock {TPTP} and beyond: Representation of quantified non-classical logics.
\newblock In Christoph Benzm{\"u}ller and Jens Otten, editors, \emph{ARQNL
  2016. Automated Reasoning in Quantified Non-Classical Logics}, volume 1770 of
  \emph{CEUR Workshop Proceedings}, pages 51--65. CEUR-WS.org, 2016.

\end{thebibliography}

\end{document}